\documentclass[conference]{IEEEtran}

\usepackage[english]{babel}

\usepackage{booktabs}
\usepackage{multirow}
\usepackage{url}
\usepackage{cite}
\usepackage{amsmath}
\usepackage{amssymb}
\usepackage{verbatim}
\usepackage{graphicx}
\usepackage{rotating}
\usepackage{microtype}
\usepackage{color}

\usepackage[bm,compactparagraph,subfigurepackage,recurtheorems]{pauli_new}
\renewcommand{\paragraph}[1]{\textbf{#1.}}

\usepackage[ruled,vlined,boxed,linesnumbered]{algorithm2e}

\newcommand\bin{\operatorname{bin}}
\newcommand\sign{\operatorname{sign}}
\newcommand\rrank{\operatorname{rrank}}
\newcommand\rrankp{\operatorname{rrank}^+}
\newcommand\srank{\operatorname{srank}}
\newcommand\round{\operatorname{round}}

\newcommand{\E}{\mathbb{E}} 

\newcommand{\minErrorRRproblemFullname}[1]{minimum-error rounding rank-#1\xspace}
\newcommand{\minErrorRRproblem}[1]{\textsc{MinErrorRR}-#1\xspace}
\newcommand{\BMNA}{\textsc{BMNA}\xspace}

\newcommand{\eat}[1]{}

\newcommand{\abstracts}{\datasetname{Abstracts}\xspace}
\newcommand{\apj}{\datasetname{APJ}\xspace}
\newcommand{\dblp}{\datasetname{DBLP}\xspace}
\newcommand{\dialect}{\datasetname{Dialect}\xspace}
\newcommand{\now}{\datasetname{Paleo}\xspace}

\newcommand{\proje}{\algname{ProjLP}\xspace}
\newcommand{\lpca}{\algname{L-PCA}\xspace}
\newcommand{\nuclear}{\algname{Nuclear}\xspace}
\newcommand{\SVD}{\algname{R-SVD}\xspace}
\newcommand{\truncSVD}{\algname{T-SVD}\xspace}
\newcommand{\asso}{\algname{Asso}\xspace}
\newcommand{\shay}{\algname{Permutation}\xspace} 
\newcommand{\nexhaust}{\algname{NNRR1}\xspace} 
\newcommand{\nmt}{\algname{MT}\xspace} 
\newcommand{\nmtexhaust}{\algname{MT+NNR1}\xspace} 
\newcommand{\lb}{\algname{Spectral~LB}\xspace} 

\graphicspath{{figs/}}
\newlength{\smallfigwidth}
\newlength{\smallfigsep}
\newlength{\legendheight}
\title{What You Will Gain By Rounding:\\Theory and Algorithms for Rounding Rank}

\author{
\IEEEauthorblockN{Stefan Neumann\IEEEauthorrefmark{1}}
\IEEEauthorblockA{University of Vienna\\
Vienna, Austria\\
stefan.neumann@univie.ac.at}
\and
\IEEEauthorblockN{Rainer Gemulla}
\IEEEauthorblockA{Data and Web Science Group \\
University of Mannheim, Germany \\
rgemulla@uni-mannheim.de}
\and
\IEEEauthorblockN{Pauli Miettinen}
\IEEEauthorblockA{Max Planck Institute for Informatics \\
Saarland Informatics Campus, Germany \\
pmiettin@mpi-inf.mpg.de}
}

\date{}

\usepackage{color}
\usepackage{tikz}

\newif\ifshowcomments
\showcommentstrue
\ifshowcomments
\newcommand\todo[1]{\textcolor{red}{#1}}
\newcommand{\mynote}[2]{\framebox{\bfseries\sffamily\scriptsize{#1}}
 {\small\textsf{\emph{#2}}}}
\else
\newcommand{\todo}[1]{}
\newcommand{\mynote}[2]{}
\fi

\newcommand{\sn}[1]{\textcolor{blue!50!black}{\mynote{SN}{#1}}}

\begin{document}
\maketitle

\begin{abstract}
When factorizing binary matrices, we often have to make a choice between using expensive combinatorial methods that retain the discrete nature of the data and using continuous methods that can be more efficient but destroy the discrete structure. Alternatively, we can first compute a continuous factorization and subsequently apply a rounding procedure to obtain a discrete representation. But what will we gain by rounding? Will this yield lower reconstruction errors? Is it easy to find a low-rank matrix that rounds to a given binary matrix? Does it matter which threshold we use for rounding? Does it matter if we allow for only non-negative factorizations? In this paper, we approach these and further questions by presenting and studying the concept of rounding rank. We show that rounding rank is related to linear classification, dimensionality reduction, and nested matrices. We also report on an extensive experimental study that compares different algorithms for finding good factorizations under the rounding rank model.


\end{abstract}


\renewcommand{\thefootnote}{\fnsymbol{footnote}}
\footnotetext[1]{Part of the work was done when the author was with MPI Informatics, Saarland University and the Saarbr\"ucken Graduate School of Computer Science.}
\renewcommand{\thefootnote}{\arabic{footnote}}

\section{Introduction}
\label{sec:introduction}

When facing data that can be expressed as a binary matrix, the data analyst
usually has two options: either she uses combinatorial methods---such as
frequent itemset mining or various graph algorithms---that will retain the
binary structure of the data, or she applies some sort of continuous-valued
matrix factorization---such as SVD or NMF---that will represent the binary
structure with continuous approximations. The different approaches come with
different advantages and drawbacks. Retaining the combinatorial structure is
helpful for interpreting the results and can preserve better other
characteristics such as sparsity. Continuous methods, on the other hand, are
often more efficient, yield better reconstruction errors, and may be interpreted
probabilistically.

A third alternative, often applied to get ``the best of both worlds,'' is to
perform a continuous factorization first and apply some function to the elements
of the reconstructed matrix to make them binary afterwards. In probabilistic
modelling, for example, the logistic function is commonly used to map real
values into the unit range. We can obtain a binary reconstruction by
rounding, i.e.\ by setting all values less than $1/2$ to $0$ and the remaining
values to $1$. Alternatively, for $\{-1,1\}$ matrices, we may take the sign of
the values of a continuous factorization to obtain a discrete representation.
Even though such methods are commonly used, relatively little is known about the consequences of this thresholding process. 
There are few, if any, methods that
aim at finding a matrix that rounds exactly to the given binary data, or finding
a low-rank matrix that causes only little error when rounded (although 
there are methods that have such a behavior as a by-product). Almost
nothing is known about the theoretical properties of such decompositions.

In this paper, we give a comprehensive treatise of these topics.  We introduce
the concept of \emph{rounding rank}, which, informally, is defined to be the least
rank of a real matrix that rounds to the given binary matrix. But does it matter
how we do the rounding? How will the results change if we constrain ourselves to
nonnegative factorizations? A solid theoretical understanding of the properties
of rounding rank will help data miners and method developers to understand what
happens when they apply rounding. 
Some of our results are novel, while others are based on results obtained from
related topics such as \emph{sign rank} and \emph{dot product graphs}.

Studying rounding rank is not only of theoretical interest.  The concept can
provide new insight or points of view for existing problems, and lead to
interesting new approaches. In essence, rounding rank provides another
\emph{intrinsic dimensionality} of the data (see,
e.g.~\cite{tatti06what}). Rounding rank can be used, for example, to determine
the minimum number of features linear classifiers need for multi-label
classification 
or the minimum number of
dimensions we need from a dimensionality reduction algorithm.
There is also a close relationship to \emph{nested matrices}~\cite{junttila11},
a particular type of binary matrices that occur, for example, in ecology. 
We show that nested matrices are equivalent to matrices with a non-negative
rounding rank of 1 and use this property to develop a new algorithm for the
problem of finding the closest nested matrix.

But just knowing about the properties of rounding rank will not help if we
cannot find good decompositions. As data miners have encountered problems
related to rounding rank earlier, there are already existing algorithms for
closely related problems. In fact, any low-rank matrix factorization algorithm
could be used for estimating (or, more precisely, bounding) the rounding rank,
but not all of them would work equally well. To that end, we survey a number of
algorithms for estimating the rounding rank and for finding the least-error
fixed rounding rank decomposition. We also present some novel methods.
One major contribution of this paper is an empirical evaluation of these
algorithms.
Our experiments aim to help the practitioners in choosing the correct algorithm
for the correct task: for example, if one wants to estimate the rounding rank of
a binary matrix, simply rounding the truncated singular value decomposition may
not be a good idea.





\section{Definitions, Background, and Theory}
\label{sec:theory}

In this section we formally define the rounding rank of a binary matrix, discuss
its properties, and compare it to other well-known
matrix-ranks. Throughout this paper, we use $\mB$ to denote a binary $m \times
n$ matrix.

\subsection{Definitions}

The \emph{rounding function} w.r.t.\ \emph{rounding threshold} $\tau \in \mathbb{R}$ 
is 
\begin{align*}
  \round_\tau(x) = 
  \begin{cases}
    1, & \text{if } x \geq \tau, \\
    0, & \text{if } x < \tau.
  \end{cases}
\end{align*}
We apply $\round_\tau$ to matrices by rounding element-wise, i.e.\ if $\mA
\in \R^{m \times n}$ is a real-valued matrix, then $\round_\tau(\mA)$ denotes an
$m \times n$ binary matrix with $[\round_\tau(\mA)]_{ij} =
\round_\tau(\mA_{ij})$.

\paragraph{Rounding rank}
Given a rounding threshold $\tau \in \mathbb{R}$, the \emph{rounding rank of
  $\mB$ w.r.t.\ $\tau$} is given by
\begin{equation}
  \label{eq:rrank}
  \rrank_\tau(\mB) = \min\{ \rank(\mA) : \mA \in \mathbb{R}^{m \times n}, \round_\tau(\mA) = \mB \}.
\end{equation}
The rounding rank of $\mB$ is thus the smallest rank of any real-valued matrix
that rounds to $\mB$. We often omit 
$\tau$ for brevity and write $\round(\mA)$ and $\rrank(\mB)$ for
$\tau=1/2$.

When $\mB$ has rounding rank $k$, there exists matrices $\mL \in \mathbb{R}^{m
  \times k}$ and $\mR \in \mathbb{R}^{n \times k}$ with $\mB = \round_\tau(\mL
\mR^T)$. We refer to $\mL$ and $\mR$ as a \emph{rounding rank decomposition} of
$\mB$.

\paragraph{Sign rank}
The \emph{sign matrix} of $\mB$,  $\mB^\pm \in \{-1,+1\}^{m \times n}$, is obtained from $\mB$ 
by replacing every $0$ by $-1$. Given
a sign matrix, its \emph{sign rank} is given by
\begin{equation}
  \label{eq:srank}
  \srank(\mB^\pm) = \min\{ \rank(\mA): \mA \in \mathbb{R}_{\neq 0}^{m \times n}, \sign(\mA) = \mB \},
\end{equation}
where $\mathbb{R}_{\neq 0} = \mathbb{R} \setminus \{0 \}$.  The sign rank
is thus the smallest rank of any real-valued matrix $\mA$ without $0$-entries
and with $\mB^\pm_{ij} = \sign(\mA_{ij})$ for all $i,j$. The sign rank is
closely related to the rounding rank as
\(
	\rrank_0(\mB) \leq \srank(\mB^\pm) \leq \rrank_0(\mB) + 1.
\)
The first inequality holds because for any $\mA\in \mathbb{R}_{\neq 0}^{m \times n}$ and with
$\sign(\mA) = \mB^\pm$, $\round_0(\mA)^\pm = \sign(\mA)$.  The second inequality
holds because if $\round(\mA) = \mB$ and $\mA$ contains $0$-entries, we can add
a constant $0<\varepsilon<\min_{a_{ij}<0}\lvert a_{ij}\rvert$ to
each entry of $\mA$ to obtain $\sign(\mA + \varepsilon) = \mB^\pm$ and
$\rank(\mA + \varepsilon) \leq \rank(\mA) + 1$.  Even when $\tau\neq 0$, the
differences remain small as suggested by
Prop.~\ref{prop:changingRoundingThreshold}.

\paragraph{Non-negative rounding rank}
We define the
\emph{non-negative rounding rank} of $\mB$ w.r.t.~$\tau$, denoted
$\rrankp_\tau(\mB)$, as the smallest $k$ such that there exist non-negative
matrices $\mL \in \nR^{m \times k}$ and $\mR \in \nR^{n \times k}$ with
$\round_\tau(\mL \mR^T) = \mB$.

\paragraph{Minimum-error rounding rank problem}
The rounding rank is concerned with exact reconstructions of $\mB$.  We relax
this by introducing the \emph{\minErrorRRproblemFullname{$k$} problem}: Find a
binary matrix $\mC \in \{0,1\}^{m \times n}$ with $\rrank(\mC) \leq k$ which
minimizes $\norm{\mB - \mC}_F$, where $\norm{\cdot}_F$ denotes the Frobenius
norm. Note that $\norm{\mB - \mC}_F^2$ corresponds to the number of entries in
which $\mB$ and $\mC$ disagree. We denote the problem by
\minErrorRRproblem{$k$}.



\subsection{Related Work}
\label{subsec:related-work}

A number of concepts closely related to rounding rank (albeit less general) have
been studied in various communities.

There is a relationship between rounding rank and dot-product
graphs~\cite{Reiterman92Embedding,Fiduccia98Dot,SphereAndDotProduct}, which
arise in social network analysis~\cite{Young2007}. Let $G$ be a graph with $n$
vertices and adjacency matrix $\mM$. Then $G$ is a \emph{dot-product graph of
  rank~$k$} if there exists a matrix $\mL \in \mathbb{R}^{m \times k}$ such that $\mM
= \round_1(\mL \mL^T)$.
The rank of a dot-product graph
corresponds to the \emph{symmetric} rounding rank of its adjacency matrix. In
this paper, we consider asymmetric factorizations and allow for rectangular
matrices.

Sign rank was studied in the communication complexity community in order to
characterize a certain communication model.  Consider two players, Alice and
Bob.  Alice and Bob obtain private inputs $x, y \in \{0,1\}^n$, respectively, and their
task is to evaluate a function $f : \{0,1\}^n \times \{0,1\}^n \to \{0,1\}$ on
their inputs. The \emph{communication matrix} $\mM_f$ of $f$ is the $2^n \times
2^n$ binary matrix with $[\mM_f]_{xy} = f(\bin(x),\bin(y))$, where $\bin : 2^n
\to \{0,1\}^n$ denotes the $n$-bit binary encoding of its input number. The
\emph{probabilistic communication complexity} of $f$ is the smallest number of
bits Alice and Bob have to communicate in order to compute $f(x,y)$ correctly
with probability larger than $\frac{1}{2}$. It is known that the probabilistic
communication complexity of $f$ and $\log(\srank(\mM_f))$ differ by
at most one~\cite{GeometricalRealizations,ProbabilisticCommunicationComplexity,LinearLowerBound}.
Sign rank was also studied in learning theory to
understand the limits of large margin
classification~\cite{SignRankVsVCDimension,PartitioningAndGeometricEmbedding,BetterLinearLowerBound,Ben03Limitations};
see Alon et al.~\cite{SignRankVsVCDimension} for a summary of applications of sign rank.
These complexity results focus on achieving lower and upper bounds on sign rank
as well as the separation of complexity classes.
We review some of these results
in subsequent sections and present them in terms of rounding rank, thereby
making them accessible to the data mining community.

Ben-David et al.~\cite[Cor.~14]{Ben03Limitations} showed that
  only a very small fraction of the $n \times n$ sign matrices can be
  well-approximated (with ``vanishing'' error in at most $n^{-O(1)}$ entries)
  by matrices of sign-rank at most $k$ unless $k = \omega(n^{1 - O(1)})$ is very
  large. 
  To the best of our knowledge, there are no known results for fixed relative
  error (e.g., 5\% of the matrix entries) or for the \minErrorRRproblem{$k$}
  problem. 


\subsection{Characterization of Rounding Rank}

Below we give a geometric
interpretation of rounding rank that helps to relate it to problems in data
mining. A similar theorem was presented in the context of communication
complexity~\cite[Th.~5]{ProbabilisticCommunicationComplexity}. Our presentation is in terms of
matrix ranks (instead of communication protocols) and gives a short proof that
provides insights into the relationship between rounding rank and geometric
embeddings.

\begin{theorem}
\label{Thm:CharacterisationRoundingRank}
Let $d \in \mathbb{N}$ and $\tau\in\R$. The
following statements are equivalent:
\begin{enumerate}
\item $\rrank_\tau(\mB) \leq d$.
\item There exist points $\vx_1, \dots, \vx_m \in \mathbb{R}^d$ and affine
  hyperplanes $H_1, \dots, H_n$ in $\R^d$ with normal vectors $\vc_1, \dots,
  \vc_n \in \mathbb{R}^d$ given by $H_j = \{ \vx \in \R^d : \langle \vx, \vc_j
  \rangle = \tau \}$ such that $\round_\tau(\langle \vx_i, \vc_j \rangle) =
  \mB_{ij}$ for all $i, j$.
\end{enumerate}
\end{theorem}
\begin{proof}
  $2 \Rightarrow 1$: Consider points $\vx_i$ and hyperplanes $H_j$ with the
  property asserted in the theorem. Define an $m \times d$ matrix $\mL$ with the
  $\vx_i$ in its rows, and an $n \times d$ matrix $\mR$ with the $\vc_j$ in its
  rows. Then $\round_\tau(\mL\mR^T) = \mB$, and hence $\rrank_\tau(\mB) \leq d$.

  $1 \Rightarrow 2$: Let $\mB = \round_\tau(\mA)$ with $\rank(\mA) \leq d$. Pick
  any two real matrices $\mL$ and $\mR$ with $d$ columns s.t.\ 
  $\mL\mR^T=\mA$. We can consider the rows $\mL_i$ of $\mL$ as points in $\R^d$
  ($\vx_i=\mL_i$) and the rows $\mR_j$ of $\mR$ as the normal vectors
  ($\vc_j=\mR_j$) of affine hyperplanes $H_j$ with offset $\tau$. Since $\mB =
  \round_\tau(\mA)$, we also get $\mB_{ij} = \round_\tau(\langle \mL_i, \mR_j
  \rangle)$ for all $i,j$.
\end{proof}

Fig.~\ref{fig:Hyperplanes} illustrates
Th.~\ref{Thm:CharacterisationRoundingRank} in $\R^2$ with $n=3$ and $\tau = 0$.
The three hyperplanes dissect the space into six convex, open regions. Each
point $\vx \in \R^2$ can be labeled with a binary vector according to whether it
is ``above'' or ``below'' each of the hyperplanes $H_j$ by using the rounding
function $\round_\tau(\langle \vx, \vc_j\rangle)$. 
\begin{figure}[t]
  \centering
  \includegraphics[width=0.75\columnwidth]{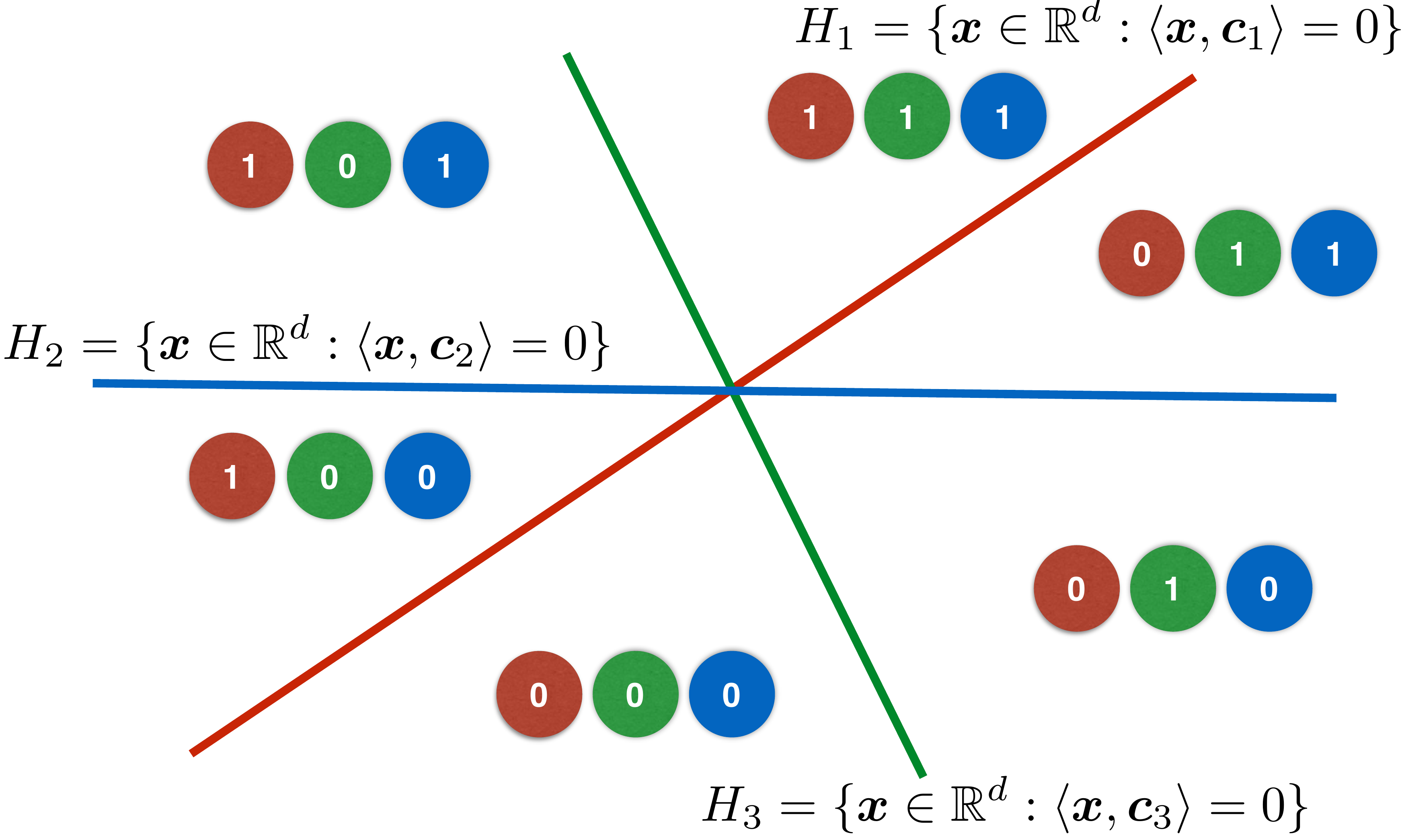}%
  \caption{Three hyperplanes in $\R^2$ with the labels of the subspaces into which they
    dissect the space. 
    Any $m\times 3$ binary matrix in which each row corresponds to one of the
    six label vectors has rounding rank at most 2.}
  \label{fig:Hyperplanes}
\end{figure}

The second point of Th.~\ref{Thm:CharacterisationRoundingRank} can be
interpreted as follows: Pick a binary matrix $\mB$ and treat each of the $n$
columns $\mB_j$ as the labels of a binary classification problem $P_j$ on $m$
points.  We can find data points $\vx_1, \dots, \vx_m$ and affine hyperplanes
$H_1, \dots, H_n$ in $\R^d$ which solve all linear classification problems $P_j$
without error if and only if the rounding rank of $\mB$ is at most~$d$. We then
interpret the $\vx_i$ as data points and the $\vc_j$ as feature weights. 
Rounding rank decompositions thus describe the ``best case'' for multiple linear
classification problems: if the rounding rank of $\mB$ is $d$, then we need at
least $d$ features to achieve perfect classification. In other words, we need to
collect at least $\rrank(\mB)$ features (or attributes) to have a chance to
classify perfectly. Similarly, if we employ dimensionality reduction, linear
classification cannot be perfect if we reduce to less than $d$ dimensions.

\begin{corollary}[informal]
  \label{cor:CorollaryRR}
  Rounding rank provides a natural lower bound on how many features we need for
  linear classification. This provides us with lower bounds on data collection
  or dimensionality reduction.
\end{corollary}

\subsection{Comparison of the Ranks}

We compare rounding rank with several well-known ranks. Many of the results in
this subsection were obtained for sign rank in the communication complexity
community; we present these results here in terms of rounding rank. To the best
of our knowledge, we are the first to make the role of the rounding
threshold explicit by introducing mixed matrices (see
Prop.~\ref{prop:changingRoundingThreshold}).

\paragraph{Boolean rank}
For binary matrices $\mL \in \{0,1\}^{m \times k}$ and $\mR \in \{0,1\}^{n
  \times k}$, the \emph{Boolean matrix product $\mL \bprod \mR^T$} is given by
the $m \times n$ binary matrix with $[ \mL \bprod \mR^T]_{ij} =
\bigvee_{\ell=1}^k (\mL_{ik} \land \mR_{jk})$ for all entries $i,j$.  The
\emph{Boolean rank of a binary matrix $\mB$}, denoted $\rankB(\mB)$, is the
smallest $k \in \mathbb{N}$ s.t.\ there exist $\mL \in \{0,1\}^{m \times k}$ and
$\mR \in \{0,1\}^{n \times k}$ with $\mB = \mL \bprod
\mR^T$~\cite{miettinen08discrete}. The rounding rank is a lower bound on the
Boolean rank.
\begin{lemma}
  $\rrank(\mB) \leq \rankB(\mB)$.
\end{lemma}
\begin{proof}
  Let $\rankB(\mB) = k$. Then there exist matrices $\mL \in \{0,1\}^{m \times
    k}$ and $\mR \in \{0,1\}^{n \times k}$ s.t.\ $\mB = \mL \bprod \mR^T$.  If
  we use the algebra of $\R$, we get $[\mL \mR^T]_{ij} \geq \frac{1}{2}$ iff
  $\mB_{ij} = 1$. This implies $\round(\mL \mR^T) = \mB$ and $\rrank(\mB) \leq k
  = \rankB(\mB)$.
\end{proof}

\paragraph{Real rank}
Comparing rounding rank and real rank, we observe that $\mB = \round(\mB)$ for all
binary matrices $\mB$.  Hence, \[ \rrank(\mB) \leq \rank(\mB).\] This is in
contrast to the relationship between Boolean rank and standard rank, which
cannot be compared (i.e.\ neither serves as a lower bound to the other)~\cite{SurveyOfCliqueCoverings}.

Note that the rounding rank can be much lower than both real and Boolean
rank. For example, the $n\times n$ ``upper triangle matrix'' with 1's on the
main diagonal and above has real and Boolean rank $n$, but rounding rank $1$
(see Th.~\ref{Thm:NestedMatrices}). As another example, we show in Section~\ref{Sec:IdentityMatrices}
in the appendix that the $n\times n$ identity
matrix has rounding rank 2 for all $n\ge3$. In fact, while we know that a
real-valued $n \times n$ matrix can have rank up to $n$, the situation is
different for rounding rank: On the one hand, for large enough $n$, all $n
\times n$ matrices $\mB$ have $\rrank(\mB) \leq (\frac{1}{2} +
o(1))n$~\cite[Cor.~1.2]{GeometricalRealizations}. On the other hand, for each
$n$, there exist $n \times n$ matrices with $\rrank(\mB) \geq \frac{n}{32}$,
i.e., the rounding rank can indeed be linear in
$n$~\cite[Cor.~1.2]{GeometricalRealizations}.


It is well-known that real-valued matrices with all entries picked uniformly at
random from some bounded proper interval have full standard rank with
probability 1. For rounding rank, an $n \times n$ binary matrix sampled uniformly
at random has rounding rank $\Omega(n)$ with high
probability (see the proof of Cor.~1.2 in \cite{GeometricalRealizations}).  Hence,
the rounding ranks of random binary matrices are expected to be large. The
real-world data matrices in our experiments often had small rounding ranks,
though.

A lower bound on the rounding rank of a binary matrix $\mB$ can be derived from
the singular values of the sign matrix $\mB^\pm$.
\begin{proposition}
\label{prop:LowerBound}
  Let $r = \rank(\mB^\pm)$ and let $\sigma_1(\mB^\pm) \geq \dots \geq \sigma_r(\mB^\pm) > 0$ be
  the non-zero singular values of $\mB^\pm$. Then
  \begin{align*}
	  (\rrank_0(\mB) + 1) \sum_{i = 1}^{\rrank_0(\mB)} \sigma_i^2(\mB^\pm) \geq mn.
  \end{align*}
\end{proposition}
Prop.~\ref{prop:LowerBound} is a slight modification of a result
in~\cite[Th.~5]{BetterLinearLowerBound} and we give the proof in the appendix.

\paragraph{Role of rounding threshold} We compare the rounding ranks of a fixed
matrix for different rounding thresholds.
We call a binary matrix \emph{mixed}, if it contains no all-zero and no all-one
columns (or rows).
\begin{proposition}
\label{prop:changingRoundingThreshold}
For any $\mB$ and arbitrary $\tau\neq \tau' \in \mathbb{R}$, $\rrank_\tau(\mB)$
and $\rrank_{\tau'}(\mB)$ differ by at most $1$. If additionally $\tau,
\tau'\neq 0$, $\rrank_\tau(\mB) = \rrank_{\tau'}(\mB)$ if
$\sign(\tau)=\sign(\tau')$ or if $\mB$ is mixed.
\end{proposition}

To prove Prop.~\ref{prop:changingRoundingThreshold} we need
Lem.~\ref{lem:UsefulHyperplaneSeparation} below.  The lemma is implied by the
Hyperplane Separation Theorem~\cite[p. 46]{ConvexOptimization}, and we prove it
in the appendix.
\begin{lemma}
  \label{lem:UsefulHyperplaneSeparation}
  Let $A$ and $B$ be two disjoint nonempty convex sets in $\R^d$, one of which
  is compact.  Then for all nonzero $c \in \mathbb{R}$, there exists a
  nonzero vector $\vv \in \R^d$, such that $\langle \vx, \vv \rangle > c$ and
  $\langle \vy, \vv \rangle < c$ for all $\vx \in A$ and $\vy \in B$.
\end{lemma}
\begin{proof}[Proof of Prop.~\ref{prop:changingRoundingThreshold}]
  First claim: Let $\tau,\tau' \in \R$ be arbitrary and pick $k\in\N$, $\mL \in \R^{m
    \times k}$, $\mR \in \R^{n \times k}$ such that $\round_\tau(\mL \mR^T) =
  \mB$.  Set $c=\tau' - \tau$, then
\begin{align*}
  \mB_{ij}
    = \round_{\tau}([\mL\mR^T]_{ij})
	= \round_{\tau'} ([\mL \mR^T]_{ij} + c).
\end{align*}
Set $\mL'= \begin{pmatrix}\mL & c\mathbf1\end{pmatrix}$ and
$\mR'=\begin{pmatrix}\mR & \mathbf1 \end{pmatrix}$, where $\mathbf1$ denotes the
all-one vector. Then $\round_{\tau'}(\mL'\mR'^T)=\mB$ and thus
$\rrank_{\tau'}(\mB) \leq k+1$.

Second claim: Without loss of generality, assume that $\mB$ contains no all-zero and no all-one columns
(otherwise tranpose the matrix). Let $\tau,\tau' \neq 0$ and let $k$ and $\mL \mR^T$ be as
before. If $\sign(\tau)=\sign(\tau')$, set $c=\tau'/\tau>0$ and
$\mR'=c\mR$. Then $\round_\tau(\mL\mR)=\round_{\tau'}(\mL\mR')$ by construction
so that $\rrank_{\tau'}(\mB)\le k$. By reversing the roles of $\tau$ and $\tau'$
in the argument, we establish $\rrank_{\tau}(\mB)=\rrank_{\tau'}(\mB)$.

Suppose $\tau,\tau' \neq 0$ (not necessarily of same sign) and let $\mB$ be
mixed. We now treat the rows of $\mL$ as points $\mL_1, \dots, \mL_m$ in
$\mathbb{R}^k$ and show that there exists an $n \times k$ matrix $\mR'$
consisting of normal vectors of affine hyperplanes in $\mathbb{R}^k$ in its rows
such that the hyperplanes separate the points with rounding threshold $\tau'$,
thereby establishing $\rrank_{\tau'}(\mB)\le\rrank_{\tau}(\mB)$. Again, by
reversing the roles of $\tau$ and $\tau'$, we obtain equality. To construct the
$j$'th row of $\mR'$, let $C_j= \{ \mL_i : \mB_{ij} = 1 \}$ and $\bar C_j = \{
\mL_i : \mB_{ij} = 0 \}$.  Notice that since $\mB$ is mixed, both $C_j$ and
$\bar C_j$ are non-empty.  We observe that the convex hulls of $C_j$ and
$\bar C_j$ are separated by the affine hyperplane with the $j$'th row of $\mR$
as its normal vector and offset from the origin $\tau$. Thus, we apply
Lem.~\ref{lem:UsefulHyperplaneSeparation} to obtain a vector $\vr'_j$ s.t.\
$\langle \vr'_j, \vc \rangle > \tau'$ for all $\vc \in C_j$ and $\langle \vr'_j,
\bar\vc \rangle < \tau'$ for all $\bar\vc \in \bar C$.  We set $\vr'_j$ to
be the $j$'th row of $\mR'$. To obtain $\mR'$, we repeating this procedure for
each of its $n$ rows.
\end{proof}

The above proof can be adopted to show that if $\mB$ is mixed, even using a
different (non-zero) rounding threshold for each row (or column) does not affect
the rounding rank.

\paragraph{Non-negative rounding rank}
While the gap between rank and non-negative rank can be arbitrarily
large~\cite{NMFGap}, for rounding rank and non-negative rounding rank this is
not the case.
\begin{proposition}
  \label{prop:rrankp_vs_rrank}
$\rrankp_\tau(\mB) \leq \rrank_\tau(\mB) + 2$.
\end{proposition}
This can be shown using ideas similar to the ones
in~\cite{ProbabilisticCommunicationComplexity} by a simple but lengthy computation.
We give a proof in the appendix.

\subsection{Computational Complexity}
The following proposition asserts that rounding rank is \NP-hard to compute
regardless of the rounding threshold.

\begin{proposition}
  \label{prop:rrNPhard}
  It is \NP-hard to decide if $\rrank_0(\mB) \leq k$ for all $k > 2$.  For $\tau
  \neq 0$, it is \NP-hard to decide if $\rrank_\tau(\mB) \leq k$ for all $k > 1$.
\end{proposition}
For sign rank (i.e. $\tau = 0$), this was proven in
\cite[Th.~1.2]{SignRankNP},\cite[Sec.~3]{VisibilityConstraints}. 
Moreover, Alon
et al.~\cite{SignRankVsVCDimension} argue that computing the sign rank is
equivalent to the existential theory of the reals. For $\tau \neq 0$,
NP-hardness was proven in~\cite[Th.~10]{SphereAndDotProduct}.

It is an open problem whether sign rank or rounding rank computation is in \NP.
Assume we store a matrix $\mA$ that achieves the rounding rank of $\mB$ by
representing all entries with rational numbers. The following proposition
asserts that in general, the space needed to store a $\mA$ can be exponential in
the size of $\mB$.  Hence, the proposition rules out proving that computing
rounding rank is in \NP{} by nondeterministically guessing a matrix $\mA$ of
small rank and rounding it.

\begin{proposition}
\label{prop:hugeEntries}
For all sufficiently large $n$, there exist $n \times n$ binary matrices $\mB$
with $\rrank(\mB) = 3$ s.t.\ 
for each matrix $\mA$ with $\rank(\mA) = 3$ and
$\round(\mA) = \mB$,
it takes $\Theta(\exp(n))$ bits to store the entries of $\mA$
using rational numbers.
\end{proposition}

Prop.~\ref{prop:hugeEntries} can be
derived from the proof of \cite[Th.~4]{SphereAndDotProduct}.

\begin{lemma}
  The \minErrorRRproblem{$k$} problem is \NP-hard to solve exactly. It is also
  \NP-hard to approximate within any polynomial-time computable factor.
\end{lemma}
\begin{proof}
  Both claims follow from Prop.~\ref{prop:rrNPhard}. If in polynomial time we
  could solve the \minErrorRRproblem{$k$} problem exactly or within any factor,
  then we could also decide if $\rrank(\mB) \leq k$ by checking if the result
  for \minErrorRRproblem{$k$} is zero.
\end{proof}


\section{Computing the Rounding Rank}
\label{sec:comp-round-rank}

In this section, we provide algorithms approximating $\rrank(\mB)$ and for
approximately solving the \minErrorRRproblem{$k$} problem. The algorithms are
based on some of the most common paradigms for algorithm design in data mining.
The \proje{} algorithm makes use of randomized projections, \SVD{} uses
truncated SVD, \lpca{} uses logistic PCA, and \asso{} is a Boolean matrix
factorization algorithm.
%
%
For each algorithm, we first discuss how to obtain an approximation to
$\rrank(\mB)$ (in the form of an upper bound) and then discuss extensions
to solve \minErrorRRproblem{$k}$.

\paragraph{Projection-based algorithm (\proje)}
We first describe a Monte Carlo algorithm to decide whether
$\rrank(\mB) \leq d$ for a given matrix $\mB$ and $d \in \mathbb{N}$. The
algorithm can output \textsc{yes} or \textsc{unknown}. If the algorithm outputs \textsc{yes}, it also
produces a rounding rank decomposition. We use this algorithm for different
values of $d$ to approximate $\rrank(\mB)$.

The decision algorithm is inspired by a simple observation: Considering an $m
\times n$ binary matrix $\mB$, we have $\mB = \round(\mB \mI)$, where $\mI$
denotes the $n \times n$ identity matrix.  We interpret each row $\mB_i$ of
$\mB$ as a point in $\R^n$ and each column $\mI_j$ of $\mI$ as the normal vector
of a hyperplane in $\R^n$. The hyperplane given by $\mI_j$ separates the
points $\mB_i$ into the classes $C_j = \{ \mB_{i} : \mB_{ij} = 1 \}$ and $\bar
C_j = \{ \mB_{i} : \mB_{ij} = 0 \}$ by the $j$'th column of $\mB$, since
$\mB_{ij} = \round(\langle \mB_i, \mI_j \rangle)$. The idea of \proje is to
take the points $\mB_i$ (the rows of $\mB$) and to project them into 
lower-dimensional space $\R^d$, $d \ll n$, to obtain vectors $\mL_1, \dots,
\mL_m \in \R^d$. We use a randomized projection that approximately preserves the
distances of the $\mB_i$ and---if $\mB$ has rounding rank at most $d$---try (or
hope) to maintain the separability of the points by hyperplanes by doing
so. Given the projected vectors in $\R^d$, we check separability by affine
hyperplanes and find the corresponding normal vectors $\mR_1, \dots, \mR_n$
using a linear program. If the $\mL_i$ turn out to be separable, we have
$\mB_{ij} = \round(\langle \mL_i, \mR_j \rangle)$ for all $i,j$ and thus $\mB =
\round(\mL \mR^T)$, where $\mL$ and $\mR$ have the $\mL_i$'s and $\mR_j$'s in
their rows, respectively. We conclude $\rrank(\mB) \leq d$ and output
\textsc{yes}. If the $\mL_i$ are not separable, no conclusions can be drawn and the
algorithm outputs \textsc{unknown}.

The Johnson--Lindenstrauss Lemma~\cite{JohnsonLindenstraussLemma} asserts that
there exists a linear mapping $\mA$ that projects points from a high-dimensional
space into a lower-dimensional space while approximately preserving the
distances. We use the projections proposed by Achlioptas~\cite{Achlioptas03Database} to obtain $\mA$.
We set $\mL_i=\mB_i\mA$. The linear program (LP) to compute the normal vector $\mR_j$ is
\begin{align*}
    \text{find} \quad & \mR_j \in\R^{d} 	   & & \\
  \text{subject to} \quad & \sum_{k = 1}^d \mL_{ik} \mR_{jk} \geq \tau + \varepsilon
			& \text{if }\mB_{ij} = 1, \\
                            & \sum_{k = 1}^d \mL_{ik} \mR_{jk} \leq \tau - \varepsilon
			& \text{if }\mB_{ij} = 0. 
\end{align*}
We enforce strict separability by introducing an offset $\varepsilon >
0$. In practice, we set $\varepsilon$ to the smallest positive number
representable by the floating-point hardware. Notice that the LP only aims at
finding a feasible solution; it has $m$ constraints and $d$ variables.


To approximate the rounding rank, we repeatedly run the above
algorithm with increasing values of $d$ until it outputs \textsc{yes}; i.e., $d=1,2,\ldots$.
Alternatively, we could use some form of binary search to find a suitable
value of $d$. In practice, however, solving the LP for large
values of $d$ slows down the binary search too much.  

\eat{
Algorithm~\ref{Alg:RandomisedDecisionAlgo} gives the pseudocode
for \proje.
We utilise Algorithm~\ref{Alg:RandomisedDecisionAlgo} to compute
an approximation of $\rrank(\mB)$: For $d = 1$ we query if
$\rrank(\mB) \leq d$. If the result is true, then we output $d$ as an approximation
for the rounding rank. Otherwise, we set $d = d + 1$ and repeat the procedure.

\begin{algorithm}[tb]
  \DontPrintSemicolon
  \KwData{A binary matrix $\mB \in \{ 0, 1 \}^{m \times n}$, a dimensionality $d \in \mathbb{N}$
	  		and a rounding threshold $\tau \in \R$.}
  \KwResult{\textit{True} and matrices $\mL$ and $\mR$ with $\round_\tau(\mL \mR^T) = \mB$
	  			if the algorithm found a rounding rank decomposition in $\R^d$;
			\textit{False} otherwise.}
  Sample a dimensionality-reduction mapping $\mA \in \R^{n \times d}$ as in Equation~2 of \cite{Achlioptas03Database} \;
  Set $\mL \leftarrow \mB \mA$ \;
  \For{ $j \leftarrow 1$ to $n$ } {
	  Find $\vr_j$ from \eqref{Eq:HyperplaneLP} \;
  	  \If{ the LP had no feasible solution }{
	      \Return{ False } \;
      }
  }
  Set $\mR \leftarrow \begin{pmatrix} \vr_1 & \cdots & \vr_n \end{pmatrix}^T$ \;
  \Return{ True and $\mL$ and $\mR$ } \;
  \caption{Algorithm deciding if $\rrank(\mB) \leq d$.}
  \label{Alg:RandomisedDecisionAlgo}
\end{algorithm}
}

To solve \minErrorRRproblem{$k$}, we modify the LP of \proje to output
an approximate solution. For this purpose, we introduce non-negative
slack-variables $\vc_i$ as in soft-margin SVMs to allow for errors, and an
objective function that minimizes the $L_1$ norm of the slack variables.  We
obtain the following LP:

\begin{alignat}{4}
	\min_{\substack{\vc \in \R_{\geq 0}^m \\ \mR_j \in \R^d}} 	& \quad &  \sum_{i = 1}^m &\;\vc_i \notag \\
  \text{subject to} & & \sum_{k = 1}^d &\;\mL_{ik} \mR_{jk} + \vc_i &\;\geq\;& \tau + \varepsilon,
			&\qquad& \text{if }\mB_{ij} = 1, \notag  \\
    & & \sum_{k = 1}^d &\;\mL_{ik} \mR_{jk} - \vc_i &\;\leq\;& \tau - \varepsilon,
			& & \text{if }\mB_{ij} = 0. \notag 
\end{alignat}

\paragraph{Rounded SVD algorithm (\SVD)}
We use rounded SVD to approximate $\rrank(\mB)$.  The algorithm
is greedy and similar to the one in \cite{NeighborhoodData}.
Given a binary matrix $\mB$, the algorithm sets $k = 1$.  Then
it computes the rank-$k$ truncated SVD of $\mB$ and rounds it.  If the rounded
matrix and $\mB$ are equal, it returns $k$, otherwise, it sets $k = k+1$
and repeats. The underlying reasoning is that the rank-$k$ SVD is the
real-valued rank $k$ matrix minimizing the distance to $\mB$ w.r.t.\ the
Frobenius norm. Hence, also its rounded version should be ``close'' to $\mB$.

To approximately solve \minErrorRRproblem{$k$}, we compute the truncated
rank-$\ell$-SVD of $\mB$ for all $\ell = 1, \dots, k$ and return the rounded
matrix with the smallest error.

\paragraph{Logistic Principal Component Analysis (\lpca)}
The logistic function $f(x) = \left(1 + e^{-x}\right)^{-1}$ is a differentiable
surrogate of the rounding function and it can be used to obtain a smooth approximation of the rounding. 

\lpca~\cite{schein03generalized} models each $\mB_{ij}$ as a Bernoulli random
variable with success probability $f(\langle \mL_i,\mR_j\rangle)$, where $\mL
\in \R^{m \times k}$ and $\mR \in \R^{n \times k}$ are unknown parameters. Given
$\mB$ and $k \in \mathbb{N}$ as input, L-PCA obtains (approximate)
maximum-likelihood estimates of $\mL$ and $\mR$. If each $f(\langle
\mL_i,\mR_j\rangle)$ is a good estimate of $\mB_{ij}=1$, then $\lVert \mB
- \round(\mL\mR^T) \rVert_F$ should be small.


To approximate the $\rrank(\mB)$, we
run \lpca on $\mB$ for $k = 1$ and check if $\round(\mL \mR^T) = \mB$.
If this is the case, we return $k$, otherwise, we set $k = k + 1$ and repeat.

To use \lpca to compute an approximation of \minErrorRRproblem{$k$}, we
simply run \lpca and apply rounding.

\paragraph{Permutation algorithm (\shay)}
The only known algorithm to approximate the sign rank of a $n \times n$ matrix
in polynomial time was given in~\cite{SignRankVsVCDimension}; it guarantees an
upper bound within an approximation ratio of $O(n / \log n)$. By
Prop.~\ref{prop:changingRoundingThreshold}, we can use this method to
approximate the rounding rank.  The algorithm permutes the rows of the
input matrix $\mB$ s.t.\ the maximum number of bit flips over all
columns is approximately minimized. It then algebraically
approximates $\rrank(\mB)$ by evaluating a certain polynomial
based on the occurring bit flips. 
\eat{
\sn{Is this intuition OK? In case you are interested, the polynomial part works the following way: 
	Consider the matrix $\mB$ after reordering. We build a polynomial for each column of $\mB$.
	A bit flip from 0 to 1 (or vice versa) in a column of $\mB$ corresponds to a root of the polynomial
	which as input gets the number of the row in which the flip occurs.
	Now we create the factorization $\mL$ and $\mR$ by writing the coefficients of the polynomial into
	a row of $\mL$ and the points at which we evaluate the polynomial into the column of $\mR$.}
}
      The algorithm cannot solve the \minErrorRRproblem{$k$} problem.

\paragraph{Nuclear norm algorithm (\nuclear)}
The nuclear norm $\norm{\mX}_*$ of a matrix $\matr{X}$ is the sum of the
singular values of $\matr{X}$ and is a convex and differentiable surrogate of
the rank function of matrix. A common relaxation for minimum-rank matrix
factorization is to minimize $\norm{\mX}_*$ instead of $\rank(\mX)$. In our
setting, we obtain the following minimization problem:
\begin{alignat}{3}
  \mX^* = &\argmin_{\mX \in \R^{m\times n}} &\quad &  \norm{\mX}_*  & & \notag \\
               &\;\text{subject to}                      &      & \mX_{ij} \geq \tau &\qquad& \text{if }\mB_{ij} = 1, \notag  \\
               &                                                &      & \mX_{ij} < \tau &              &\text{if }\mB_{ij} = 0. \notag 
\end{alignat}

This method has some caveats: While $\mX^*$ must have small singular values, it
may still have \emph{many}. Additionally, by Prop.~\ref{prop:hugeEntries}, some
entries of a matrix $\mA$ achieving the rounding rank might be extremely
large. In such a case, some of the singular values of $\mA$ must also be large,
and consequently the nuclear norm of the matrix is large. Thus, $\mX^*$ might
have a too large rank. This algorithm cannot be extended to solve
\minErrorRRproblem{$k$}.


\section{Nested Matrices}
\label{sec:nestedness}

A binary matrix is nested if we can reorder its
columns such that after the reordering, the one-entries in each row form a
contiguous segment starting from the first column~\cite{MannilaTerziNested}.
Intuitively, nested matrices model subset/superset relationships between the
rows and columns of a matrix. Such structures are, for example, found in
presence/absence data of locations and species~\cite{MannilaTerziNested}.

We show that nested matrices are exactly the matrices with non-negative rounding
rank~1. Formally,
  a binary matrix $\mB$ is \emph{directly nested} if for each one-entry
  $\mB_{ij}=1$, we have $\mB_{i'j'}=1$ for all $i' \in \{1, \dots, i - 1\}$
  and $j' \in \{1, \dots, j-1\}$.
  A binary matrix $\mB$ is \emph{nested} if there exist permutation
  matrices $\mP_1$ and $\mP_2$, such that $\mP_1\mB\mP_2$ is directly nested.

\begin{theorem}
\label{Thm:NestedMatrices}
  Let $\bm0\neq\mB \in \{ 0, 1\}^{m \times n}$. Then $\mB$ is nested if and only if
  $\rrankp(\mB) = 1$.
\end{theorem}
\begin{proof}
  $\Rightarrow$: We reorder the rows and columns of
  $\mB$ by the number of $1$s they contain in descending order. This gives us
  permutation matrices $\mP_1$ and $\mP_2$ s.t.\ $\mB' = \mP_1 \mB \mP_2$ is
  directly nested.  Set $\vp = \mB' \bm1$, i.e., $\vp$ is the vector containing
  the row sums of $\mB'$.  Then for $\vec{l'}$ and $\vec{r'}$ with $\vec{l}'_i =
  2^{\vec{p}_i - 1}$ and $\vec{r}'_j = 2^{-j}$, $\mB' = \round(\vec{l}'
  \cdot (\vec{r}')^T)$.  Setting $\vec{l} = \mP_1^T \vec{l}'$ and $\vec{r} =
  \mP_2 \vec{r}'$, we get $\mB = \round(\vl \cdot \vr^T)$. Hence, we have
  $\rrank(\mB) = 1$.

  $\Leftarrow$: Let $\vl \geq 0$ and $\vr \geq 0$ be s.t.\ $\mB = \round(\vl
  \vr^T)$.  Then there exist permutation matrices $\mP_1$ and $\mP_2$ s.t.\ for
  $\vl' = \mP_1 \vl$ we have $\vl_1' \geq \dots \geq \vl_m'$ and for $\vr' =
  \mP_2^T \vr$ we have $\vr_1' \geq \dots \geq \vr_n'$.  Set $\mB' = \round(\vl'
  (\vr')^T)$ and observe $\vl_i' \vr_j' \geq \vl_{i+1}' \vr_j'$ for all
  $i,j$. Therefore, for each entry of $\mB'$, $\mB'_{ij} = \round(\vl_i' \vr_j')
  \geq \round(\vl_{i+1}' \vr_j') = \mB'_{(i+1)j}$.  Similarly,
  $\mB'_{ij} = \round(\vl_i' \vr_j') \geq \round(\vl_i' \vr_{j+1}') =
  \mB'_{i(j+1)}$.  Therefore, $\mB'$ is directly nested.  We conclude that $\mB
  = \vl \vr^T$ is nested since $\mB = \round(\vl \vr^T) = \mP_1^{T}
  \round(\mP_1(\vl \vr^T)\mP_2) \mP_2^{T} = \mP_1^{T} \mB' \mP_2^{T}$.
\end{proof}

Binary matrices with rounding rank~1 are also closely related to nested matrices.

\begin{proposition}
  \label{prop:rrank1_eq_block_nested}
  Let $\bm0\neq\mB \in \{ 0, 1\}^{m \times n}$. The following statements are
  equivalent:
  \begin{enumerate}
    \item $\rrank(\mB) = 1$.
	\item \eat{$\mB$ is nested or} there exist permutation matrices $\mP_1$ and $\mP_2$
		  and nested matrices $\mB_1$ and $\mB_2$, such that
	      \begin{align*}
			  \mB = \mP_1
			  \begin{pmatrix}
				\mB_1 & 0     \\
				0     & \mB_2
			  \end{pmatrix}\mP_2.
		  \end{align*}
  \end{enumerate}
\end{proposition}
The proof is in the appendix.


\paragraph{Algorithms}
Mannila and Terzi~\cite{MannilaTerziNested} introduced the \emph{Bidirectional
  Minimum Nestedness Augmentation} (\BMNA) problem: Given a binary matrix $\mB$,
find the nested matrix $\mA$ which minimizes $\norm{\mB - \mA}_F$.  We will
discuss three algorithms to approximately solve this problem.

\cite{MannilaTerziNested} gave an algorithm, \nmt, which approximates
a solution for the \BMNA problem by iteratively eliminating parts of the matrix
that violate the nestedness.

Next, we propose a alternating minimization algorithm, \nexhaust, which exploits
Th.~\ref{Thm:NestedMatrices}. 
\nexhaust maintains two vectors $\vec{l} \in \nR^m$ and $\vec{r} \in \nR^n$ and
iteratively minimizes the error $\lVert\mB - \round(\vec{l} \cdot
\vec{r}^T)\rVert_F$.  It starts by fixing $\vec{r}$ and updates $\vec{l}$, such
that the error is minimized. Then $\vec{l}$ is fixed and $\vec{r}$ is updated.
This procedure is repeated until the error stops reducing or we have reached a
certain number of iterations.

We describe an update of $\vec{l}$ for fixed $\vec{r}$; updating $\vr$ for given
$\vl$ is symmetric.  Observe that changing entry $\vl_i$ only alters the $i$'th
row of $\mA = \vl \cdot \vr^T$, and consequently $\mA_i$ is not affected by any
$\vl_k$ with $k \neq i$.  Hence, we only describe the procedure for updating
$\vl_i$.  Define the set $V_i = \{ \vr_j : \mB_{ij} = 1\}$ of all values of
$\vr$ where $\mB_i$ is non-zero.  We make the following observations: If we set
$\vl_i < \frac{1}{2 \max(\vr)}$, then $\mA_i$ only contains zeros after the
update.  If $\frac{1}{2 \max(\vr)} < \vl_i < \frac{1}{2 \max(V_i)}$, then after
the update all non-zeros of $\mA_i$ will be in entries where $\mB_i$ has a zero.
If $\vl > \frac{1}{2 \min(V_i)}$, we add too many $1$s to $\mA_i$.  Thus, all
values that we need to consider for updating $\vl_i$ are $\frac{1}{2 \max(\vr)}$
and the values in $\{ \frac{1}{2v} : v \in V_i\}$. The algorithm tries all
possible values for $\vl_i$ exhaustively and computes the error at each step.

We can also use the results of \nmt as initialization for \nexhaust:
We run \nmt and obtain a nested matrix $\mB$.  Now we use the construction from
step 1 of the proof of Th.~\ref{Thm:NestedMatrices} to obtain $\vl$ and $\vr$
with $\mB = \round(\vl \vr^T)$, and try to improve using \nexhaust.

Finally, we can use \SVD to solve the \BMNA problem approximately. By the
Perron--Frobenius Theorem \cite[Ch. 8.4]{horn13matrix}, the principal left and
right singular vectors of a non-negative matrix are also non-negative. Hence we
may use the \SVD algorithm to obtain the rank-$1$ truncated SVD and round. By
Th.~\ref{Thm:NestedMatrices}, the result must be nested.


\section{Experiments}
\label{sec:experiments}

We conducted an experimental study on synthetic and real-world datasets to
evaluate the relative performance of each algorithm for estimating the rounding
rank or for \minErrorRRproblem$k$.

\subsection{Implementation Details}
\label{sec:exp:implementation}
For
\lpca, we used the implementation by the authors of \cite{schein03generalized}.
We implemented \nmt and \shay in C and  all other algorithms in
Matlab. For \nuclear, we used the CVX package with the SeDuMi solver
\cite{cvx}. For solving the linear programs in \proje, we used Gurobi.

Due to numerical instabilities, \nuclear often returned a matrix with only
positive singular values (i.e.\ of full rank). We countered this by zeroing the smallest singular values of the returned matrix that did not affect to the result of the rounding.

All experiments were conducted on a computer with eight Intel Xeon E5530
processors running at 2.4\,GHz and 48\,GB of main memory. All our
algorithms and the synthetic data generators are available online.\!\footnote{\url{http://dws.informatik.uni-mannheim.de/en/resources/software/rounding-rank/}}

\subsection{Results With Synthetic Data}
\label{sec:exp:synthetic}

We start by studying the behavior of the algorithms under controlled synthetic datasets. 

\subsubsection{Data generation}
\label{sec:exp:synth:data}

We generated synthetic data by sampling two matrices $\mL\in\R^{m\times
  k}$ and $\mR\in\R^{n\times k}$ and then rounding their product to obtain $\mB =
  \round_\tau(\mL\mR^T)$ with rounding rank \emph{at
  most} $k$. The actual rounding rank of $\mB$ can be lower, however, because
there may be matrices $\mL'\in\R^{m\times k'}$ and $\mR'\in\R^{n\times
  k'}$ with $k' < k$ and $\round_\tau(\mL'\mR'^T) = \mB$. (In fact, we
sometimes found such matrices.) In some experiments, we additionally applied
noise by flipping elements selected uniformly at random. We report as
\emph{noise level} $p$ the ratio of the number of flipped elements to the number
of non-zeros in the original noise-free matrix.

We sampled every element of $\mL$ and $\mR$ i.i.d.~using two families of
distributions: uniform and normal distribution. For both distributions, we first
pick a desired expected value $\mu = \E[(\mL\mR^T)_{ij}]$ of each entry in
$\mL\mR^T$. We then parameterize the distributions such that the expected value
for an element of $\mL$ or $\mR$ is $q = \sqrt{\mu/k}$. For the normal
distribution, we set the variance to $1$, and for the uniform distribution, we
sampled from range $[q-1/2, q+1/2]$. 


We generated two sets of matrices. In the first set, the matrices were very
small, and it was used to understand the behavior of the slower algorithms. In
the second set, the matrices were medium-sized, to give us more realistic-sized
data, but we could use only some of the methods with these data. When generating
the data, we varied four different parameters: number of rows $m$, the planted
rank $k$, the expected value $\mu$, and the level of noise $p$. In all
experiments, we varied one of these parameters, while keeping the others
fixed. We generated all datasets with rounding threshold $\tau=1/2$. For the
small data, we used $n=100$ columns and the number of rows varied from $60$ to
$220$ with steps of $40$ with the default value being $n=100$.
The rank $k$ in the small matrices varied from $5$ to $30$
with steps of $5$, default being $k=10$; the expected value $\mu$ varied from
$0.1$ to $0.7$ with $0.1$ steps (default was $\mu=0.5$); the noise level $p$
varied from $0.05$ to $0.5$ with steps of $0.05$, and by default we did not
apply any noise. We generated ten random matrices with each parameter setting to
test the variation of the results.

For the medium-sized matrices, we used $n=300$ columns and the number of rows
varied from $400$ to $600$ with steps of $50$ the default being $m=500$; the
planted rank $k$ varied from $40$ to $100$ with default value $k = 60$; the
expected value and the noise were as with the small data. We generated five
random matrices with each parameter setting.

\subsubsection{Rounding rank}
\label{sec:exp:synth:rank}

In our first set of experiments, we studied the performance of the different
algorithms for estimating the rounding rank. The results for the small synthetic
datasets are summarized in Fig.~\ref{fig:synth:small}. The results are given for
the uniformly distributed factor matrices; the results with normally distributed
factors were largely similar and are postponed to the appendix.

\begin{figure*}[tb]
  \centering
  \includegraphics[height=\legendheight]{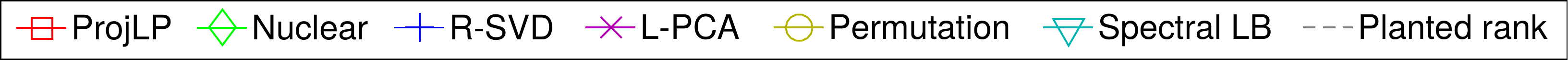} \\
  \rotatebox[origin=l]{90}{\hspace*{1em}\small Uniform dist.}\hspace*{\smallfigsep}%
  \subfigure[Rank, vary $m$]{%
    \label{fig:synth:small:n:unif:rank}%
    \includegraphics[width=\smallfigwidth]{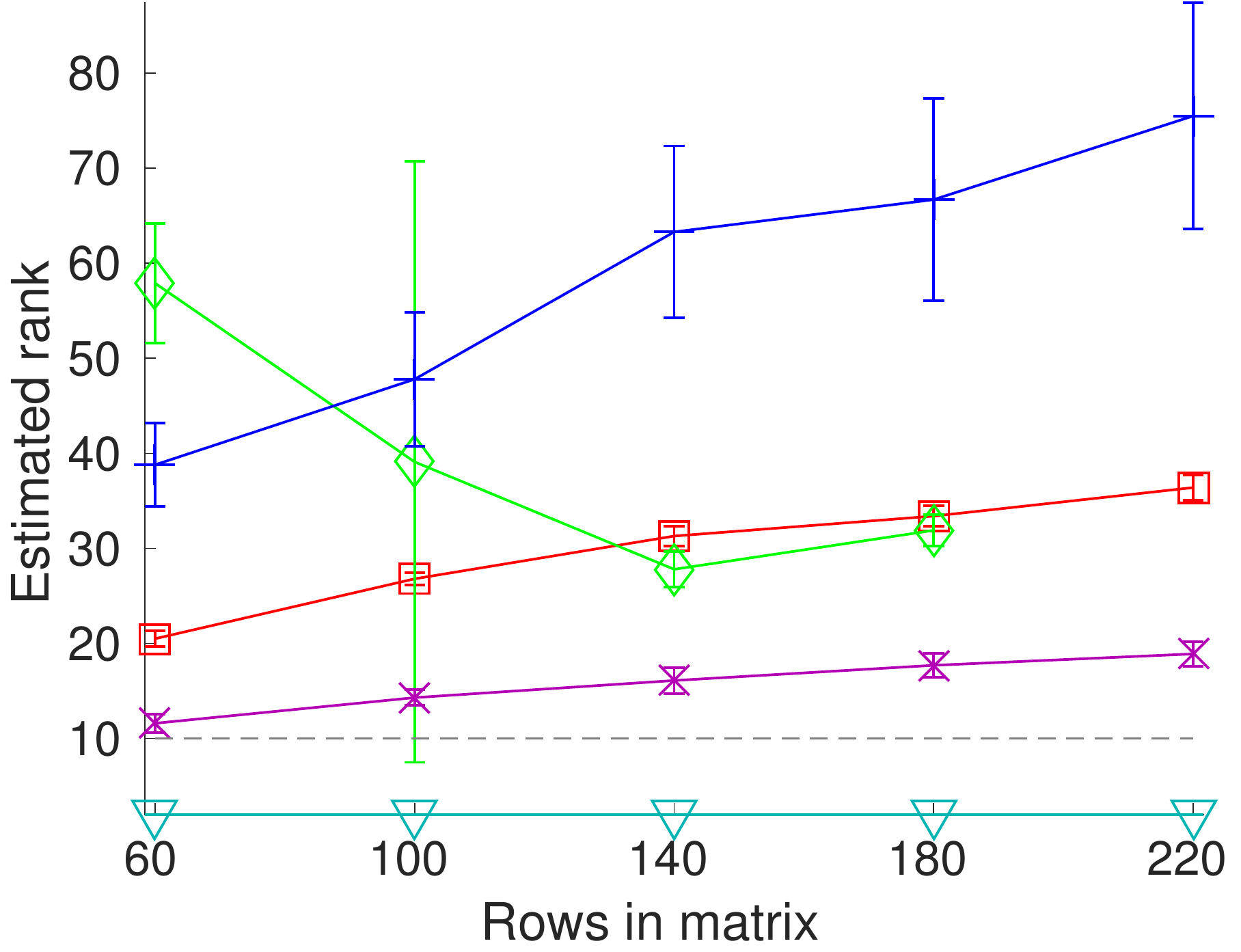}%
  }\hspace*{\smallfigsep}%
  \subfigure[Rank, vary $k$]{%
    \label{fig:synth:small:k:unif:rank}%
    \includegraphics[width=\smallfigwidth]{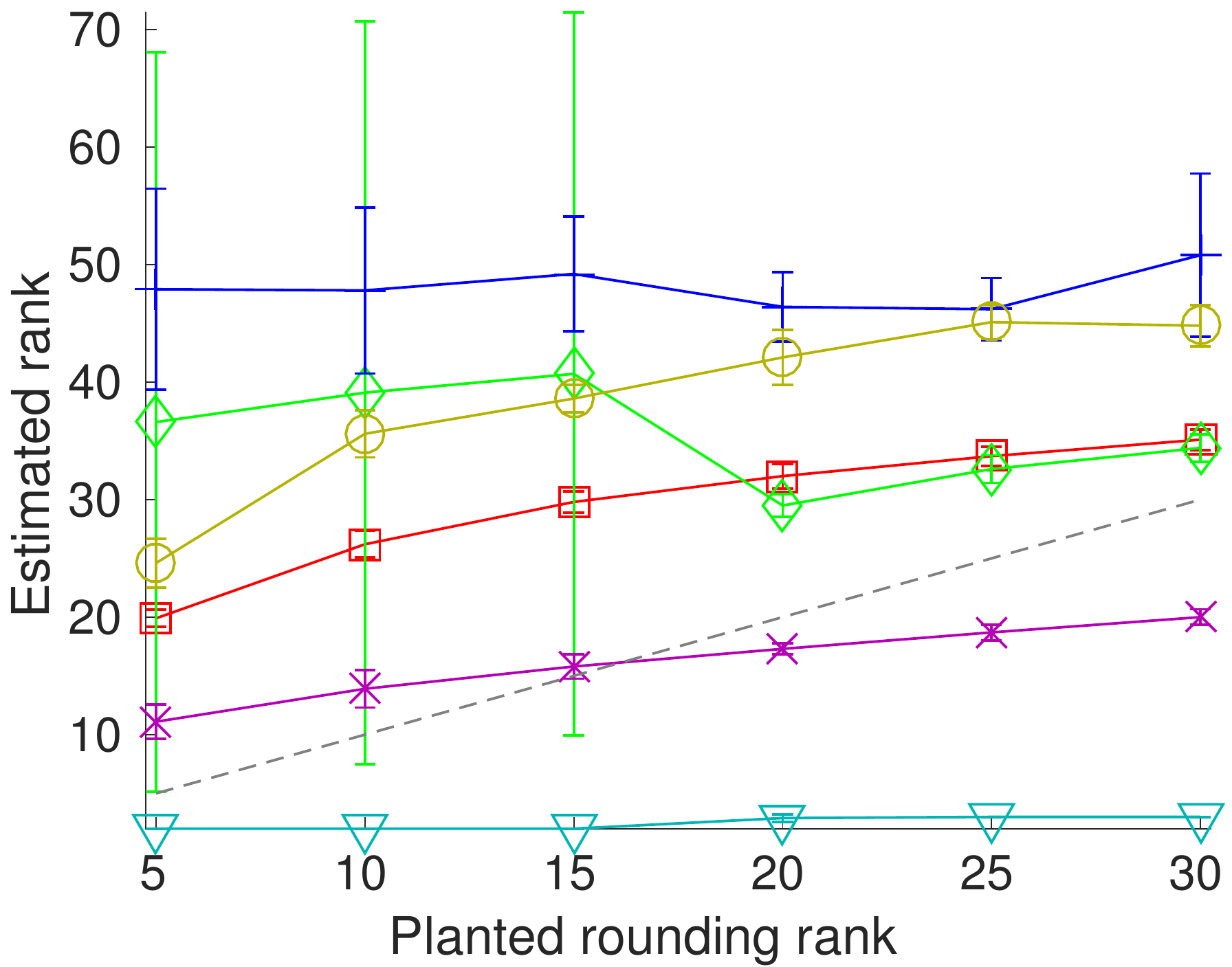}%
  }\hspace*{\smallfigsep}%
  \subfigure[Rank, vary $\mu$]{%
    \label{fig:synth:small:dens:unif:rank}%
    \includegraphics[width=\smallfigwidth]{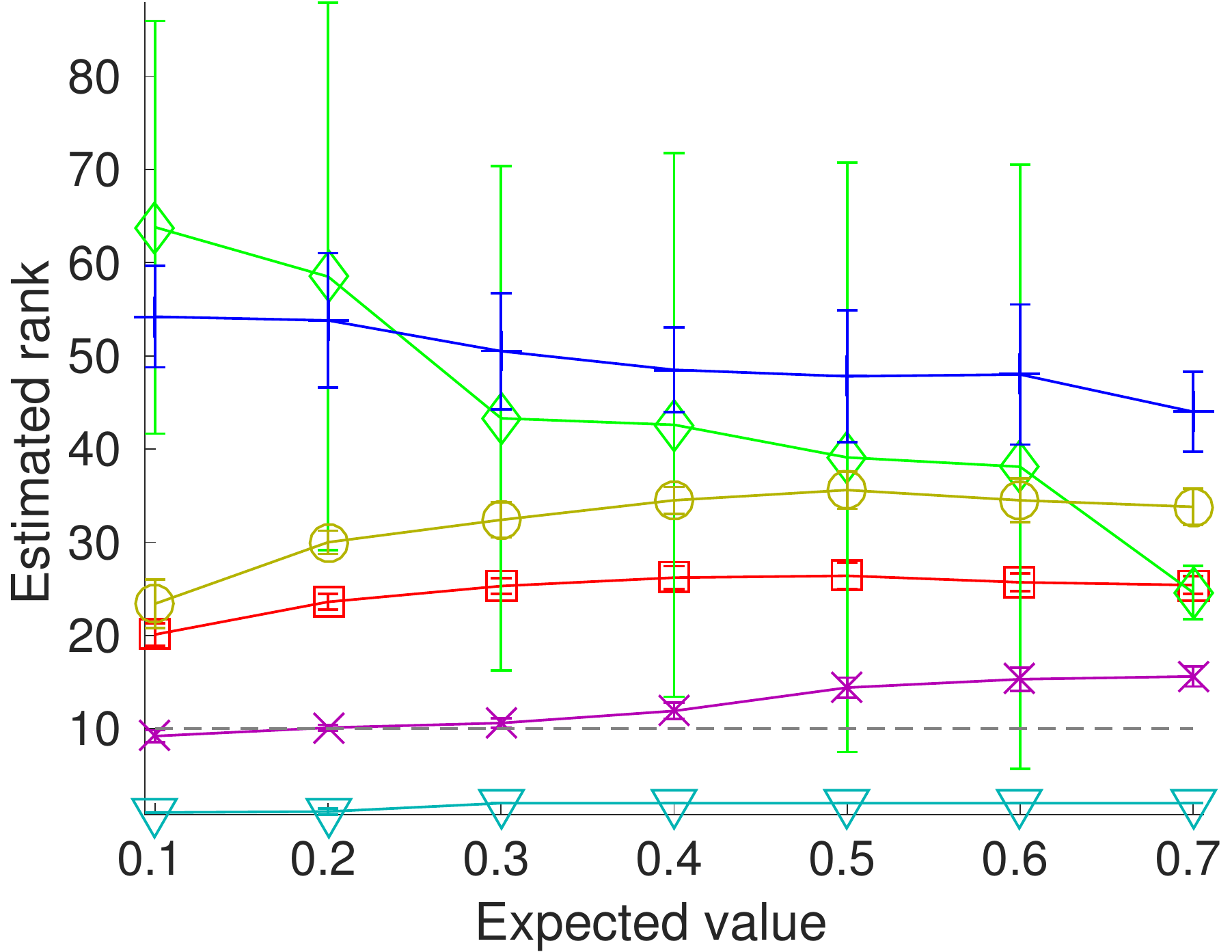}%
  }\hspace*{\smallfigsep}%
  \subfigure[Rank, vary $p$]{%
    \label{fig:synth:small:noise:unif:rank}%
    \includegraphics[width=\smallfigwidth]{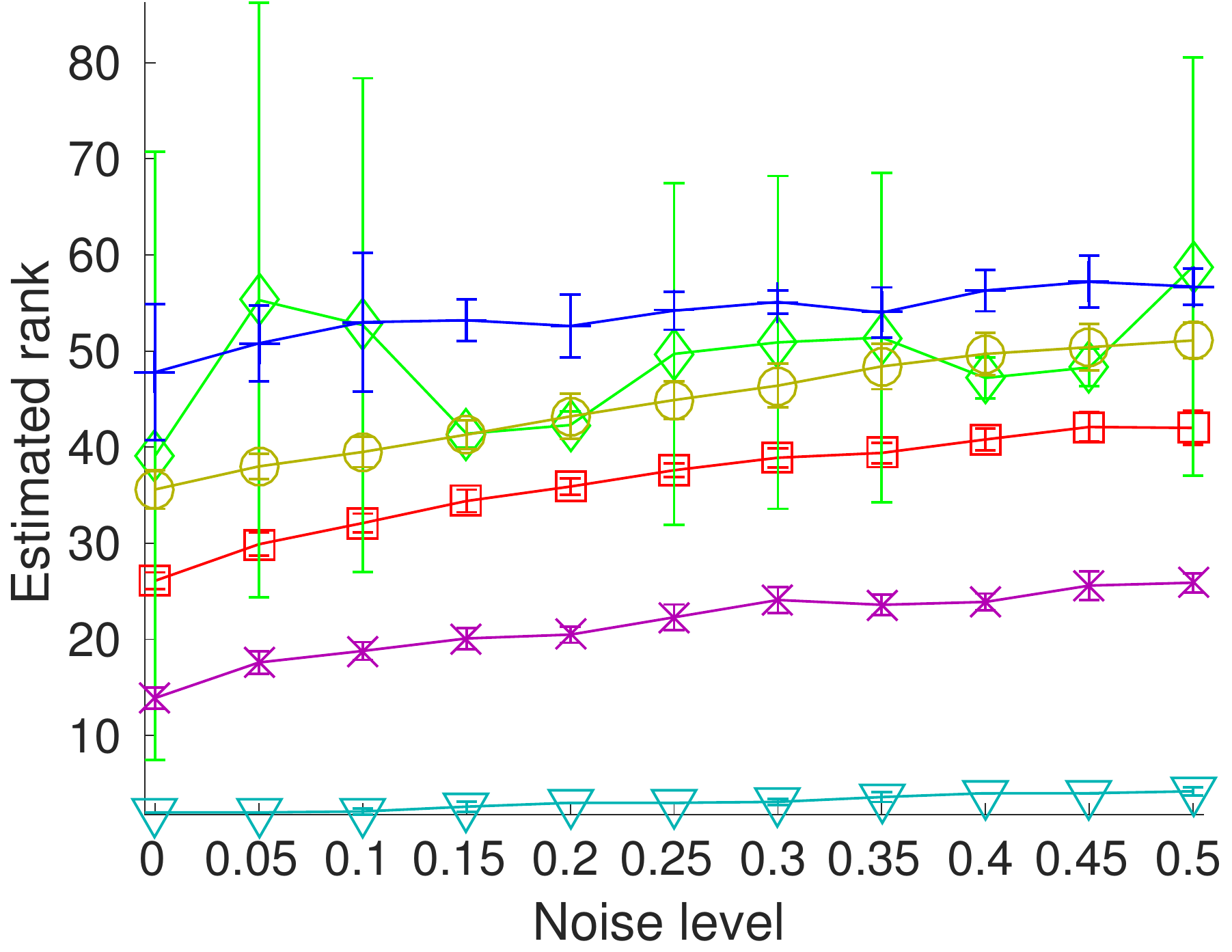}%
  }\\
  \rotatebox[origin=l]{90}{\hspace*{3em}\small Time}\hspace*{\smallfigsep}%
  \subfigure[Time, vary $m$]{%
    \label{fig:synth:small:n:unif:time}%
    \includegraphics[width=\smallfigwidth]{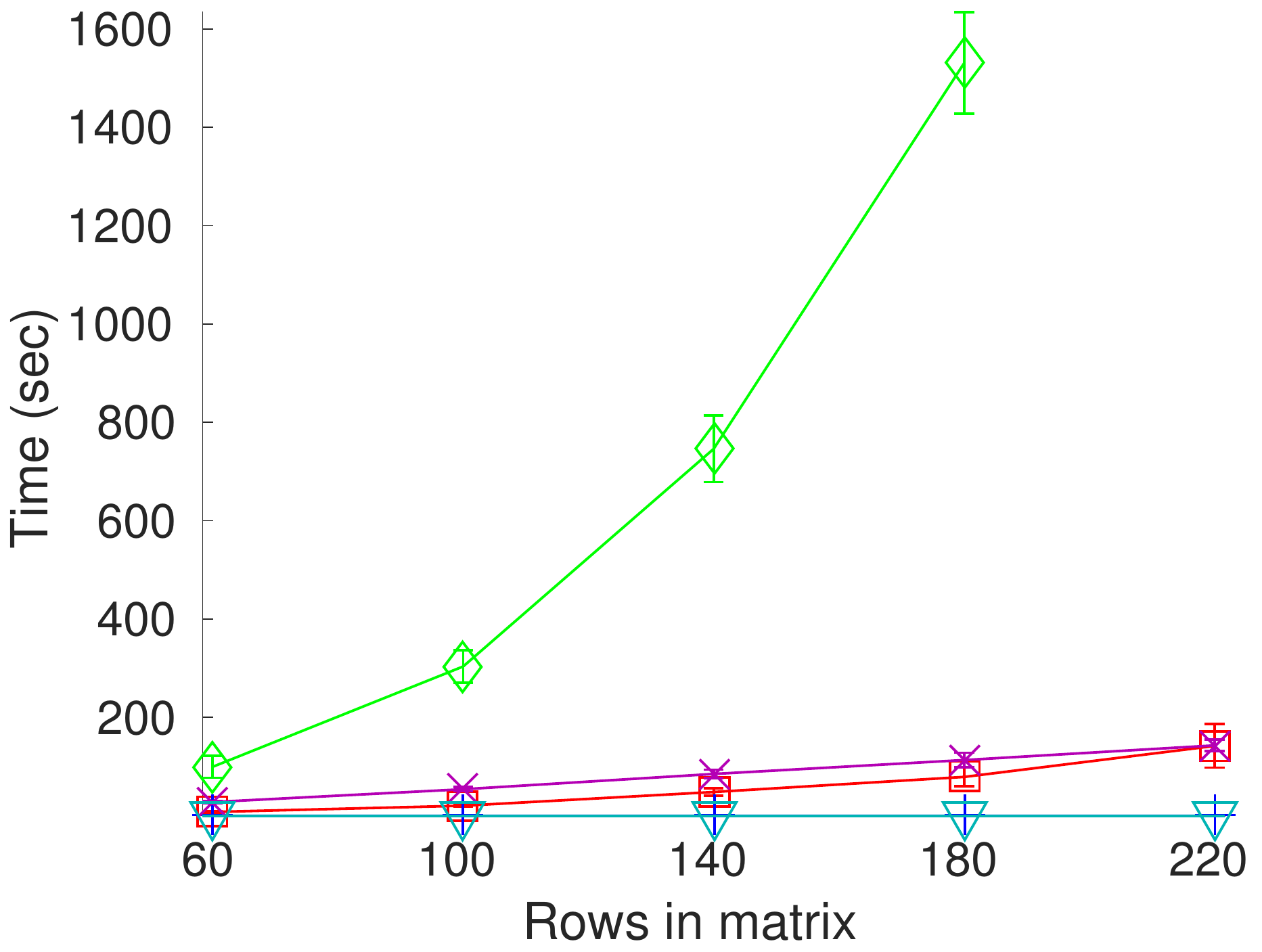}%
  }\hspace*{\smallfigsep}%
  \subfigure[Time, vary $k$]{%
    \label{fig:synth:small:k:unif:time}%
    \includegraphics[width=\smallfigwidth]{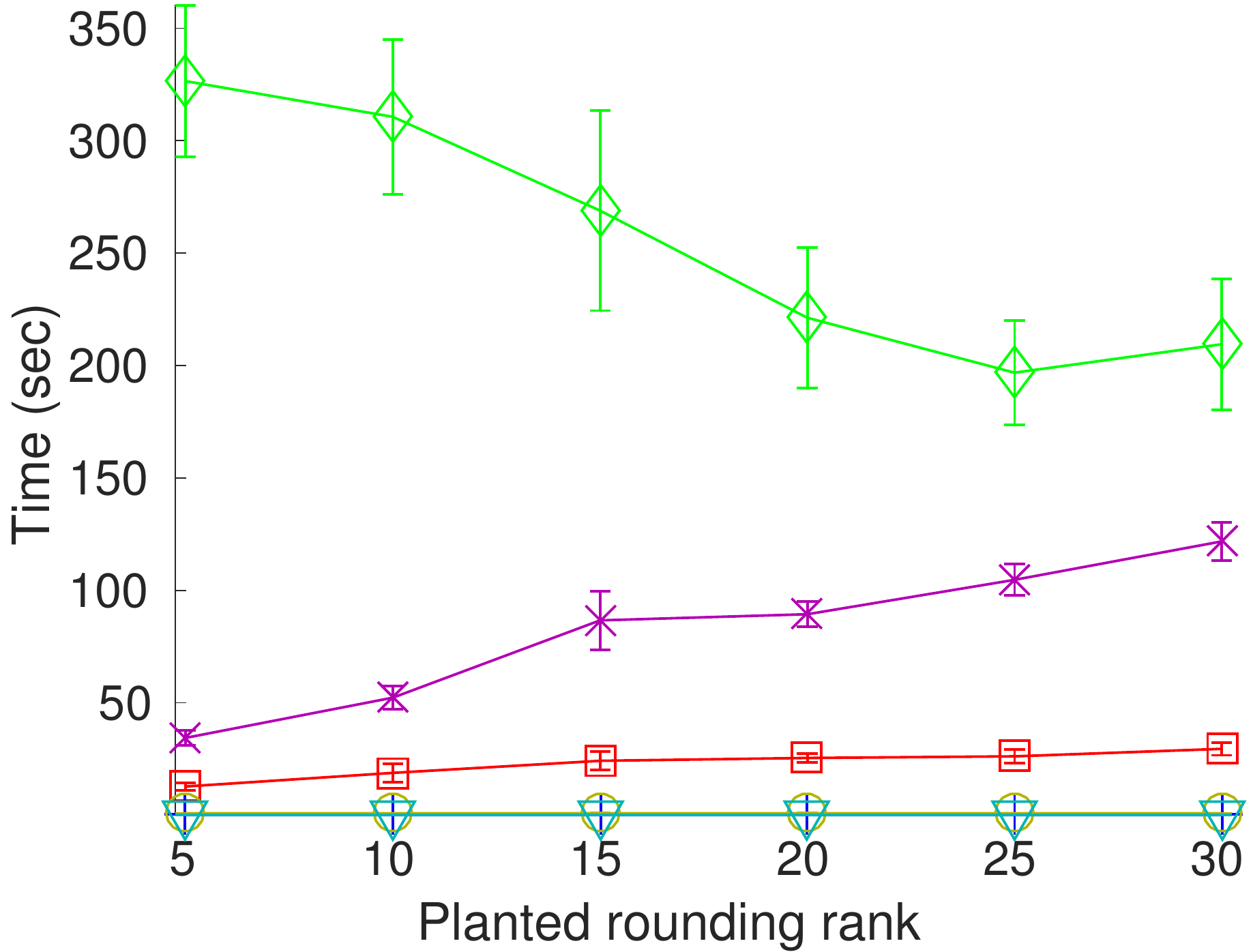}%
  }\hspace*{\smallfigsep}%
  \subfigure[Time, vary $\mu$]{%
    \label{fig:synth:small:dens:unif:time}%
    \includegraphics[width=\smallfigwidth]{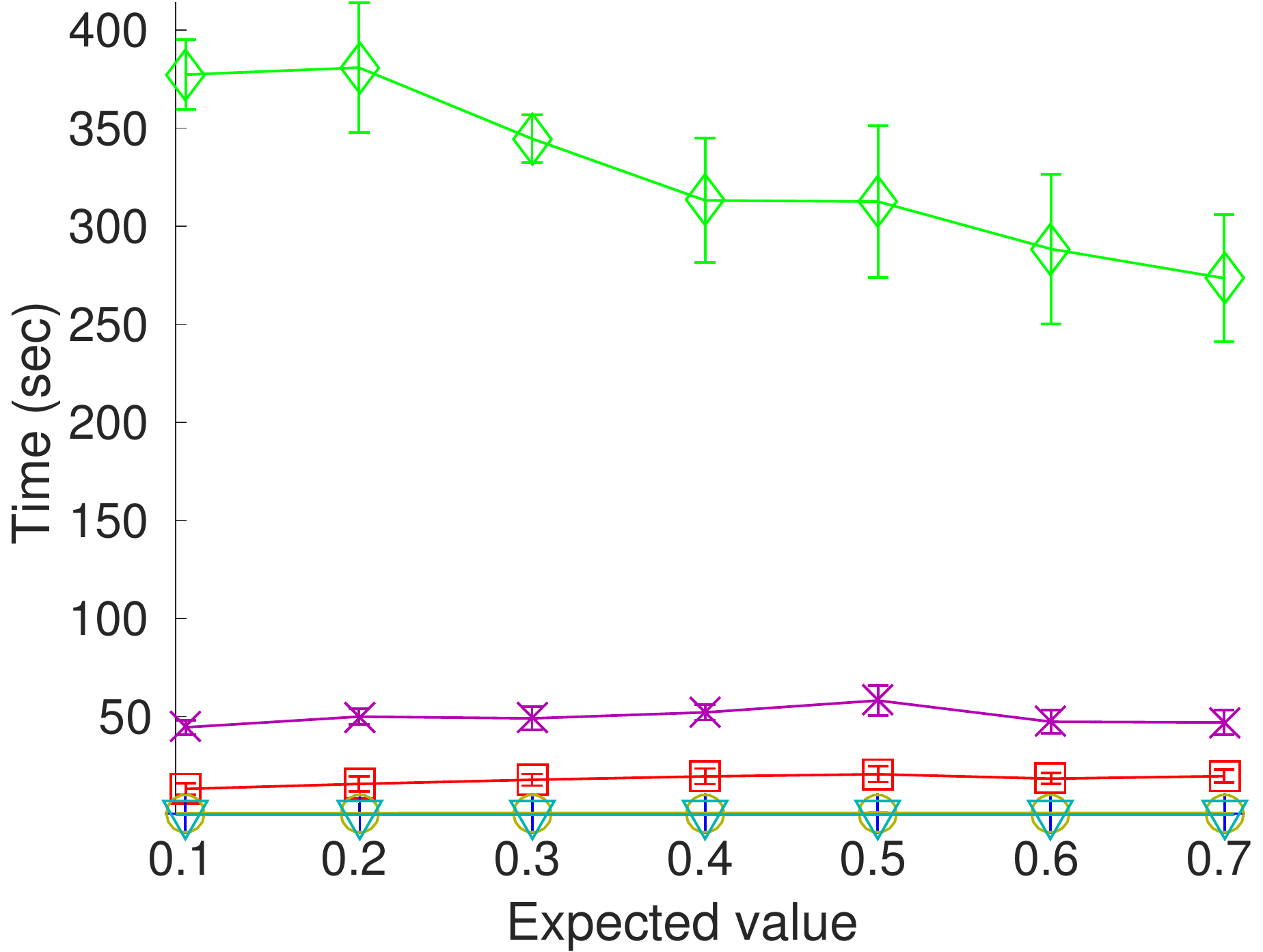}%
  }\hspace*{\smallfigsep}%
  \subfigure[Time, vary $p$]{%
    \label{fig:synth:small:noise:unif:time}%
    \includegraphics[width=\smallfigwidth]{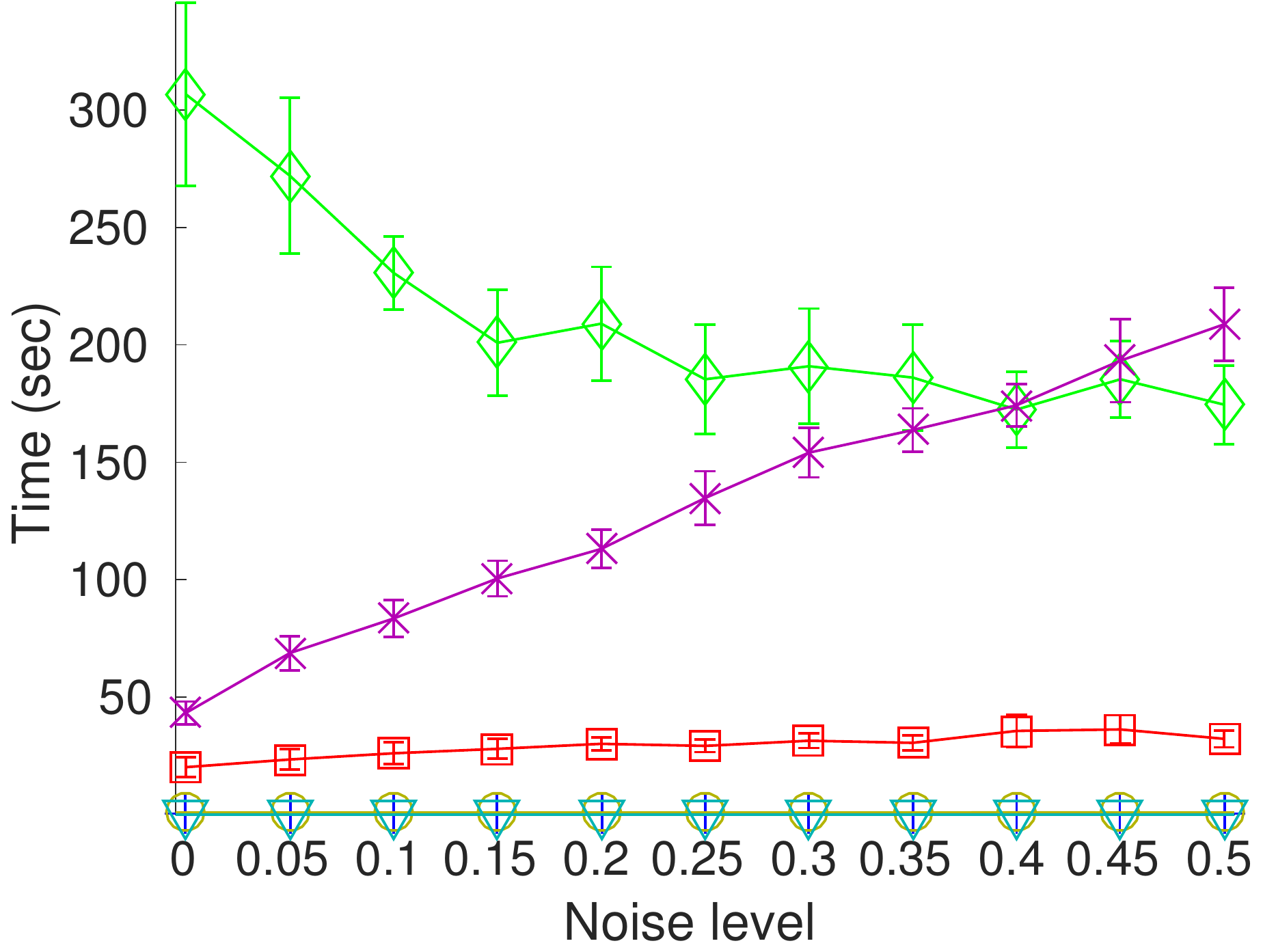}%
  }%
  \caption{Estimated rounding ranks and running times on small synthetic data varying different parameters. The top row gives estimated rank for uniformly distributed factor matrices and the bottom row shows running times. \shay can only run on square matrices and was excluded from the ``vary $m$'' experiments. All data points are averages over 10 random matrices and the width of the error bars is twice the standard deviation. }
  \label{fig:synth:small}
\end{figure*}

We used \proje, \nuclear, \SVD, and \lpca. We also used \shay in all experiments
except when we varied the number of rows (\shay only works with square
matrices). We also computed a lower bound \lb~on $\rrank_0$ using
Prop.~\ref{prop:LowerBound}. Finally, in experiments with no noise, we also plot
the planted rank (inner dimension of the factor matrices), which acts as an
\emph{upper bound} of the actual rounding rank.

As can be seen from Fig.~\ref{fig:synth:small}, the estimated lower bound is
almost always less than 3, even when the data contains significant amounts of noise. It
seems reasonable to assume that the true rounding rank of the data is therefore
closer to the upper bound of our planted rank than the estimated lower bound
given by \lb.

Of the algorithms tested here, \proje, and \shay are the only ones that aim
directly to find the rounding rank, with \shay being the only one with
approximation guarantees (albeit weak ones). Our experiments show that \shay is
not competitive to most other methods; good theoretical properties do not ensure
a good practical behavior. \proje performs much better, being typically the
second-best method.
\SVD is commonly employed in the literature, but our experiments show clearly that for computing the rounding rank, it is not recommended. 


\lpca consistently produced the smallest (i.e.~best) rank estimate but it was also the second-slowest
method. \proje, the  second-best method for
estimating the rank, was much faster. 
\SVD often produced the worst estimates, but it is also the fastest
method. 
The running times are broadly as expected: \nuclear has to solve a semidefinite programming problem, \lpca solves iteratively dense least-squares problems, \proje only needs to solve linear equations, and \SVD computes a series of orthogonal projections. 

Varying the different parameters yielded mostly expected results with the most interesting result being how little the rank and noise had effect to the results. We assume that this is (at least partially) due to the robustness of the rounding rank: increasing the noise, say, might not have increased the rounding rank of the matrix. This is clearly observed when the rank is varied (Fig.~\ref{fig:synth:small:k:unif:rank}), where \lpca actually obtains smaller rounding rank than the planted one.

\subsubsection{Minimum-error decomposition} \label{sec:exp:synth:err}

We now study the algorithms' capability to return low-error fixed-rank
decompositions. We leave out \shay and \nuclear as they only approximate
rounding rank. Instead, we add a method to compare against: \truncSVD. 
It computes the standard truncated SVD, that is, we do not apply
any rounding. \truncSVD is used for providing a baseline: in principle, the methods
that apply rounding should give better results as they utilize the added
information that the final matrix must be binary. At the same time, however, the
rounding procedure may emphasize small errors (e.g., incorrectly representing a
$1$ with $0.49$ contributes $\approx 0.26$ to the sum of squares; after
rounding, the contribution is $1$). We also tested the \asso~\cite{miettinen08discrete} algorithm for Boolean matrix factorization (BMF).
Like any BMF algorithm, \asso returns a rounding rank decomposition restricted to binary factor matrices.
The performance of \asso's approximations was so much worse than the performance
of the other methods that we decided to omit it from the results.

To compare the algorithms, we use the relative reconstruction error, that is,
the squared Frobenius norm of the distance between the data and its
representation relative to the squared norm of the data. For all method except
\truncSVD, the relative reconstruction error agrees with the absolute number of
errors divided by the number of non-zeros in the data.

\begin{figure*}[tb]
  \centering
  \includegraphics[height=\legendheight]{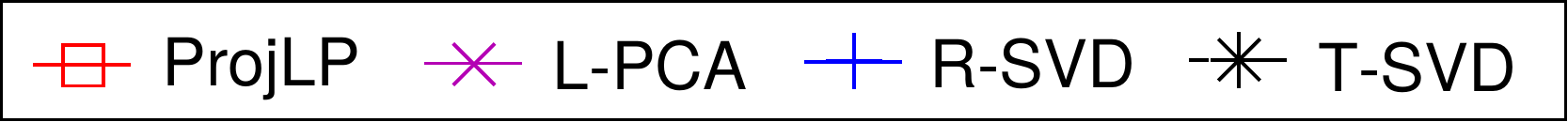} \\
  \rotatebox[origin=l]{90}{\hspace*{1em}\small Uniform dist.}\hspace*{\smallfigsep}%
  \subfigure[Error, vary $m$]{%
    \label{fig:synth:big:n:unif:err}%
    \includegraphics[width=\smallfigwidth]{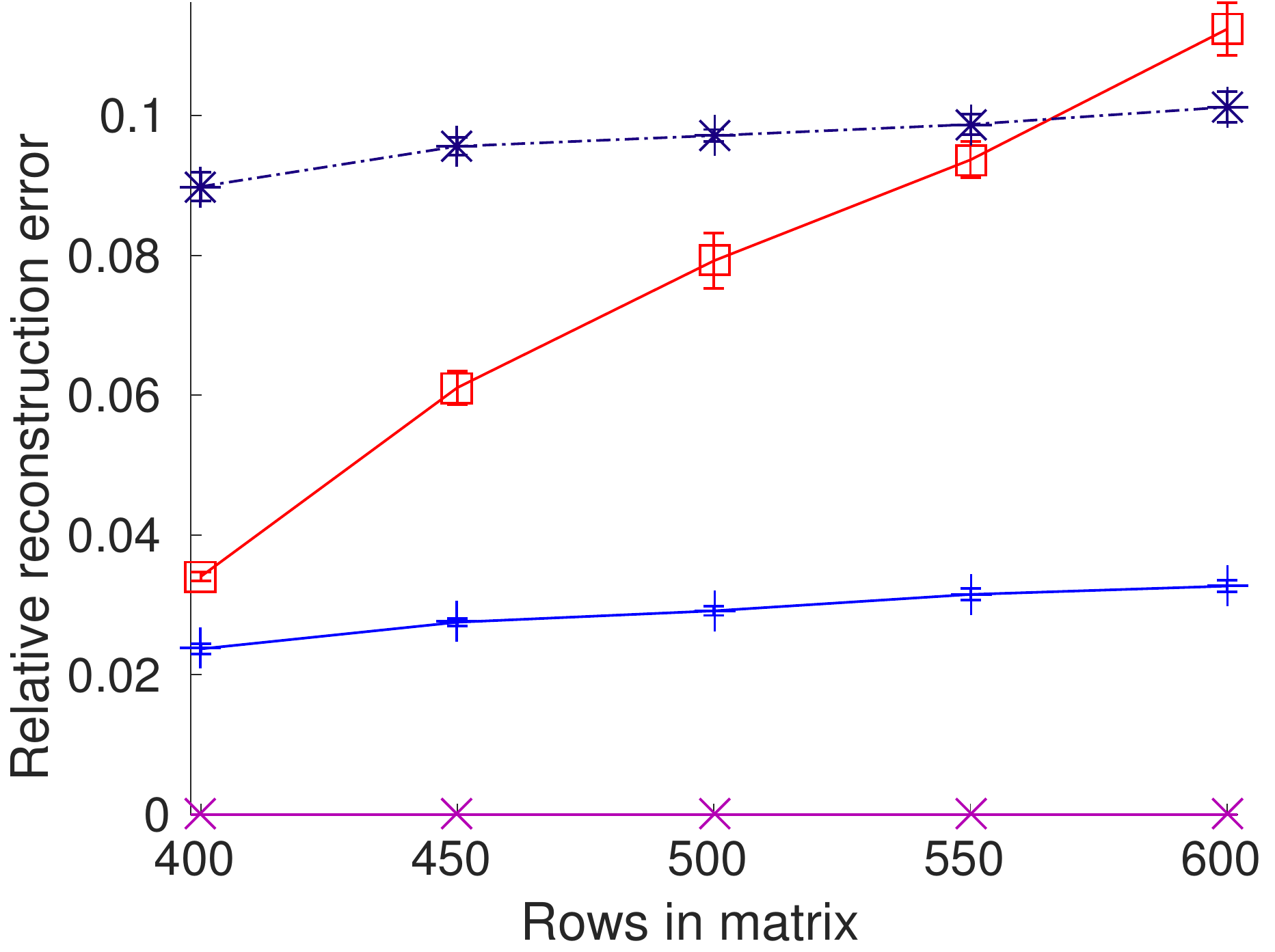}%
  }\hspace*{\smallfigsep}%
  \subfigure[Error, vary $k$]{%
    \label{fig:synth:big:k:unif:err}%
    \includegraphics[width=\smallfigwidth]{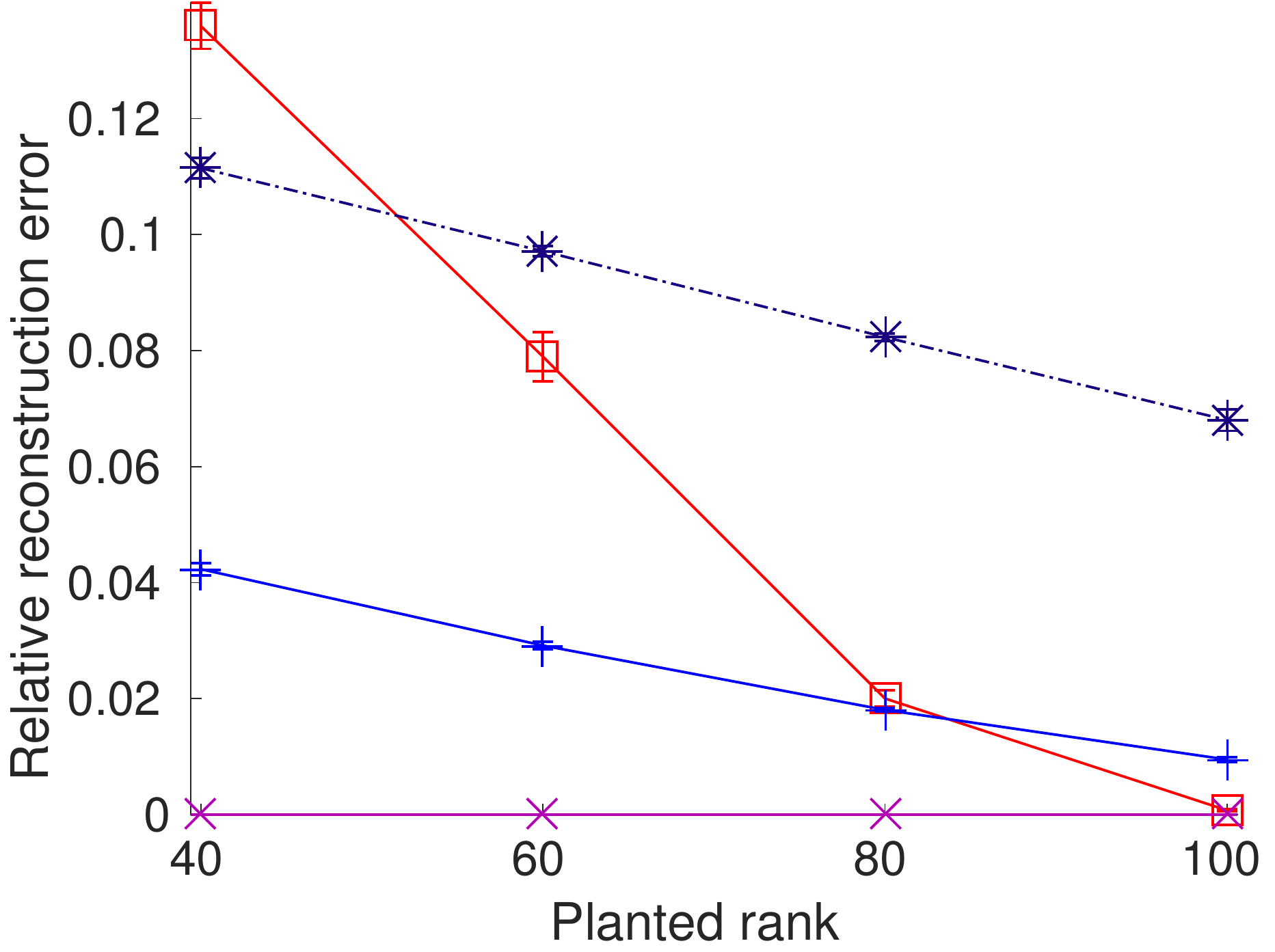}%
  }\hspace*{\smallfigsep}%
  \subfigure[Error, vary $\mu$]{%
    \label{fig:synth:big:dens:unif:err}%
    \includegraphics[width=\smallfigwidth]{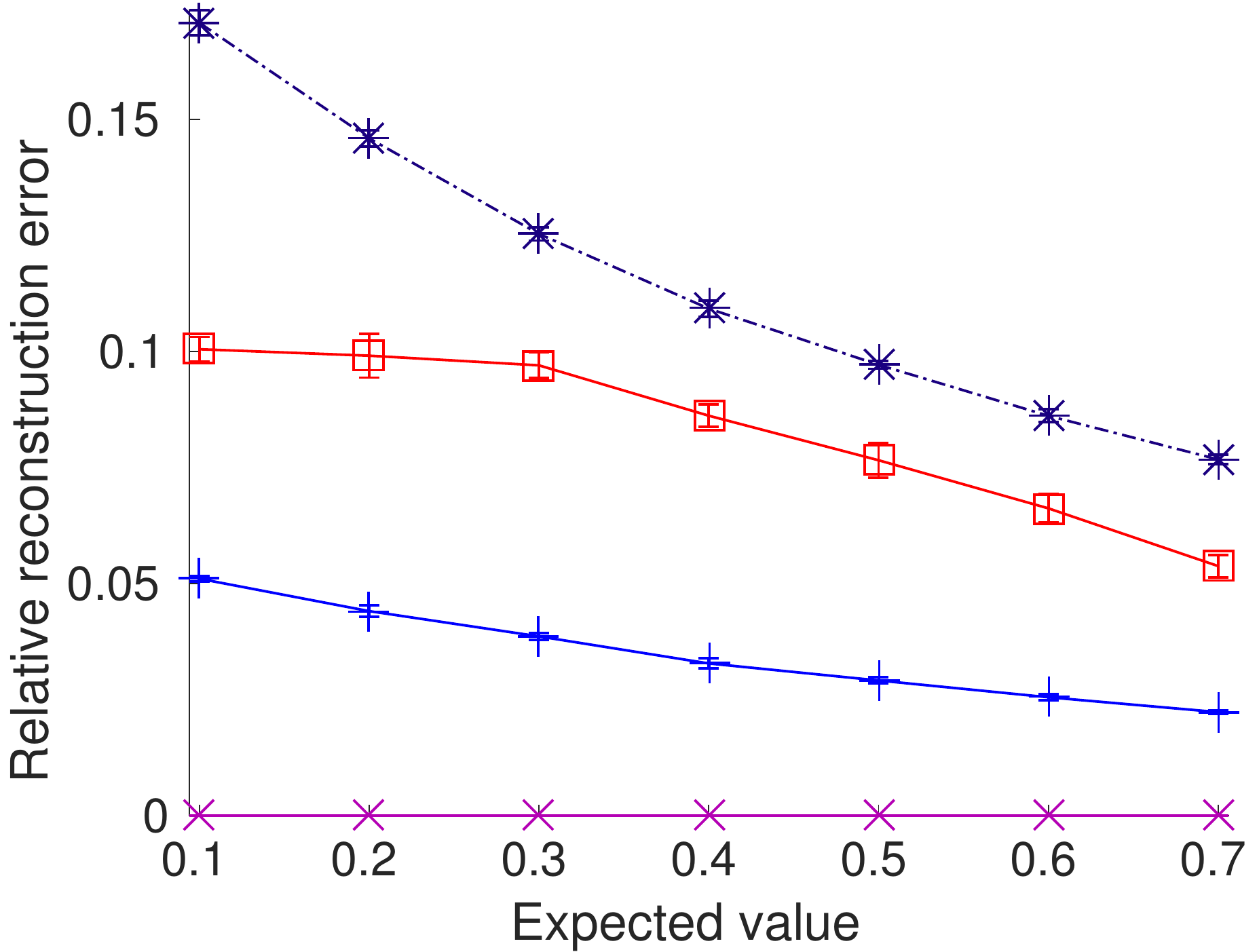}%
  }\hspace*{\smallfigsep}%
  \subfigure[Error, vary $p$]{%
    \label{fig:synth:big:noise:unif:err}%
    \includegraphics[width=\smallfigwidth]{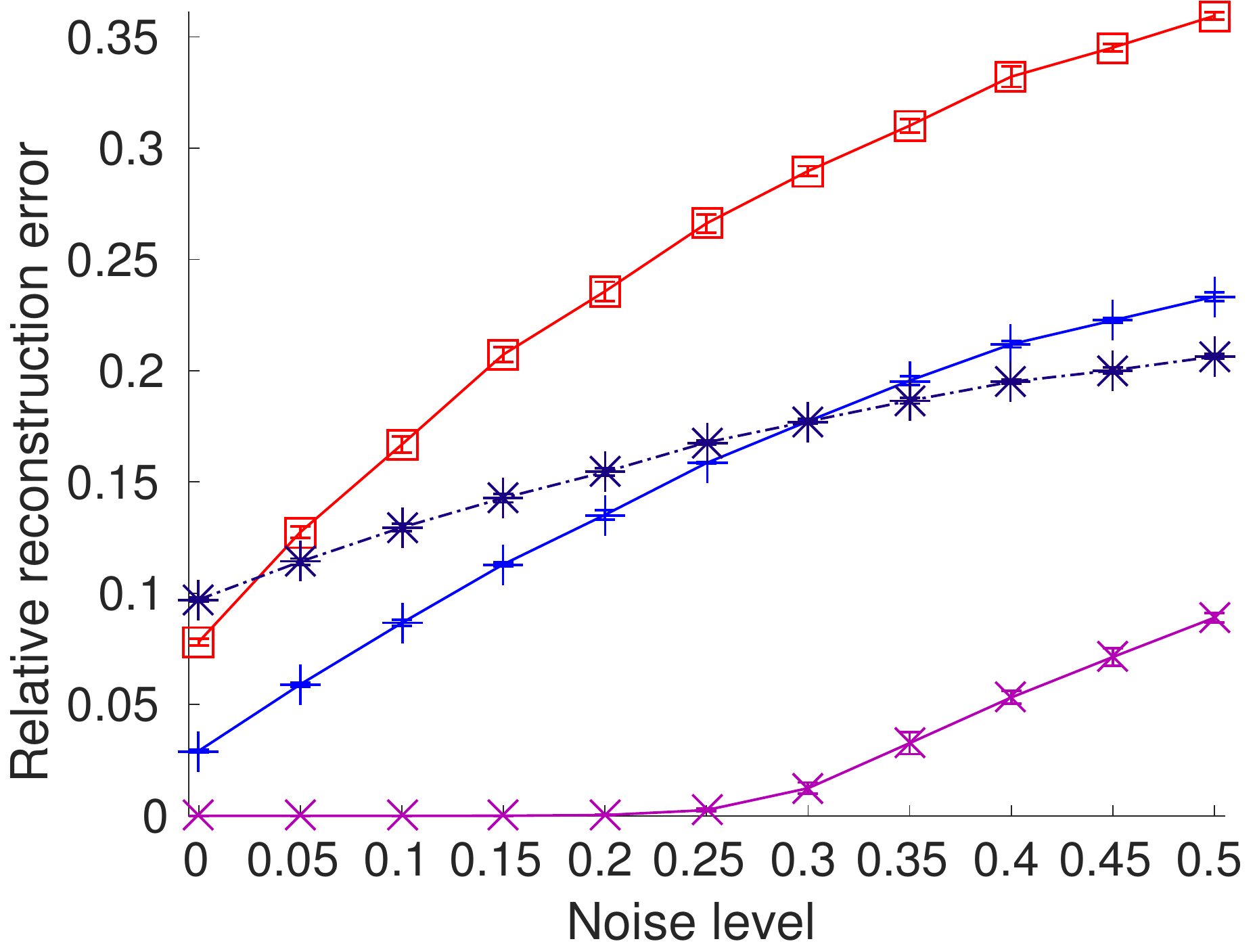}%
  }
  \caption{Relative reconstruction errors on medium-sized synthetic data with uniformly distributed factors. The results of \asso are omitted as they were significantly worse than the other results. All data points are averages over 10 random matrices and the width of the error bars is twice the standard deviation. }
  \label{fig:synth:err}
\end{figure*}

The results for these experiments are presented in Fig.~\ref{fig:synth:err}.  We
only report the reconstruction with uniformly distributed factors:
the running times were as with the above experiments, and the results with normally
distributed factors were generally similar to the reported ones. The other results are in the appendix. As in the above
experiments, \lpca is the best method, and the slowest as well, taking sometimes
an order of magnitude longer than \proje.  The best all-rounder here, though, is
the \SVD method: it provided reasonable results and was by far the fastest method.

\subsection{Results with Real-World Data} \label{sec:exp:real}

We now turn our attention to real-world datasets. For these experiments we
used only \proje, \lpca, and \SVD to estimate the rounding rank, and added
\truncSVD and \asso for the minimum-error decompositions.

\paragraph{Datasets} The basic properties of the datasets are listed in
Tab.~\ref{tab:real:rank:int}. The \abstracts data
set\footnote{\url{http://kdd.ics.uci.edu/databases/nsfabs/nsfawards.html}} is a
collection of project abstracts that were submitted to the National Science
Foundation of the USA in applications for funding. The data is
documents-by-terms matrix giving the appearance of terms in documents.  The
\dblp data\footnote{\url{http://dblp.uni-trier.de/db/}} is an
authors-by-conferences matrix containing information who published where. The
\now data set\footnote{\url{http://www.helsinki.fi/science/now/}} contains
information about the locations at which fossils of certain species were
found. It was fetched by \cite{NOWData} and preprocessed according to
\cite{NOWData2}.  The \dialect data~\cite{DialectsData1,DialectsData2} contains
information about which linguistic features appear in the dialect spoken in
various parts of Finland. The \apj dataset is a binary matrix containing access
control rules from Hewlett-Packard~\cite{RoleMining}.

\paragraph{Rounding rank} First we computed the upper bounds for the rounding
ranks with the different methods. The results and running times are shown in
Tab.~\ref{tab:real:rank:int}. As with the synthetic experiments, \lpca is
again giving the best results, followed by \proje and \SVD, the latter of which
returns often significantly worse results than the other two. In the running
times the order is reversed, \lpca taking orders of magnitude longer than
\proje, which is still slower than \SVD.

\begin{table}[tb] 
  \caption{Upper bounds for rounding rank with $\tau=0.5$ for the
    real-world data. Known Boolean ranks from~\cite{belohlavek15from-below}. \lpca did not finish on the \abstracts data in reasonable time.} 
  \label{tab:real:rank:int} \label{tab:real:rank}\label{tab:real:ranktime}
  \centering \small { \setlength{\tabcolsep}{3pt} 
    \scalebox{0.95}{
\begin{tabular}{@{}lrrrrrrr@{}} \toprule 
                      & \multicolumn{4}c{Dataset properties} & \multicolumn{3}c{Upper bounds on $\rrank$} \\
\cmidrule{2-5} \cmidrule(l){6-8}
Dataset           & $m$   & $n$  & $\rank$ & $\rankB$ & \proje        & \lpca            & \SVD         \\ \midrule
\abstracts & 12841 & 4894 & 4893    & --       & 449           & --               & 4421  \\
           &       &      &         &          & \emph{(437h)} & --               & \emph{(9h)}   \\
\apj       & 2044  & 1164 & 455     & 453      & 29            & 9                & 443          \\
           &       &      &         &          & \emph{(151s)} & \emph{(109min)}   & \emph{(35s)} \\
\dblp      & 19    & 6980 & 19      & 19       & 12            & 11               & 19           \\
           &       &      &         &          & \emph{(46s)}  & \emph{(77min)}   & \emph{(2s)}  \\
\dialect   & 1334  & 506  & 506     & --       & 91            & 78               & 445          \\
           &       &      &         &          & \emph{(527s)} & \emph{(54h)} & \emph{(17s)} \\
\now       & 124   & 139  & 123     & --       & 26            & 13               & 68           \\ 
           &       &      &         &          & \emph{(10s)}  & \emph{(271s)}    & \emph{(1s)}  \\ \bottomrule
\end{tabular}
}


  } 
\end{table}



Note that the estimated rounding
ranks in Tab.~\ref{tab:real:rank:int} are significantly less than the respective
normal or Boolean ranks. For example, for the \apj data, the normal rank is
$455$, the Boolean rank is $453$, but \lpca~shows that the rounding rank is at
most $9$. Similarly, the normal and Boolean ranks for \dblp are $19$, while the
rounding rank is no more than $11$. In most cases, the rounding rank is about an
order of magnitude smaller than the real rank. 
This shows that the expressive power of the methods significantly increases by applying the rounding.

\paragraph{Minimum-error decompositions} The relative reconstruction errors for
the real-world datasets together with running times are presented in
Tab.~\ref{tab:real:err}. Again, \lpca is often---but not always---the best
method, especially with higher ranks. Again, the running time was high
though. An exception to this is the \abstracts data, where \lpca is in fact
faster than \proje (although it is still extremely slow). Again, \proje is often
the second-best, and more consistently so with higher ranks.


\begin{table*}[t] \caption{Reconstruction errors relative to the number of
non-zeros and running times in real-world data.} \label{tab:real:err} \centering \small
{\setlength{\tabcolsep}{3pt} \scalebox{0.95}{
\begin{tabular}{@{}lrrrrrrrrrrrrrrr@{}}
\toprule 
	 & \multicolumn{3}{c}{\abstracts} & \multicolumn{3}{c}{\apj} & \multicolumn{3}{c}{\dblp} & \multicolumn{3}{c}{\dialect} & \multicolumn{3}{c}{\now} \\
	 \cmidrule{2-4}	 \cmidrule(l){5-7}	 \cmidrule(l){8-10}	 \cmidrule(l){11-13}	 \cmidrule(l){14-16}
\multicolumn{1}{r}{$k=$} & $10$ & $50$ & $100$ & $5$ & $10$ & $15$ & $5$ & $10$ & $15$ & $10$ & $30$ & $50$ & $5$ & $10$ & $20$ \\
\midrule 
\multicolumn{15}{@{}l}{\it Relative reconstruction error} \\
\proje & $1.152$ & $1.091$ & $0.842$ & $\mathbf{0.626}$ & $0.302$ & $0.099$ & $0.408$ & $0.060$ & $0.003$ & $0.378$ & $0.130$ & $0.036$ & $0.701$ & $0.360$ & $0.037$  \\
\lpca & $0.993$ & $0.863$ & $\mathbf{0.459}$ & $0.631$ & $\mathbf{0.194}$ & $\mathbf{0.034}$ & $\mathbf{0.150}$ & $\mathbf{0.003}$ & $\mathbf{0.000}$ & $\mathbf{0.200}$ & $\mathbf{0.031}$ & $\mathbf{0.002}$ & $0.552$ & $\mathbf{0.089}$ & $\mathbf{0.000}$  \\
\SVD & $0.995$ & $0.937$ & $0.843$ & $0.641$ & $0.611$ & $0.573$ & $0.488$ & $0.225$ & $0.058$ & $0.258$ & $0.137$ & $0.094$ & $0.697$ & $0.516$ & $0.260$  \\
\truncSVD & $\mathbf{0.917}$ & $\mathbf{0.838}$ & $0.766$ & $0.640$ & $0.596$ & $0.559$ & $0.382$ & $0.198$ & $0.064$ & $0.212$ & $0.120$ & $0.089$ & $\mathbf{0.515}$ & $0.410$ & $0.283$  \\
\asso & $0.988$ & $0.971$ & $0.960$ & $0.663$ & $0.637$ & $0.603$ & $0.531$ & $0.347$ & $0.187$ & $0.442$ & $0.358$ & $0.333$ & $0.793$ & $0.706$ & $0.602$  \\
\midrule
\multicolumn{15}{@{}l}{\it Running time (seconds)} \\
\proje & $9627$ & $27201$ & $72956$ & $54$ & $64$ & $63$ & $12$ & $13$ & $12$ & $26$ & $83$ & $169$ & $2$ & $2$ & $2$  \\
\lpca & $675$ & $17849$ & $39954$ & $39$ & $235$ & $297$ & $98$ & $114$ & $109$ & $96$ & $166$ & $218$ & $11$ & $12$ & $17$  \\
\SVD & $3$ & $6$ & $13$ & $1$ & $1$ & $1$ & $1$ & $1$ & $1$ & $1$ & $1$ & $1$ & $1$ & $1$ & $1$  \\
\truncSVD & $2$ & $6$ & $13$ & $1$ & $1$ & $1$ & $1$ & $1$ & $1$ & $1$ & $1$ & $1$ & $1$ & $1$ & $1$  \\
\asso & $2366$ & $11701$ & $23023$ & $8$ & $17$ & $23$ & $29$ & $52$ & $75$ & $49$ & $145$ & $238$ & $1$ & $1$ & $1$  \\
\bottomrule 
\end{tabular} 
}
 } \end{table*}




\subsection{Nestedness} \label{sec:exp:nestedness}

Here we studied the possibility to use the non-negative
rounding rank-1 decomposition to solve the \BMNA problem. For these purposes, we
generated nested matrices, perturbed them with noise, and tried to find the
closest nested matrix using \nmt, \nexhaust, their combination \nmtexhaust, and
\SVD. All nested matrices were \by{200}{300} and we varied the density of the
data (from $0.1$ to $0.7$ with steps of $0.1$) and the noise level (from $0.05$
to $0.5$ with steps of $0.05$). A default density of $\mu=0.5$ was used when the
noise was varied, and noise level $p=0.15$ was used when the density was
varied. 

Our results are shown in Fig.~\ref{fig:nested}. \nmt and \nexhaust produced
similar results, with \nmt being slightly better. The combined \nmtexhaust is no
better than \nmt, and \SVD is significantly worse. In the running times, though,
we see that \nmt takes much more time than the other approaches.

\begin{figure*}[tbp]
  \centering
  \includegraphics[height=\legendheight]{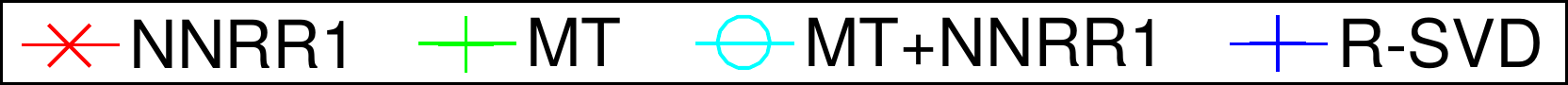}
  \subfigure[Error, vary density]{%
    \label{fig:nested:dens:err}%
    \includegraphics[width=\smallfigwidth]{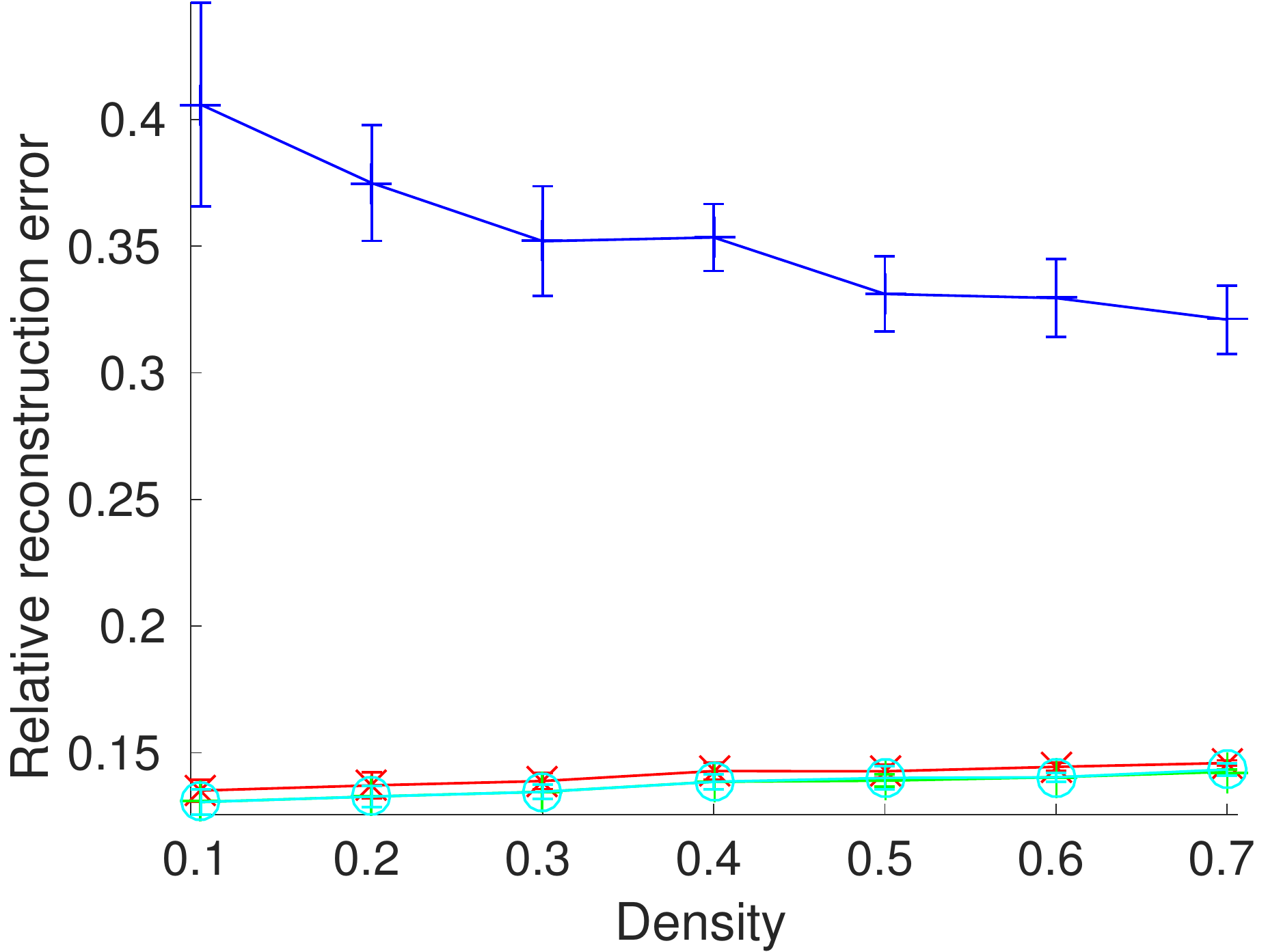}%
  }\hspace*{\smallfigsep}%
  \subfigure[Error, vary noise level]{%
    \label{fig:nested:noise:err}%
    \includegraphics[width=\smallfigwidth]{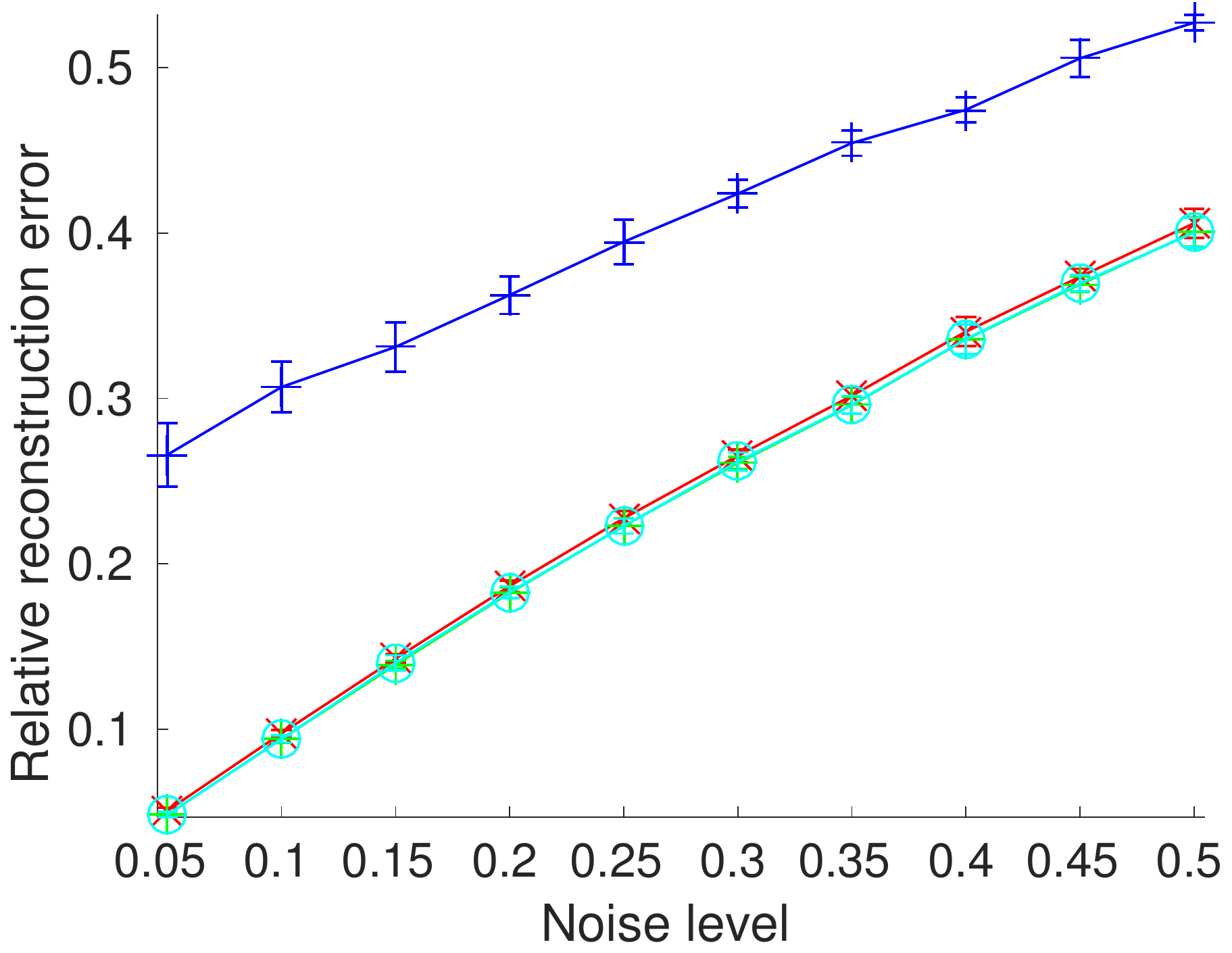}%
  }\hspace*{\smallfigsep}%
  \subfigure[Time, vary density]{%
    \label{fig:nested:dens:time}%
    \includegraphics[width=\smallfigwidth]{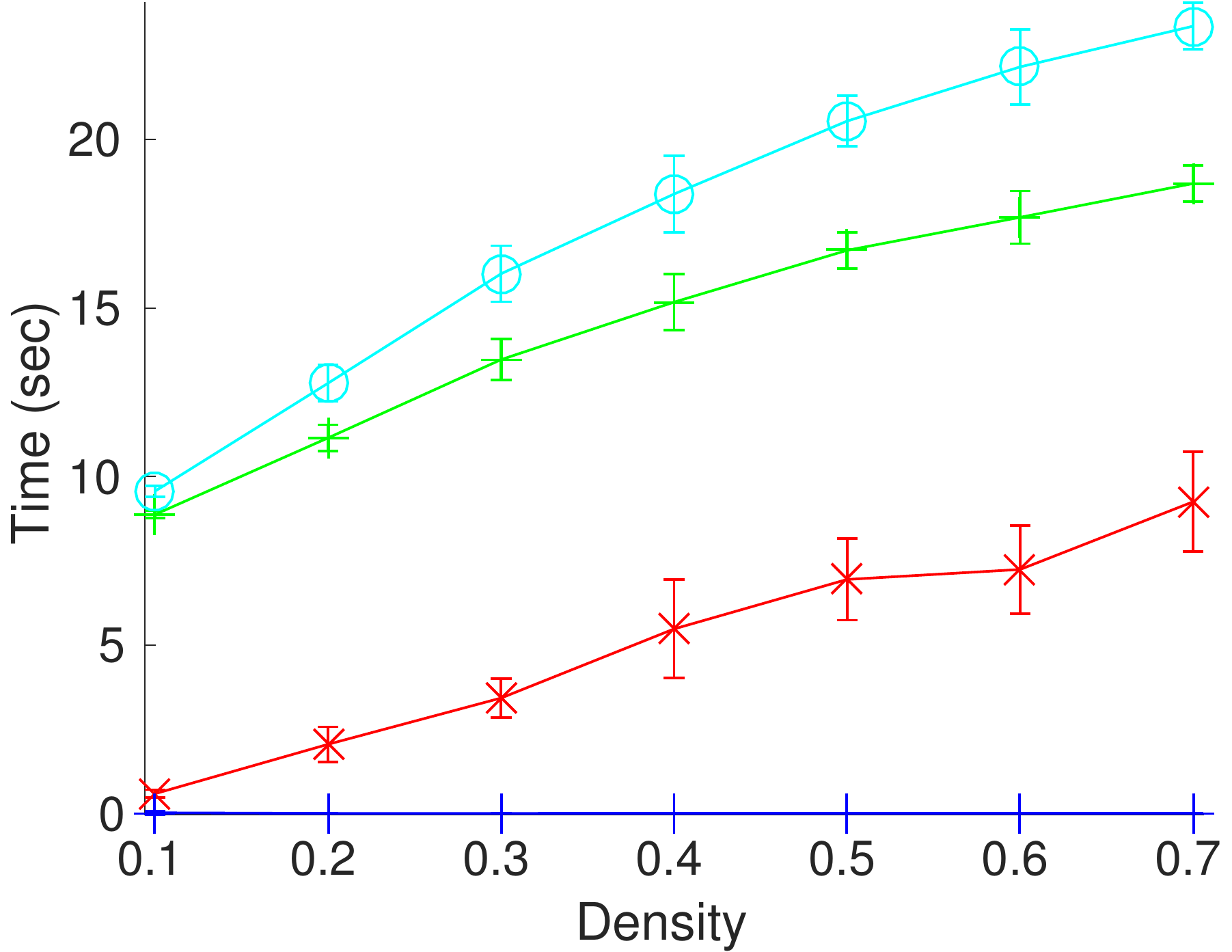}%
  }\hspace*{\smallfigsep}%
  \subfigure[Time, vary noise]{%
    \label{fig:nested:noise:time}%
    \includegraphics[width=\smallfigwidth]{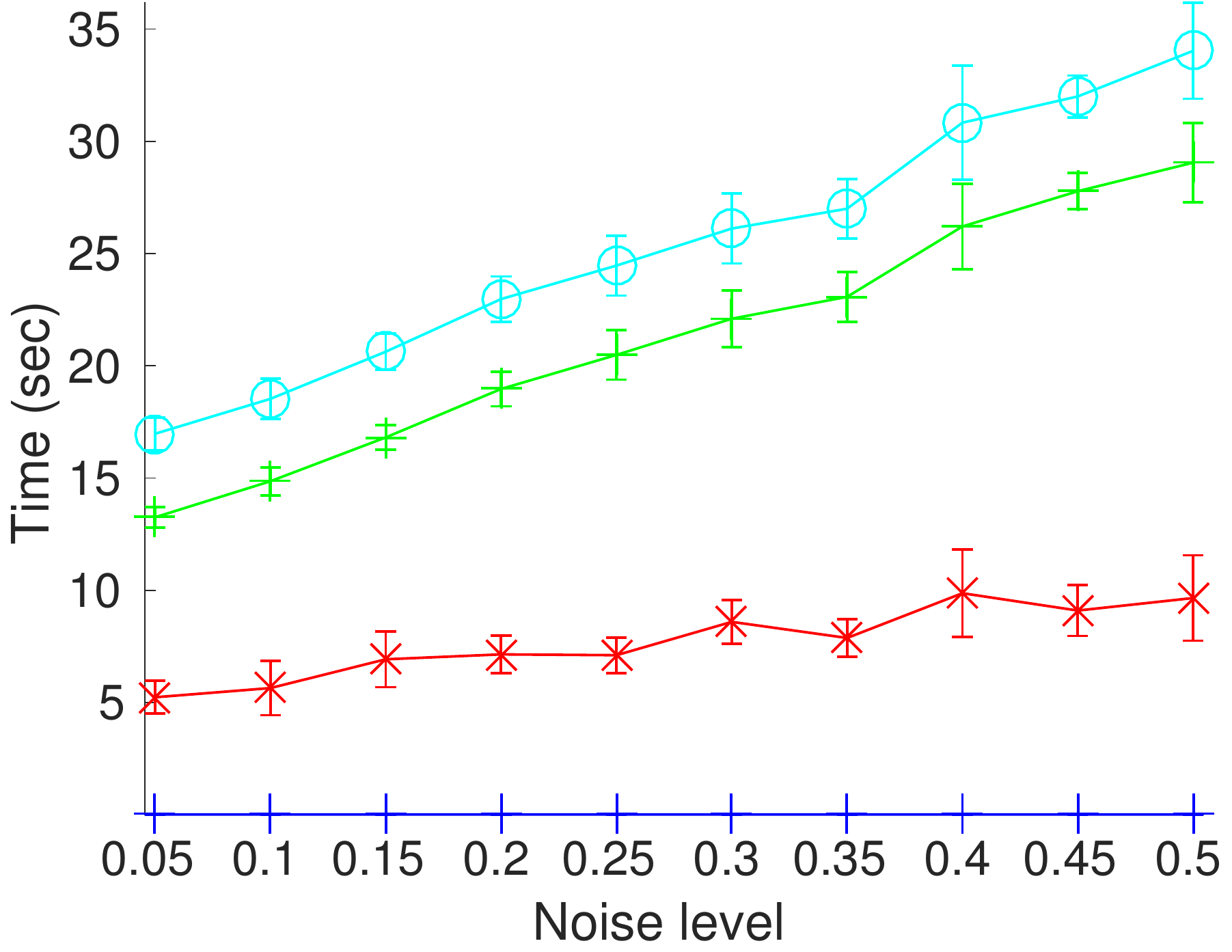}%
  }%
  \caption{Relative reconstruction errors and running times on nested data. The running times for \nmtexhaust\ exclude the running time of the \nmt algorithm. All data points are averages over 10 random matrices and the width of the error bars is twice the standard deviation. }
  \label{fig:nested}
\end{figure*}



\section{Conclusions}
\label{sec:conclusions}

Rounding rank is a natural way to characterize the commonly-applied rounding
procedure.
Rounding rank has some significant differences to real rank: for example,
restricting the factor matrices to be non-negative has almost no consequences to
rounding rank. Rounding rank provides a robust definition of an intrinsic
dimension of a data, and as we saw in the experiments, real-world data sets can
have surprisingly small rounding ranks.  At the same time, rounding rank-related
problems appear naturally in various different fields of data analysis; for
example, the connection to nested matrices is somewhat surprising, and allowed
us to develop new algorithms for the \BMNA problem.

Unfortunately, computing the rounding rank, and the related minimum-error decomposition, is computationally very hard. We have studied a
number of algorithms---based on common algorithm design paradigms in data
mining---in order to understand how well they behave in our problems. None of
the tested algorithms emerges as a clear winner, though.

The most obvious future research direction is to find better
algorithms that aim directly for good rounding rank decompositions and scale to
larger data sizes.  Another question is if the factors obtained by a rounding
rank decomposition reveal \emph{interpretable} insights into the data.
The
connections of rounding rank to other problems also propose natural follow-up
questions.
For example, communities in graphs are often nested (sub-)matrices~\cite{metzler16hyperbolae}. Could rounding rank decompositions be used to find non-clique-like communities?



\bibliographystyle{abbrv}
\bibliography{bibliography}

\appendix
\renewcommand{\theequation}{A.\arabic{equation}}
\setcounter{equation}{0}
\renewcommand{\thetheorem}{A.\arabic{theorem}}
\setcounter{theorem}{0}
\newtheorem*{lemma*}{Lemma}
\newtheorem*{proposition*}{Proposition}


We will first provide proofs of the lemmata and propositions omitted in the main
text, and then provide additional results of our experimental evaluation.

\subsection{Identity Matrices Have Rounding Rank 2}
\label{Sec:IdentityMatrices}

For $n \geq 3$, let $\mI_n \in \mathbb{R}^{n \times n}$ be the identity matrix.
From Proposition~\ref{prop:rrank1_eq_block_nested} we get that $\rrank(\mI_n) > 1$, since
the identity matrix is not nested.
We look at the matrix
\begin{align*}
	\mA &=
\begin{pmatrix}
   1      	& -\frac{1}{2}   	 	\\
   2^{-1}   & -\frac{1}{2} 4^{-1} 	\\
   \vdots 	& \vdots 			 	\\
   2^{-n+1} & -\frac{1}{2} 4^{-n+1}
\end{pmatrix}
\begin{pmatrix}
	1 & 2 & \hdots & 2^{n-1} \\
	1 & 4 & \hdots & 4^{n-1}
\end{pmatrix} \\
&=
\begin{pmatrix}
	1 - \frac{1}{2} & 2-\frac{4}{2} & 4-\frac{16}{2} & \cdots \\
	\frac{1}{2} - \frac{1}{8} & 1-\frac{4}{8} & 2-\frac{16}{8} & \cdots \\
	\frac{1}{4} - \frac{1}{32} & \frac{1}{2}-\frac{4}{32} & 1-\frac{16}{32} & \cdots \\
	\vdots & \vdots & \vdots & \ddots
\end{pmatrix},
\end{align*}
and observe that $\mA_{ij} = \frac{1}{2}$, if $i = j$, and $\mA_{ij} < \frac{1}{2}$, otherwise.
Thus, we get $\round(\mA) = \mI_n$ and therefore $\rrank(\mI_n) = 2$.

\subsection{Proof of Proposition~\ref{prop:LowerBound}}
\cite{BetterLinearLowerBound} proved a lower bound for sign rank.

\begin{lemma}[{\hspace*{-.5em}\cite[Th.~5]{BetterLinearLowerBound}}]
  Let $\mB \in \{-1, +1\}^{m \times n}$.
  Let $r = \rank(\mB)$ and let $\sigma_1(\mB) \geq \dots \geq \sigma_r(\mB) > 0$ be
  the singular values of $\mB$. Denote the sign rank of $\mB$ by $d$. Then
  \begin{align*}
    d \sum_{i = 1}^d \sigma_i^2(\mB) \geq mn.
  \end{align*}
\end{lemma}

We use the previous lemma to prove a lower bound on rounding rank.

\begin{varthm}[Proposition~\ref{prop:LowerBound} (again)]
  Let $r = \rank(\mB^\pm)$ and let $\sigma_1(\mB^\pm) \geq \dots \geq \sigma_r(\mB^\pm) > 0$ be
  the singular values of $\mB^\pm$. Then
  \begin{align*}
	  (\rrank_0(\mB) + 1) \sum_{i = 1}^{\rrank_0(\mB)} \sigma_i^2(\mB^\pm) \geq mn.
  \end{align*}
\end{varthm}
\begin{proof}
	As argued in the main text, $\rrank_0(\mB) \leq \srank(\mB)$ for all $\mB$.
	This implies
	\begin{align*}
		\sum_{i=1}^{\srank(\mB)} \sigma_i^2(\mB^\pm) \geq \sum_{i=1}^{\rrank_0(\mB)} \sigma_i^2(\mB^\pm).
	\end{align*}
	Using the previous lemma for sign rank and $\rrank_0(\mB) + 1 \geq \srank(\mB)$, we get
	\begin{align*}
    	\rrank_0(\mB) + 1 &\geq \srank(\mB) \\
			&\geq \frac{mn}{\sum_{i=1}^{\srank(\mB)} \sigma_i(\mB^\pm)}\\
			&\geq \frac{mn}{\sum_{i=1}^{\rrank_0(\mB)} \sigma_i(\mB^\pm)}.
	\end{align*}
	After multiplying with the denominator of the last equation, we obtain the desired result.
\end{proof}

\subsection{Proof of Lemma~\ref{lem:UsefulHyperplaneSeparation}}
We revisit the Hyperplane Separation Theorem.

\begin{lemma*}[Hyperplane Separation Theorem {\cite[page 46]{ConvexOptimization}}]
Let $A$ and $B$ be two disjoint nonempty closed convex sets in $\mathbb{R}^d$, one of which is compact.
Then there exists a nonzero vector $\vv \in \mathbb{R}^d$ and real numbers $c_1 < c_2$,
such that $\langle \vx, \vv\rangle > c_2$ and $\langle \vy, \vv \rangle < c_1$
for all $\vx \in A$ and $\vy \in B$.
\end{lemma*}

Now we prove Lemma~\ref{lem:UsefulHyperplaneSeparation} of the paper.

\begin{varthm}[Lemma~\ref{lem:UsefulHyperplaneSeparation} (again)]
Let $A$ and $B$ be two disjoint nonempty convex sets in $\mathbb{R}^d$, one of which is compact.
Then for all $c \in \mathbb{R} \setminus \{ 0 \}$ there exists a nonzero vector $\vv \in \mathbb{R}^d$,
such that $\langle \vx, \vv \rangle > c$ and $\langle \vy, \vv \rangle < c$ for all $\vx \in A$ and $\vy \in B$.
\end{varthm}
\begin{proof}
Let $c \in \mathbb{R} \setminus \{ 0 \}$ be arbitrary.
We apply the Hyperplane Separation Theorem to $A$ and $B$
to obtain a vector $\vv'$ and numbers $c_1 < c_2$
with $\langle \vx, \vv' \rangle > c_2$ and $\langle \vy, \vv' \rangle < c_1$
for all $\vx \in A$ and all $\vy \in B$.
Now we consider three cases.

Case 1: $c_1 \neq 0$ and $c_2 \neq 0$ and $\sign(c_1) = \sign(c_2)$.
We set $\alpha = c / c_2$ and $\vv = \alpha \vv'$.
Then we get for $\vx \in A$:
\begin{align*}
\langle \vx, \vv \rangle = \alpha \langle \vx, \vv' \rangle > \alpha c_2 = c,
\end{align*}
as well as for $\vy \in B$:
\begin{align*}
\langle \vy, \vv \rangle = \alpha \langle \vy, \vv' \rangle < \alpha c_1 = \frac{c_1}{c_2} c < c,
\end{align*}
where in the last inequality we used that $0 < c_1/c_2 < 1$.

Case 2: $c_1 \neq 0$ and $c_2 \neq 0$ and $\sign(c_1) \neq \sign(c_2)$.
Since the signs of $c_1$ and $c_2$ disagree, we have $c_1 < 0 < c_2$.
Thus, we can pick $c_1' \in (0,c_2)$ arbitrarily
and still maintain all properties guaranteed by the Hyperplane Separation Theorem
for $\vv$, $c_1'$ and $c_2$.
Now we are in case 1.

Case 3: $c_1 = 0$ or $c_2 = 0$.
We can pick numbers $d_1, d_2 \in (c_1, c_2)$ with $d_1 < d_2$.
Observe that both $d_1$ and $d_2$ are non-zero.
Then we have $\langle \vx, \vv' \rangle > c_2 > d_2$
and $\langle \vy, \vv' \rangle < c_1 < d_1$
for all $\vx \in A$ and all $\vy \in B$.
Now we can use case 1 for $\vv'$, $d_1$ and $d_2$.
\end{proof}

\subsection{Proof of Proposition~\ref{prop:rrankp_vs_rrank}}
The proof follows ideas of \cite{ProbabilisticCommunicationComplexity} and adds some details
	to achieve the non-negativity.

\begin{varthm}[Proposition~\ref{prop:rrankp_vs_rrank} (again)]
  Let $\mB \in \{0, 1\}^{m \times n}$ be a binary matrix. Then
  \begin{align*}
    \rrank(\mB) \leq \rrank_+(\mB) \leq \rrank(\mB) + 2.
  \end{align*}
\end{varthm}
\begin{proof}
  The first inequality is trivial as the standard rounding rank
  is more general than the non-negative rounding rank.

  The trickier part is the second inequality.
  The idea of the proof is to take points and hyperplanes achieving the rounding rank of the matrix
  and to project them into a higher dimensional space, where they are non-negative.
  This projection happens via an explicit construction, that gets somewhat technical.

  Let $k = \rrank(\mB)$.
  Then by definition there exist matrices $\mL \in \mathbb{R}^{m \times k}$ and $\mR \in \mathbb{R}^{n \times k}$
  	with $\mB = \round(\mL \mR^T)$ for rounding threshold $\frac{1}{2}$. As before we will interpret the rows
	$\vl_1, \dots, \vl_m$ of $\mL$ as points in $\mathbb{R}^k$ and the rows $\vr_1, \dots, \vr_n$ of $\mR$
	as normal vectors of affine hyperplanes in $\mathbb{R}^k$.

  For each $\vr_j = \left(\vr_{j1}, \dots, \vr_{jk}\right)$, we set
    $\vr_j' = \left(\vr_{j1}, \dots, \vr_{jk}, - \frac{1}{2}, \frac{1}{2} - \sum_{m = 1}^k \vr_{jm}\right) \in \mathbb{R}^{k+2}$
    and observe that these vectors define hyperplanes in $\mathbb{R}^{k+2}$ containing the origin, i.e.\ we have
	$0 \in \left\{ \vx \in \mathbb{R}^{k+2} : \langle \vx, \vr_j' \rangle = 0 \right\}$.
    We set $d_j = \max\{ | \vr_{jm}' | : m = 1, \dots, k+2 \}$ and define $\vr_j'' = \frac{1}{2d_j} \vr_j'$.
	Observe that for all $m = 1, \dots, k+2$, we have $-\frac{1}{2} \leq \vr_{jm}'' \leq \frac{1}{2}$.

  For each $\vl_i = \left(\vl_{i1}, \dots, \vl_{ik} \right)$, we set $c_i = \max\{ |\vl_{i1}|, \dots, |\vl_{ik}|, 1 \}$ and
	we further define $\vl_i' = (c_i + \vl_{i1}, \dots, c_i + \vl_{ik}, c_i + 1, c_i) \in \mathbb{R}^{k+2}$
	and observe that
	$\vl_i'$ is non-zero and non-negative. By $\vl_i''$ we denote $\vl_i'$ after normalizing with the $L^1$-norm, i.e.\ 
	$\vl_i'' = \vl_i' / || \vl_i' ||_1$, where $|| \vl_i' ||_1 = \sum_{m = 1}^{k+2} | \vl_{im} | $.


  Now we do a short intermediate computation that shows that the
  $\vl_i''$ and $\vr_j''$ indeed still round to matrix $\mB$
  with rounding threshold $0$:
  \begin{align}
    \langle \vr_j'', \vl_i'' \rangle
    &= \frac{1}{||\vl_i'||_1} \langle \vr_j'', \vl_i' \rangle \nonumber \\
    &= \frac{1}{2d_j ||\vl_i'||_1} \langle \vr_j', \vl_i' \rangle \nonumber \\
    &= \frac{1}{2d_j ||\vl_i'||_1} \left( \sum_{m=1}^k \vr_{jm} (c_i + \vl_{im}) - \frac{1}{2}(c_i + 1) \right. \nonumber \\
	   & \hspace{1cm} \left. + \left(\frac{1}{2} - \sum_{m = 1}^k \vr_{jm}\right)c_i \right) \nonumber \\
    &= \frac{1}{2d_j ||\vl_i'||_1} \left( \sum_{m=1}^k \vr_{jm} \vl_{im} - \frac{1}{2} \right) \nonumber \\
    &= \frac{1}{2d_j ||\vl_i'||_1} \left( \langle \vr_j, \vl_i \rangle - \frac{1}{2} \right) \label{eq:technicalThingy1} \\
    &=
      \begin{cases}
  	    \geq 0, & \text{if $\langle \vr_j, \vl_i \rangle \geq \frac{1}{2}$}, \\
  	    < 0, & \text{otherwise}.
  	  \end{cases} \nonumber
  \end{align}

  We move on to define $\vr_j''' \in \mathbb{R}^{k+2}$ by setting
  $\vr_{jl}''' = \frac{1}{2} + \vr_{jl}''$ for all $l = 1, \dots, k+2$.
  Observe that each component of $\vr_j'''$ is non-negative.
  We perform another intermediate computation, that we will need later:
  \begin{align}
    \langle \left(\frac{1}{2}, \dots, \frac{1}{2}\right), \vl_i'' \rangle \nonumber
    &= \frac{1}{2} \sum_{m=1}^{k+2} \vl_{im}'' \nonumber \\ 
    &= \frac{1}{2 ||\vl_i'||_1} \sum_{m=1}^{k+2} \vl_{im}' \nonumber \\
    &= \frac{1}{2 ||\vl_i'||_1} \left( \sum_{m=1}^{k} (c_i + \vl_{im}) + 2c_i + 1 \right) \nonumber \\
    &= \frac{1}{2 ||\vl_i'||_1} \left( (k+2) c_i + 1 + \sum_{m=1}^{k} \vl_{im} \right). \label{eq:technicalThingy2}
  \end{align}

  Now we observe that the $\vr_j'''$ and $\vl_i''$ give a non-negative rounding rank
  decomposition of $\mB$ for different rounding thresholds,
  where we use \eqref{eq:technicalThingy1} and \eqref{eq:technicalThingy2} in the second step:
  \begin{align}
    \langle \vr_j''', \vl_i'' \rangle \nonumber
    &= \langle \vr_j'' + \left( \frac{1}{2}, \dots, \frac{1}{2} \right), \vl_i'' \rangle \nonumber \\
    &= \frac{ \langle \vr_j, \vl_i \rangle - \frac{1}{2} }{2d_j ||\vl_i'||_1}
		 + \frac{(k+2) c_i + 1 + \sum_{m=1}^{k} \vl_{im}}{2 ||\vl_i'||_1}. \label{eq:techroundingthresh}
  \end{align}

  Notice that the first summand of~\eqref{eq:techroundingthresh} is non-negative
  	iff $\langle \vr_j, \vl_i \rangle \geq \frac{1}{2}$.
  Thus, if we use the second summand as rounding threshold, then we would round correctly.
  The issue is that this rounding threshold depends on $\vl_i$.

  To solve this problem and to get everything to rounding threshold $\frac{1}{2}$, we rescale the $\vl_i''$.
    We denote the second summand of~\eqref{eq:techroundingthresh} by $\alpha$ and observe that $\alpha \geq 0$
	by choice of $c_i$. Now we set $\vl_i''' = \frac{1}{2 \alpha} \vl_i''$ and obtain:
    \begin{align}
      \langle \vr_j''', \vl_i''' \rangle
      &= \frac{1}{2\alpha} \langle \vr_i''', \vl_i'' \rangle \nonumber \\
      &= \frac{ \langle \vr_j, \vl_i \rangle - \frac{1}{2} }{4 \alpha d_j ||\vl_i'||_1}
	  			+ \frac{1}{2}, \label{eq:finaltechnicality}
    \end{align}
    where we used~\eqref{eq:techroundingthresh} in the last step.
    The inner product is non-negative by choice of $\vl_i'''$ and $\vr_j'''$
    and the first summand of~\eqref{eq:finaltechnicality} is non-negative
	  iff $\langle \vr_j, \vl_i \rangle \geq \frac{1}{2}$.
    Thus, $\langle \vr_j''', \vl_i''' \rangle \geq \frac{1}{2}$
	  iff $\langle \vr_j, \vl_i \rangle \geq \frac{1}{2}$
	  iff $\mB_{ij} = 1$.
    Therefore, the $\vr_j'''$ and $\vl_i'''$ give a non-negative rounding rank decomposition of $\mB$ for rounding threshold $\frac{1}{2}$.
\end{proof}

\subsection{Proof of Proposition~\ref{prop:rrank1_eq_block_nested}}

\begin{varthm}[Proposition~\ref{prop:rrank1_eq_block_nested} (again)]
  Let $\mB \in \{ 0, 1\}^{m \times n}$ with $\mB \neq 0$. The following statements are equivalent:
  \begin{enumerate}
    \item $\rrank(\mB) = 1$.
	\item $\mB$ is nested or there exist permutation matrices $\mP_1$ and $\mP_2$
		  and nested matrices $\mB_1$ and $\mB_2$, such that
	      \begin{align*}
			  \mB = \mP_1
			  \begin{pmatrix}
				\mB_1 & 0     \\
				0     & \mB_2
			  \end{pmatrix}\mP_2.
		  \end{align*}
  \end{enumerate}
\end{varthm}
\begin{proof}
  $1 \Rightarrow 2$:
  Let $\mB=\round(\vl\vr^T)$. If $\vl$ (or $\vr$) is non-negative or non-positive, 
  $\mB$ is nested. To see this, observe that $\round(\vl \vr^T)$ remains unmodified
  if we replace entries of opposite sign in $\vr$ (or $\vl$) by $0$ and then
  take absolute values. Then we can apply Theorem~\ref{Thm:NestedMatrices}.

  Otherwise, both $\vl$ and $\vr$ contain both strictly negative and
  strictly positive entries.  Then there exists some permutation matrix $\mP_1$, such that
  $\mP_1^{-1}\vl$ is non-increasing. We pick the vectors $\vl_+\ge 0$ and
  $\vl_-\le 0$, such that $\mP_1^{-1}\vl=\begin{pmatrix}\vl_+ \\ \vl_- \end{pmatrix}$.
  Similarly, there is some permutation matrix $\mP_2$, such that
  $\mP_2\vr$ is non-increasing and we set $\vr_+$ and $\vr_-$ accordingly.

  Using this notation we can do a quick computation,
  \begin{align*}
    \mB 
    &= \round(\vl\vr^T) \\
	&= \round(\mP_1(\mP_1^T\vl)(\mP_2\vr)^T\mP_2) \\
    &= \mP_1 \round\left(
				\begin{pmatrix}
				  \vl_+ \\
				  \vl_-
				\end{pmatrix}
				\begin{pmatrix}
				  \vr_+ \\ 
				  \vr_-
				\end{pmatrix}^T
			  \right)
	    \mP_2  \\
    &= \mP_1 
       \begin{pmatrix}
	   		\round(\vl_+\vr_+^T) & \round(\vl_+\vr_-^T) \\
      		\round(\vl_-\vr_+^T) & \round(\vl_-\vr_-^T)
       \end{pmatrix}^T
	   \mP_2 \\
	&=  \mP_1
		\begin{pmatrix}
      		\mB_1 & 0     \\
			0     & \mB_2
    	\end{pmatrix}
		\mP_2,
  \end{align*}
  where $\mB_1=\round(\vl_+\vr_+^T)$ and $\mB_2=\round(\vl_-\vr_-^T)=\round((-\vl_-)(-\vr_-^T))$.
  The last equality holds since $\round(\vl_+\vr_-^T)=\bm0$ and
  $\round(\vl_-\vr_+^T)=\bm0$. Finally, we observe that $\mB_1$ and $\mB_2$ are
  nested matrices by Theorem~11.

  $2 \Rightarrow 1$:
  If $\mB$ is nested, then $\rrank(\mB) \leq \rrank_+(\mB) = 1$ by Theorem~11.
  Suppose $\mB$ is not nested and we are given $\mP_1$, $\mP_2$, $\mB_1$ and $\mB_2$
  as in the statement of the Lemma.
  Then $\mB_1$ and $\mB_2$ are non-zero (otherwise $\mB$ would be nested) and they have non-negative
  rounding rank one by Theorem~11. Thus, we can assume that $\mB_1 = \round(\vl_1 \vr_1^T)$
  and $\mB_2 = \round(\vl_2 \vr_2^T)$ for some non-negative vectors $\vl_1$, $\vl_2$, $\vr_1$ and $\vr_2$.

  Now we observe that
  \begin{align*}
	  \round \left(
		\begin{pmatrix}
		  \vl_1 \\ -\vl_2
		\end{pmatrix} 
		\begin{pmatrix}
		  \vr_1 \\
		  -\vr_2
		\end{pmatrix}^T
		\right)
	  = 
	  \begin{pmatrix}
		\mB_1 & 0     \\
		0     & \mB_2
	  \end{pmatrix}.
  \end{align*}
  Thus, by setting $\vl = \mP_1 \begin{pmatrix} \vl_1 \\ -\vl_2 \end{pmatrix}$ and
  $\vr = \mP_2^T \begin{pmatrix} \vr_1 \\ -\vr_2 \end{pmatrix}$ we get $\mB = \round(\vl \vr^T)$.
\end{proof}

\subsection{Experimental Results}
\label{sec:experimental-results}

The results on estimating the rounding rank on small synthetic data with normally distributed factors are presented in Figure~\ref{fig:apx:synth:rank:times}. The results on the minimum error fixed rounding rank experiments with medium-sized, normally distributed data are presented in Figures~\ref{fig:apx:synth:unif:time} (for timing results with uniformly-distributed factors) and~\ref{fig:apx:synth:err} (for results with normally distributed factors). In all cases, the results are essentially similar to the corresponding results presented in the main paper. 

\begin{figure*}
  \centering
  \includegraphics[height=\legendheight]{rank_small_legend} \\
  \rotatebox[origin=l]{90}{\hspace*{1em}\small Normal dist.}\hspace*{\smallfigsep}%
  \subfigure[Rank, vary $m$]{%
    \label{fig:synth:small:n:norm:rank}%
    \includegraphics[width=\smallfigwidth]{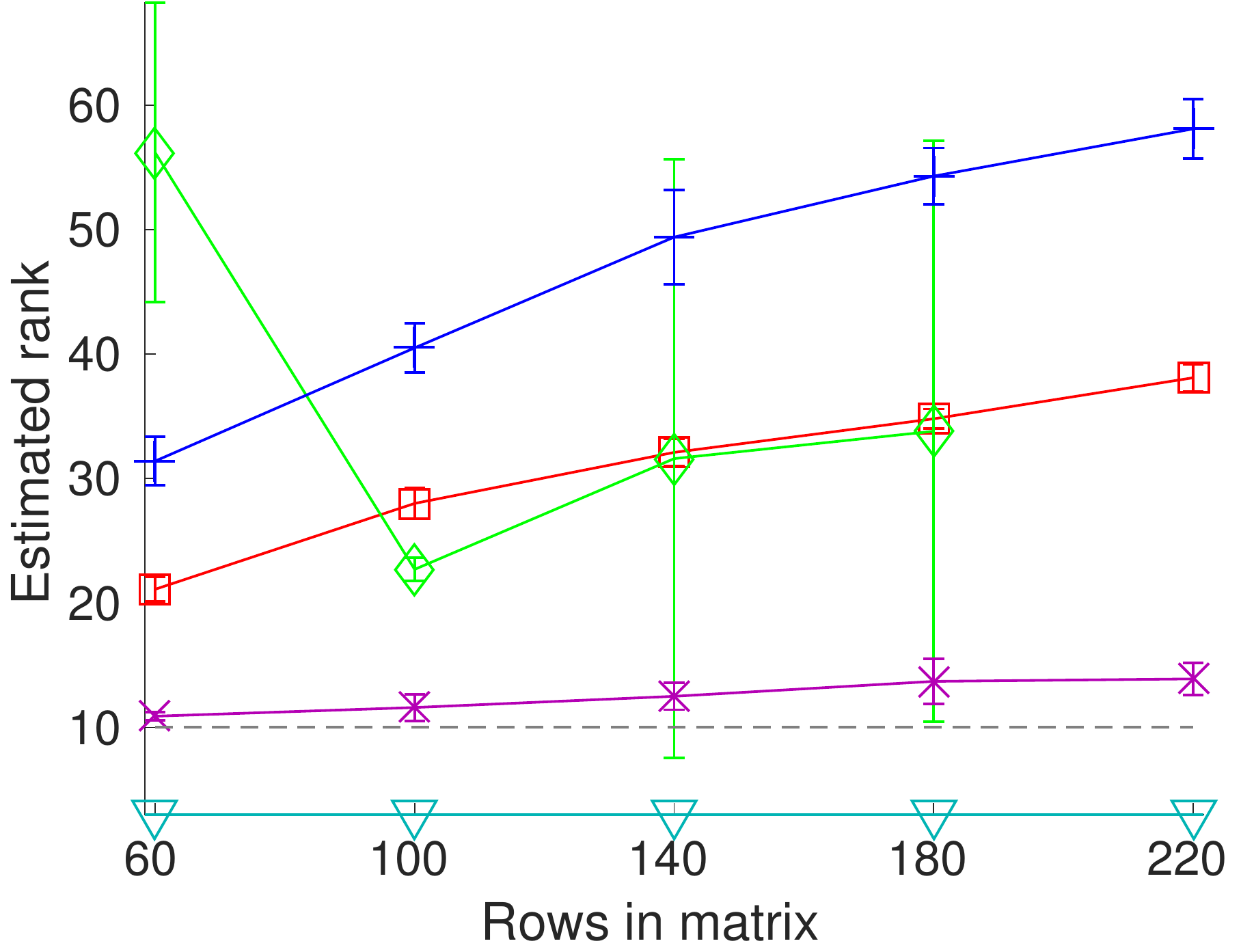}%
  }\hspace*{\smallfigsep}%
  \subfigure[Rank, vary $k$]{%
    \label{fig:synth:small:k:norm:rank}%
    \includegraphics[width=\smallfigwidth]{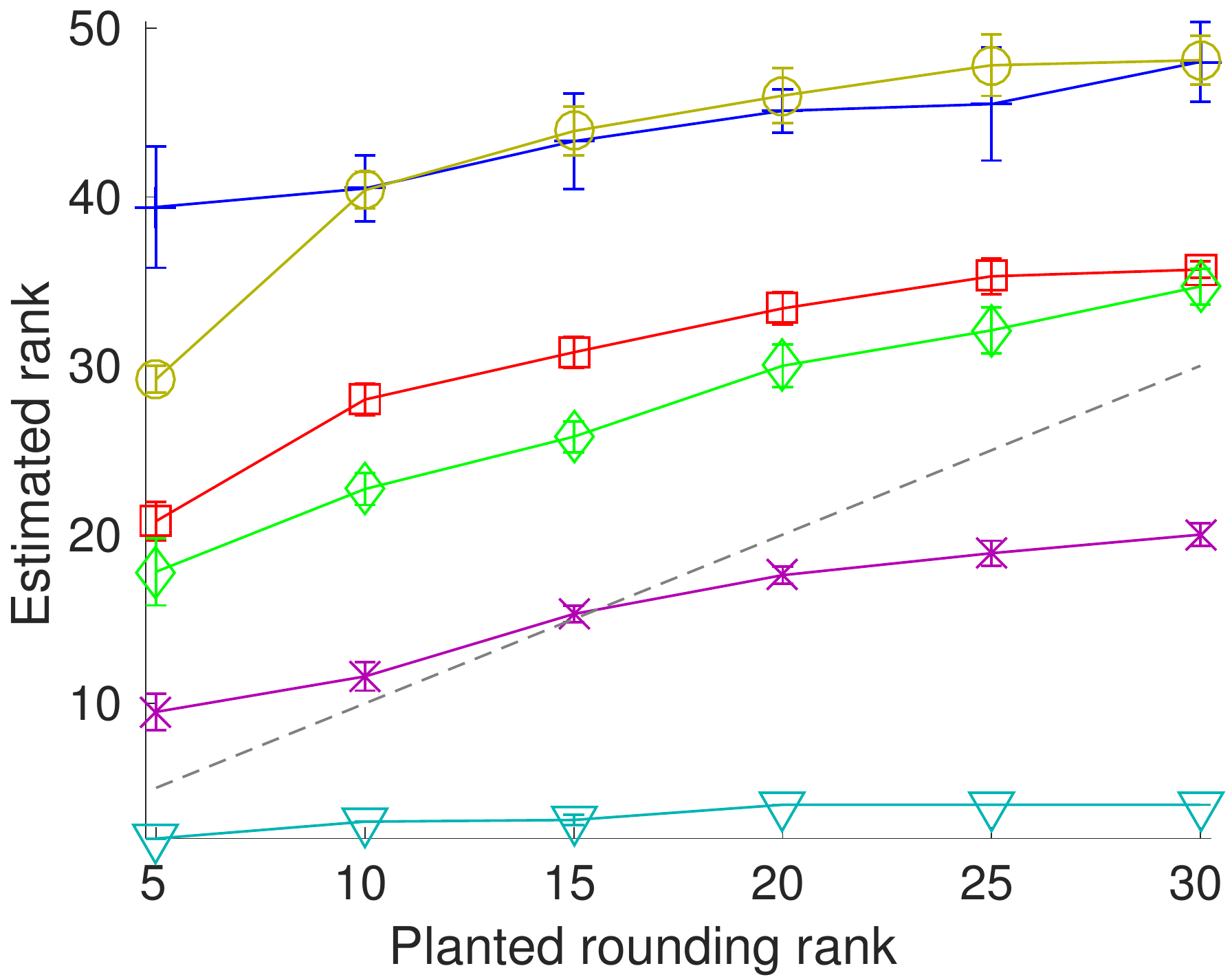}%
  }\hspace*{\smallfigsep}%
  \subfigure[Rank, vary $\mu$]{%
    \label{fig:synth:small:dens:norm:rank}%
    \includegraphics[width=\smallfigwidth]{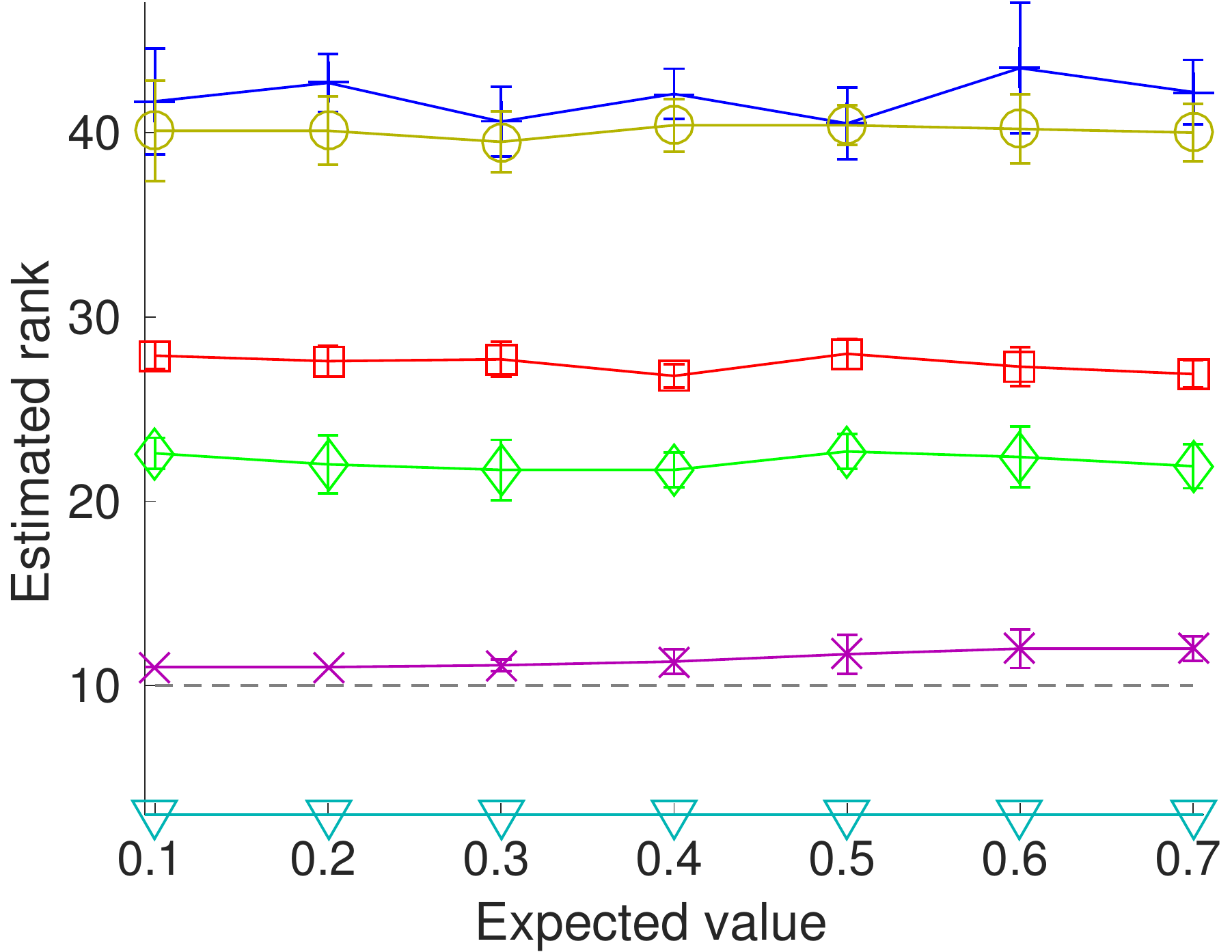}%
  }\hspace*{\smallfigsep}%
  \subfigure[Rank, vary $p$]{%
    \label{fig:synth:small:noise:norm:rank}%
    \includegraphics[width=\smallfigwidth]{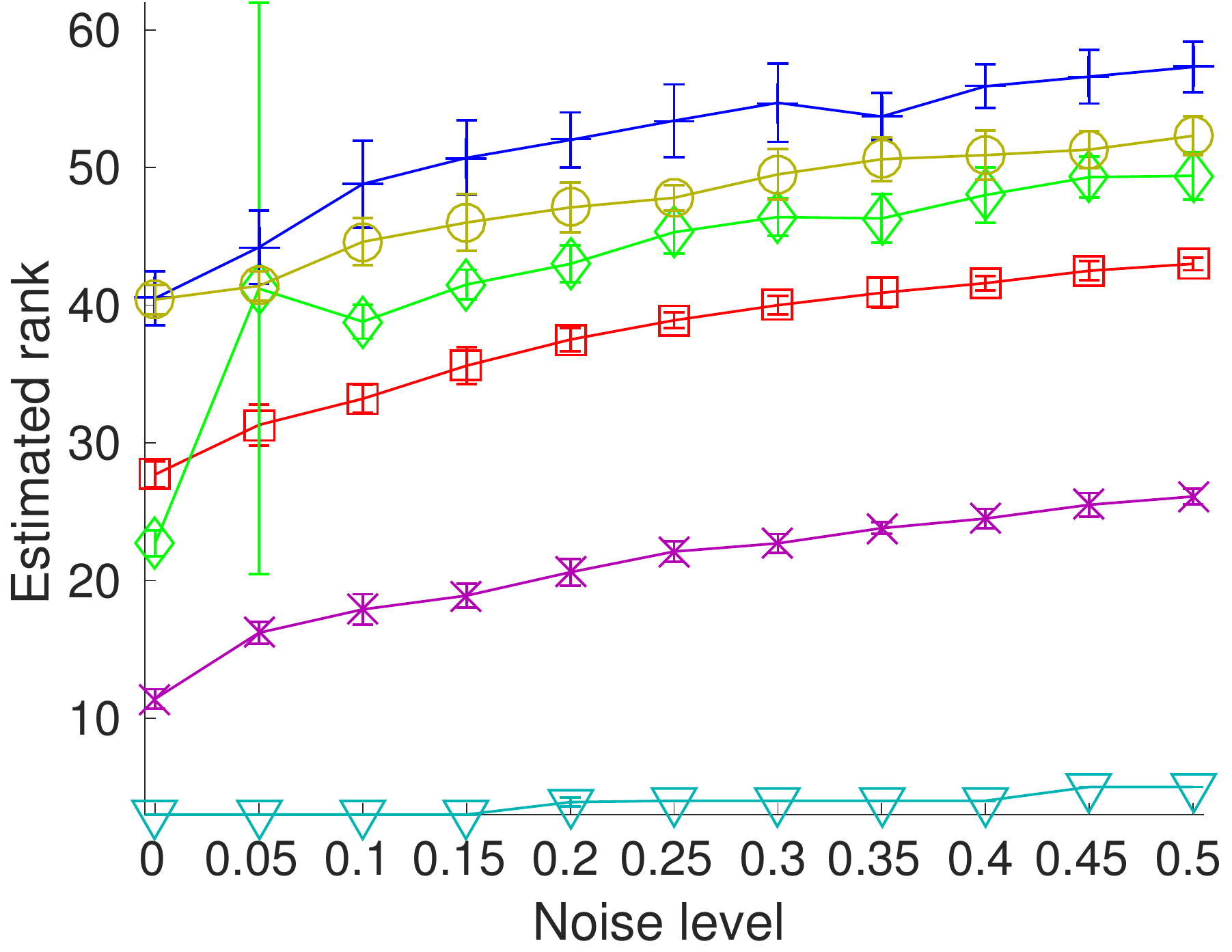}%
  }\\
  \rotatebox[origin=l]{90}{\hspace*{3em}\small Time}\hspace*{\smallfigsep}%
  \subfigure[Time, vary $m$]{%
    \label{fig:synth:small:n:norm:time}%
    \includegraphics[width=\smallfigwidth]{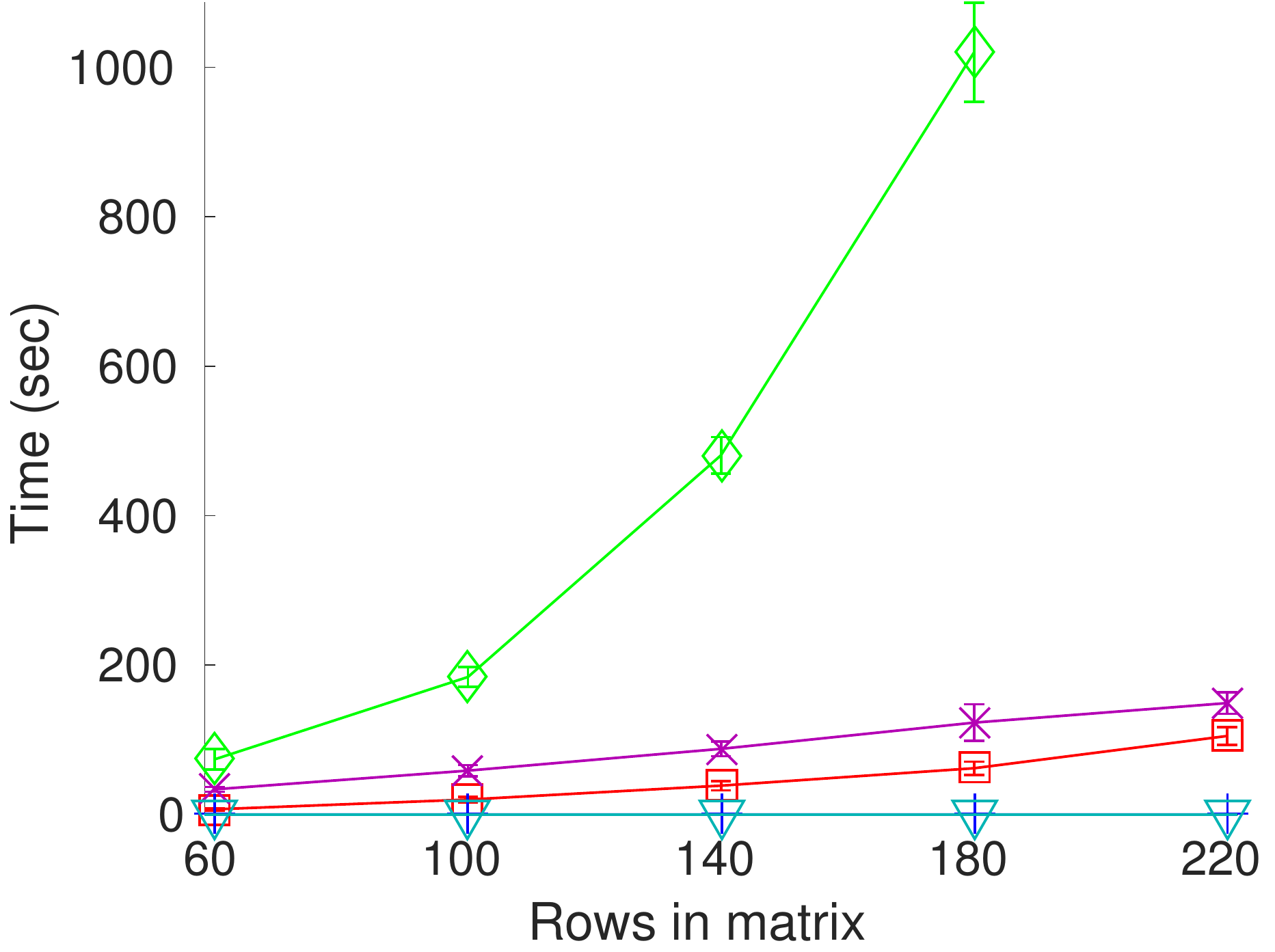}%
  }\hspace*{\smallfigsep}%
  \subfigure[Time, vary $k$]{%
    \label{fig:synth:small:k:norm:time}%
    \includegraphics[width=\smallfigwidth]{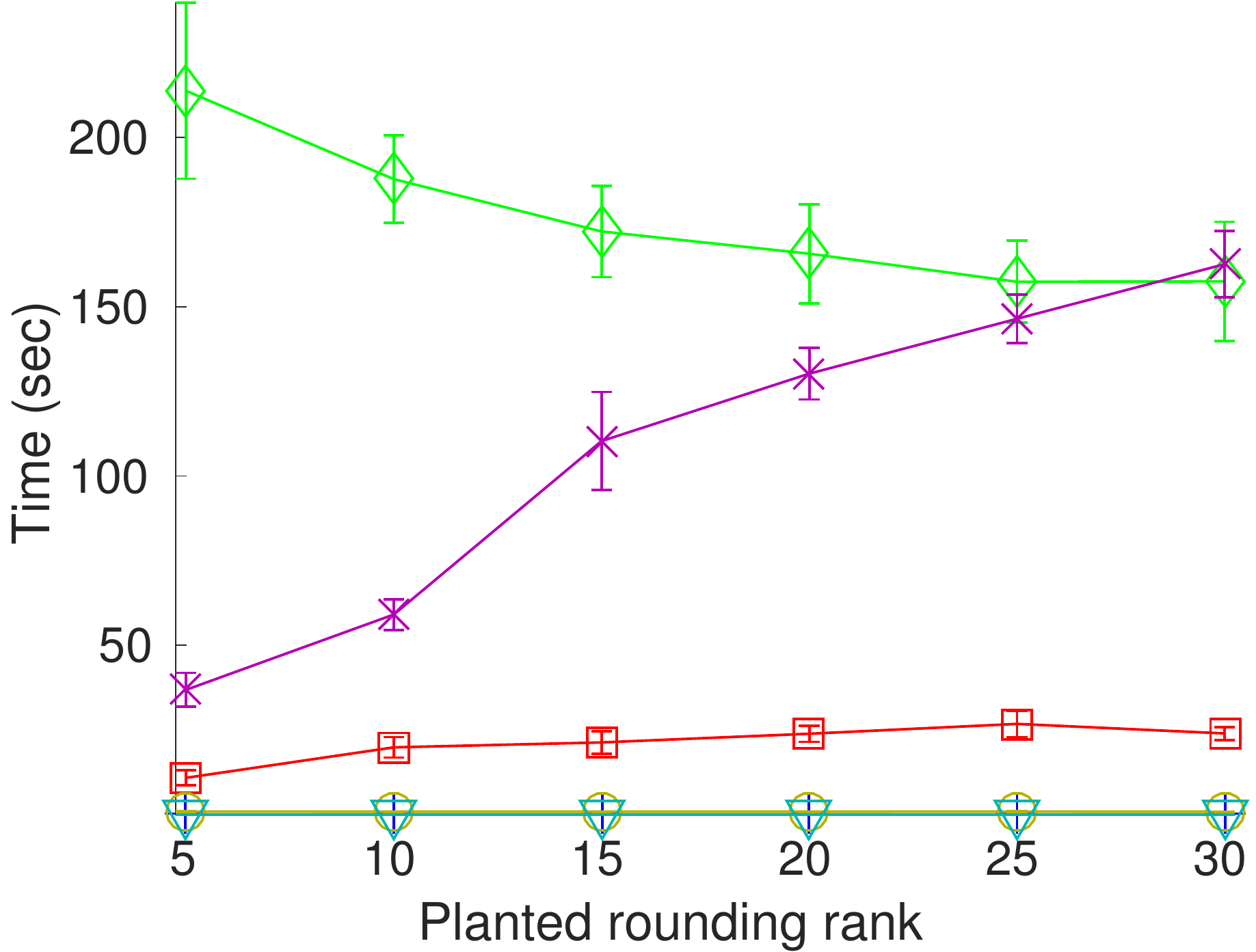}%
  }\hspace*{\smallfigsep}%
  \subfigure[Time, vary $d$]{%
    \label{fig:synth:small:dens:norm:time}%
    \includegraphics[width=\smallfigwidth]{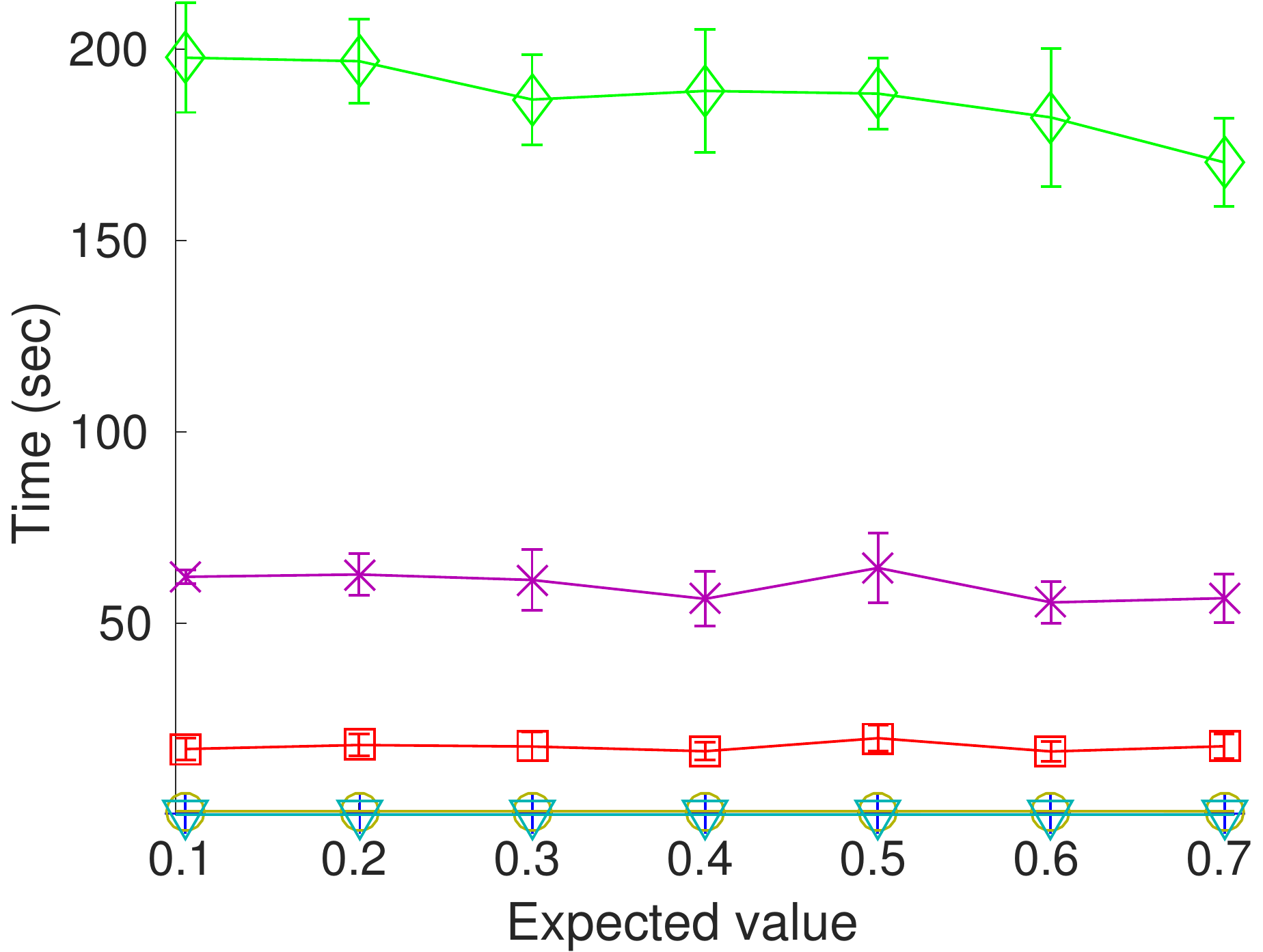}%
  }\hspace*{\smallfigsep}%
  \subfigure[Time, vary $p$]{%
    \label{fig:synth:small:noise:norm:time}%
    \includegraphics[width=\smallfigwidth]{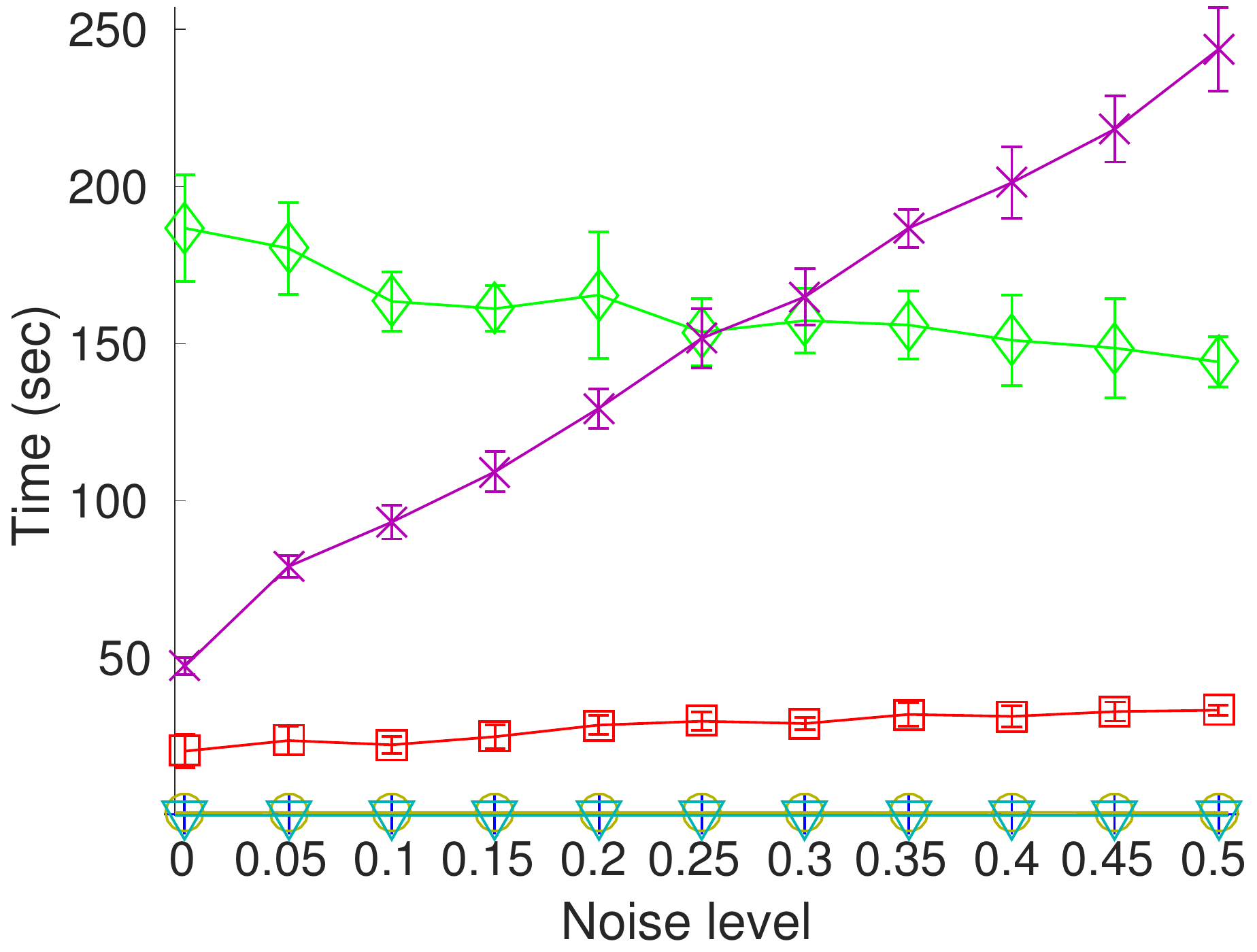}%
  }%
  \caption{Estimated rank and running times when using small synthetic data sets with normally distributed factor matrices (cf. Figure 2 of the main submission).}
  \label{fig:apx:synth:rank:times}
\end{figure*}

\begin{figure*}
  \centering
  \includegraphics[height=\legendheight]{err_big_legend} \\
  \rotatebox[origin=l]{90}{\hspace*{3em}\small Time}\hspace*{\smallfigsep}%
  \subfigure[Time, vary $m$]{%
    \label{fig:synth:big:n:unif:errtime}%
    \includegraphics[width=\smallfigwidth]{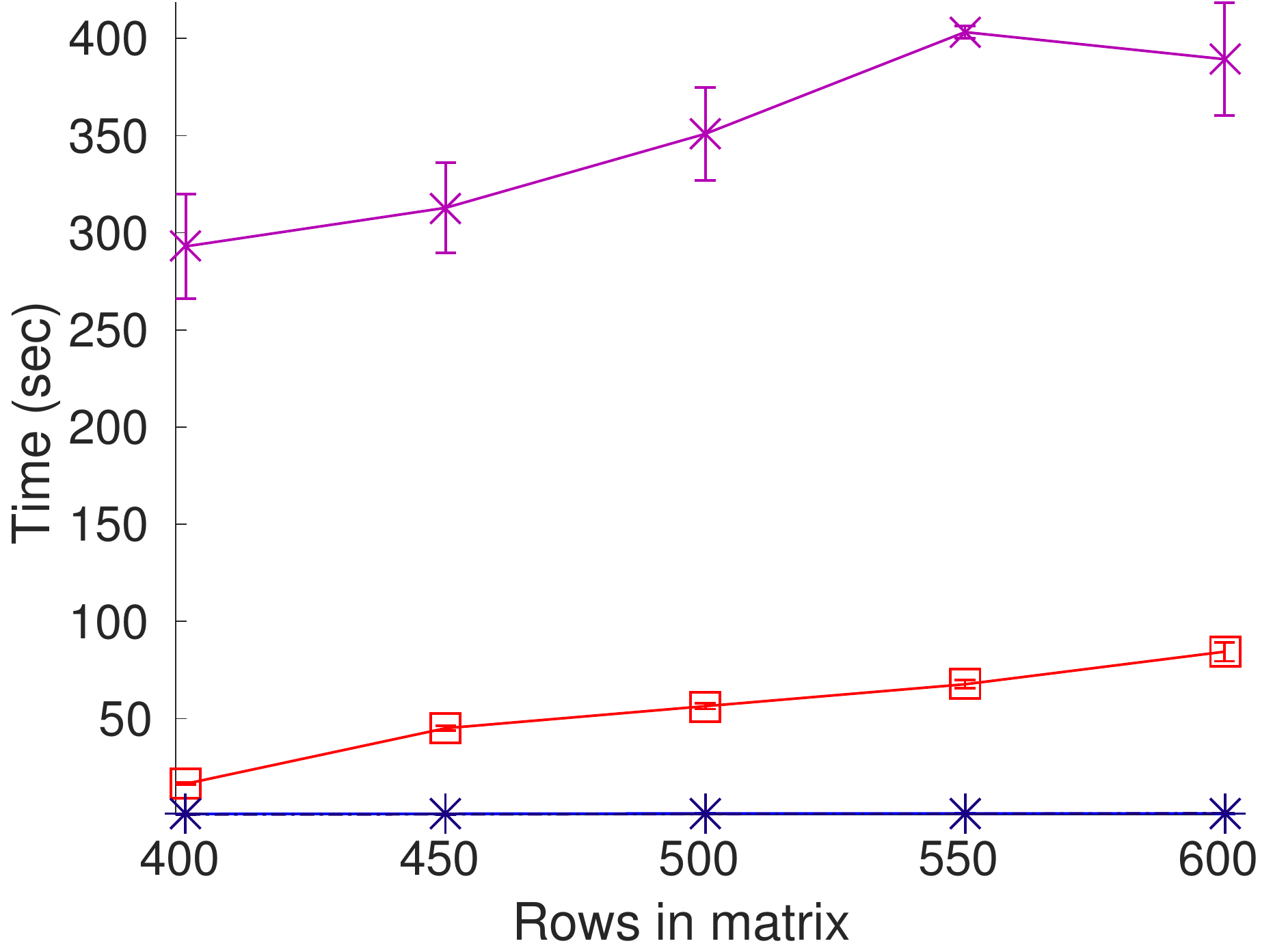}%
  }\hspace*{\smallfigsep}%
  \subfigure[Time, vary $k$]{%
    \label{fig:synth:big:k:unif:errtime}%
    \includegraphics[width=\smallfigwidth]{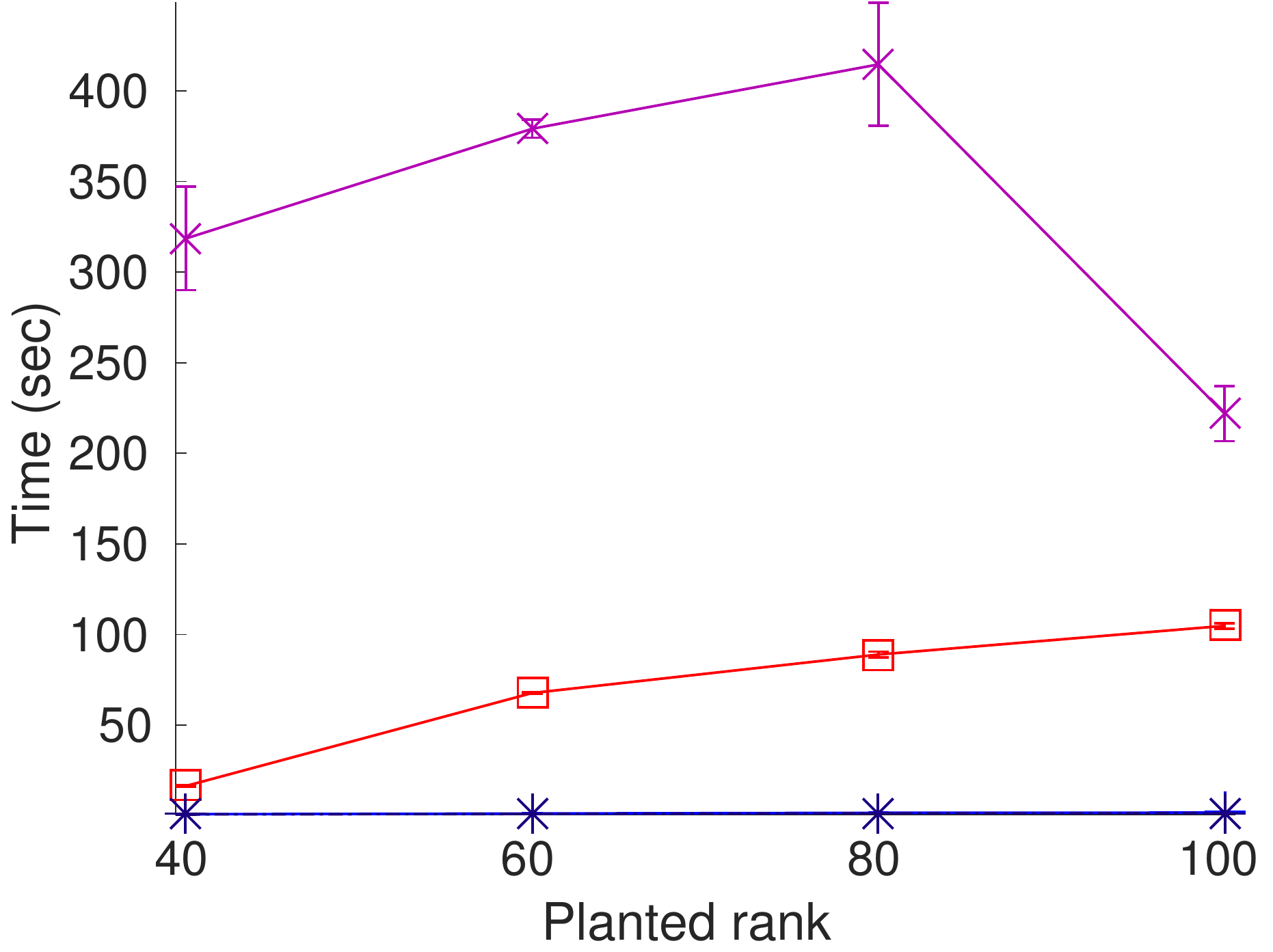}%
  }\hspace*{\smallfigsep}%
  \subfigure[Time, vary $\mu$]{%
    \label{fig:synth:big:dens:unif:errtime}%
    \includegraphics[width=\smallfigwidth]{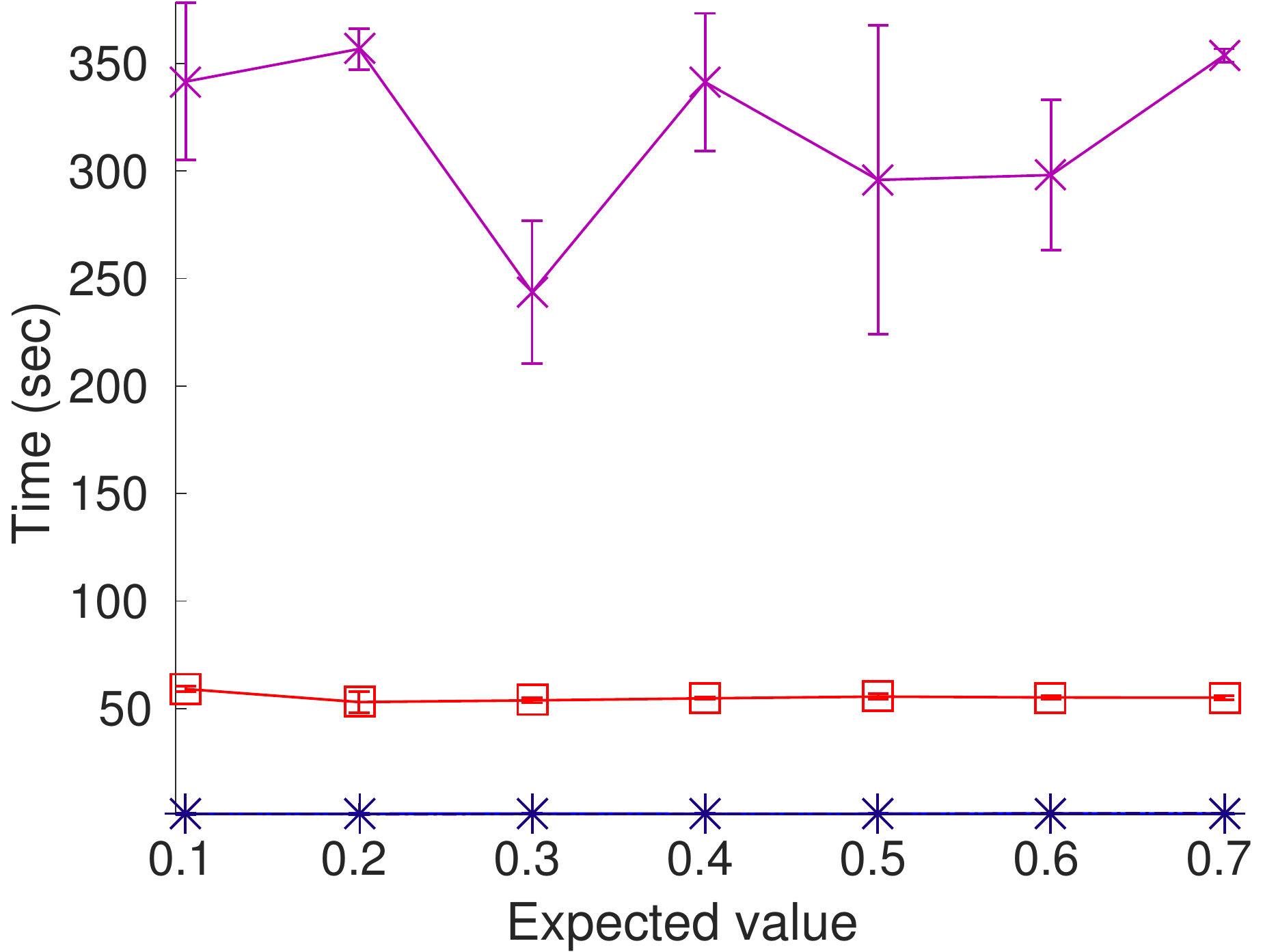}%
  }\hspace*{\smallfigsep}%
  \subfigure[Time, vary $p$]{%
    \label{fig:synth:big:noise:unif:errtime}%
    \includegraphics[width=\smallfigwidth]{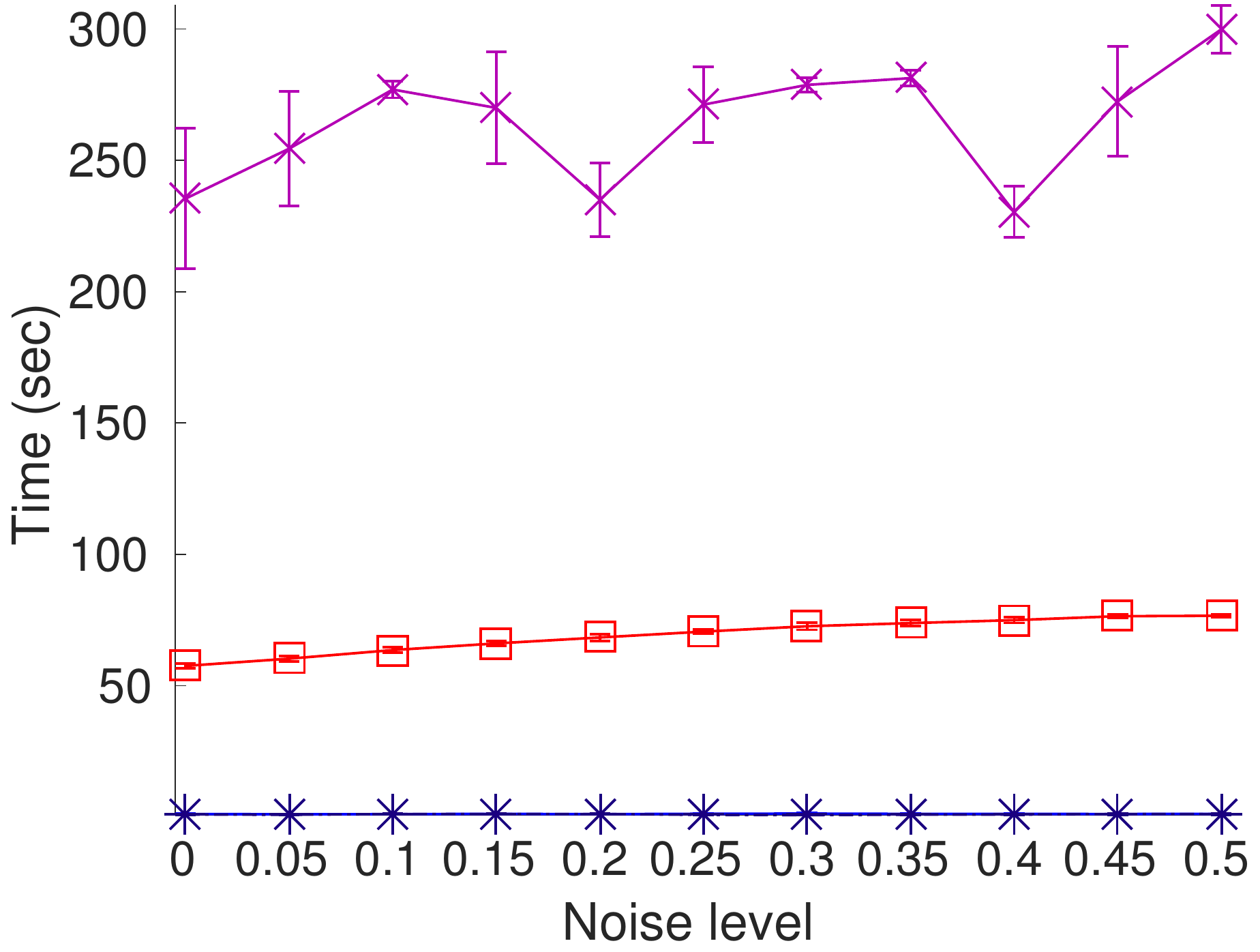}%
  }%
  \caption{Running times for the minimum-error fixed rounding rank decompositions on medium-sized synthetic data with uniformly distributed factors.}
  \label{fig:apx:synth:unif:time}
\end{figure*}

\begin{figure*}
  \centering
  \includegraphics[height=\legendheight]{err_big_legend} \\
  \rotatebox[origin=l]{90}{\hspace*{1em}\small Normal dist.}\hspace*{\smallfigsep}%
  \subfigure[Error, vary $m$]{%
    \label{fig:synth:big:n:norm:err}%
    \includegraphics[width=\smallfigwidth]{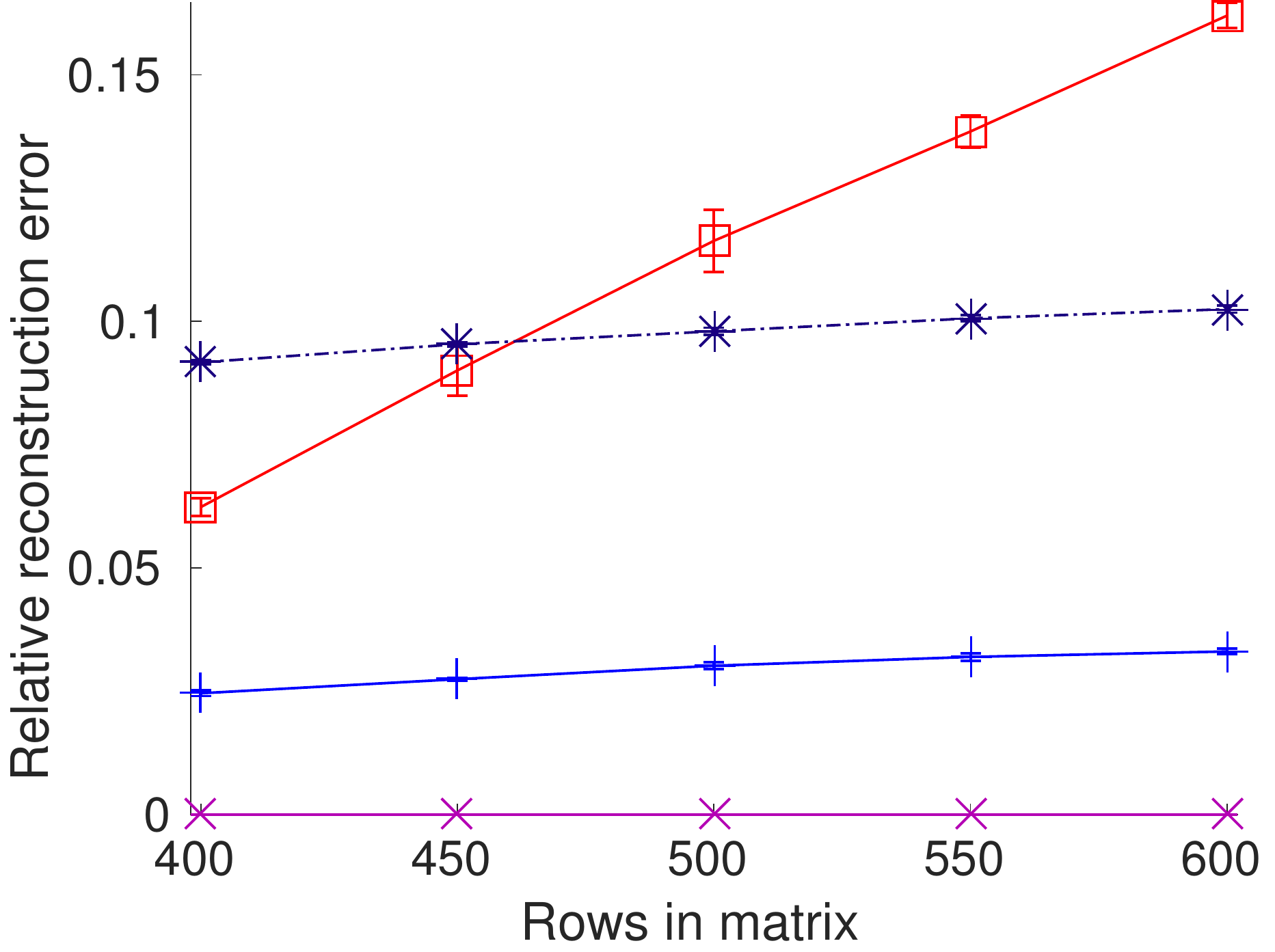}%
  }\hspace*{\smallfigsep}%
  \subfigure[Error, vary $k$]{%
    \label{fig:synth:big:k:norm:err}%
    \includegraphics[width=\smallfigwidth]{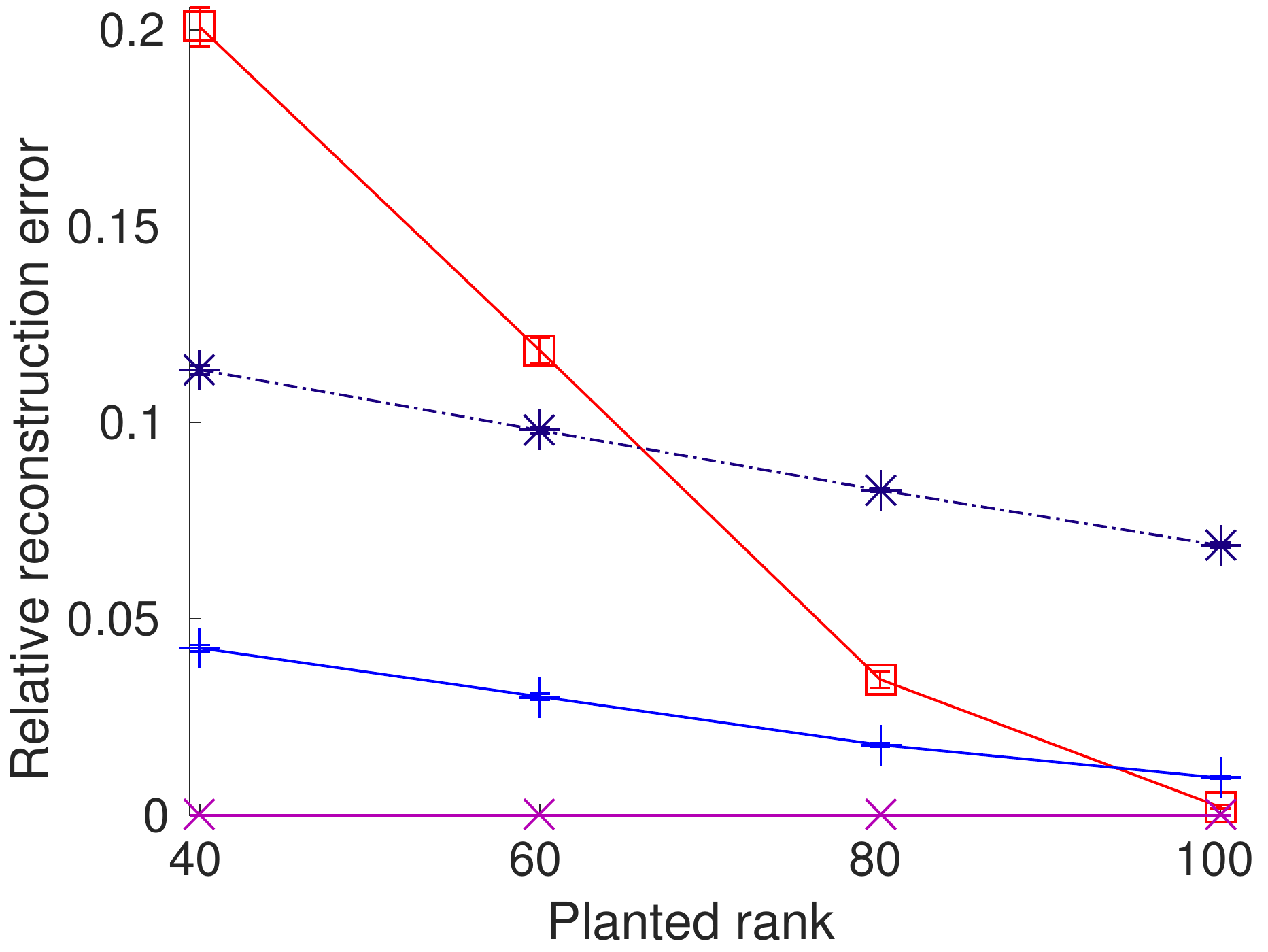}%
  }\hspace*{\smallfigsep}%
  \subfigure[Error, vary $d$]{%
    \label{fig:synth:big:dens:norm:err}%
    \includegraphics[width=\smallfigwidth]{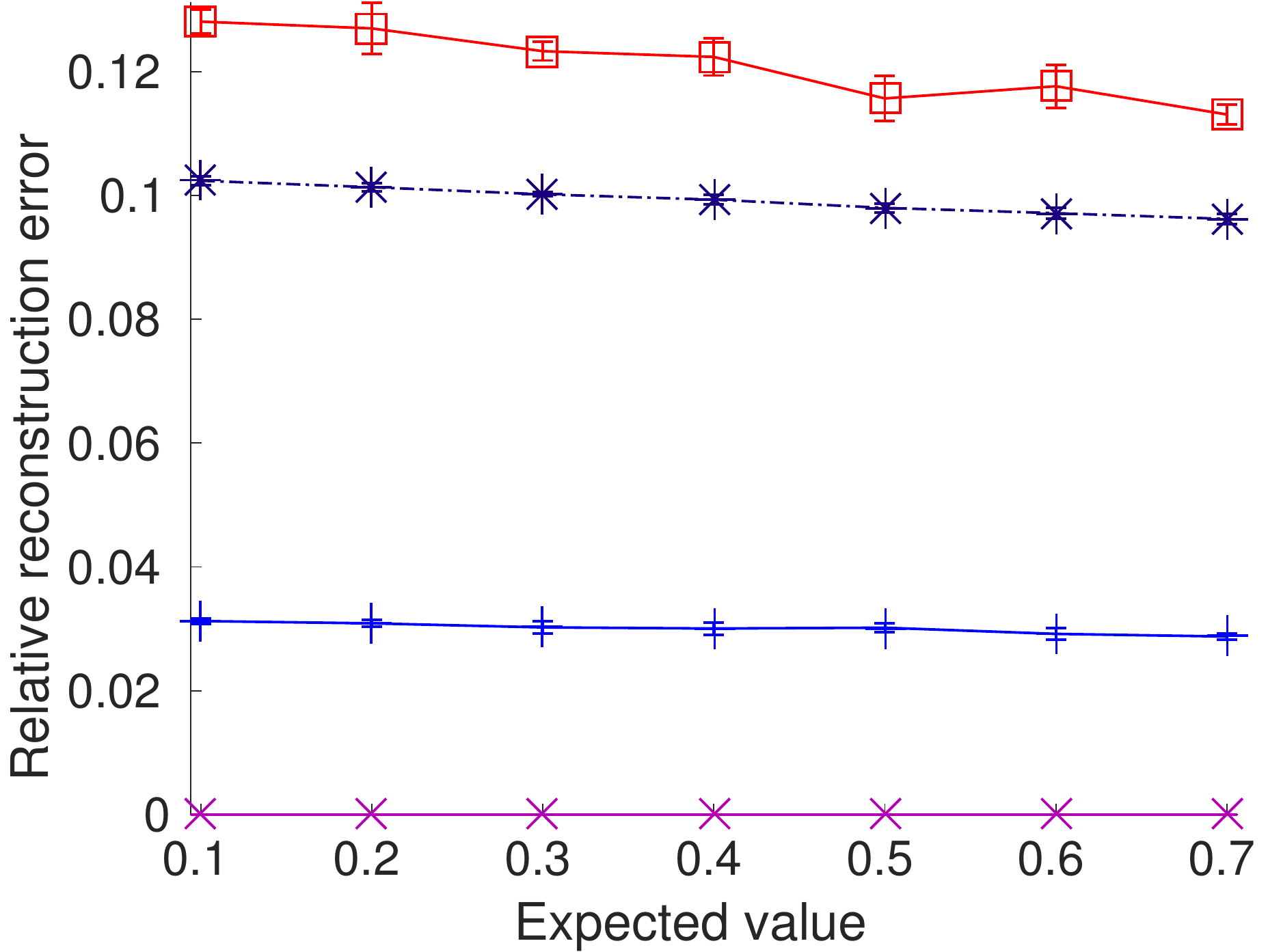}%
  }\hspace*{\smallfigsep}%
  \subfigure[Error, vary $p$]{%
    \label{fig:synth:big:noise:norm:err}%
    \includegraphics[width=\smallfigwidth]{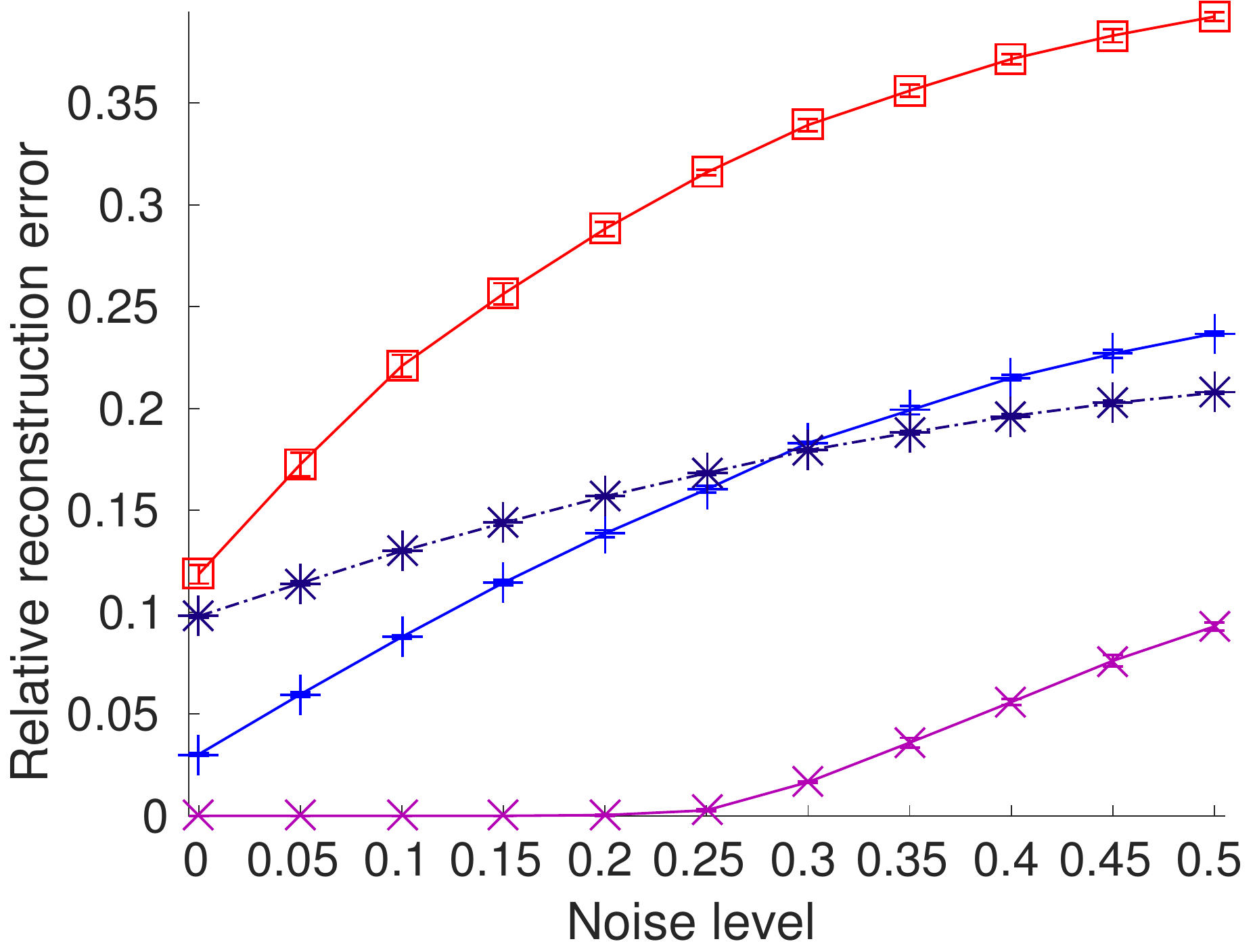}%
  }\\
  \rotatebox[origin=l]{90}{\hspace*{3em}\small Time}\hspace*{\smallfigsep}%
  \subfigure[Time, vary $m$]{%
    \label{fig:synth:big:n:norm:errtime}%
    \includegraphics[width=\smallfigwidth]{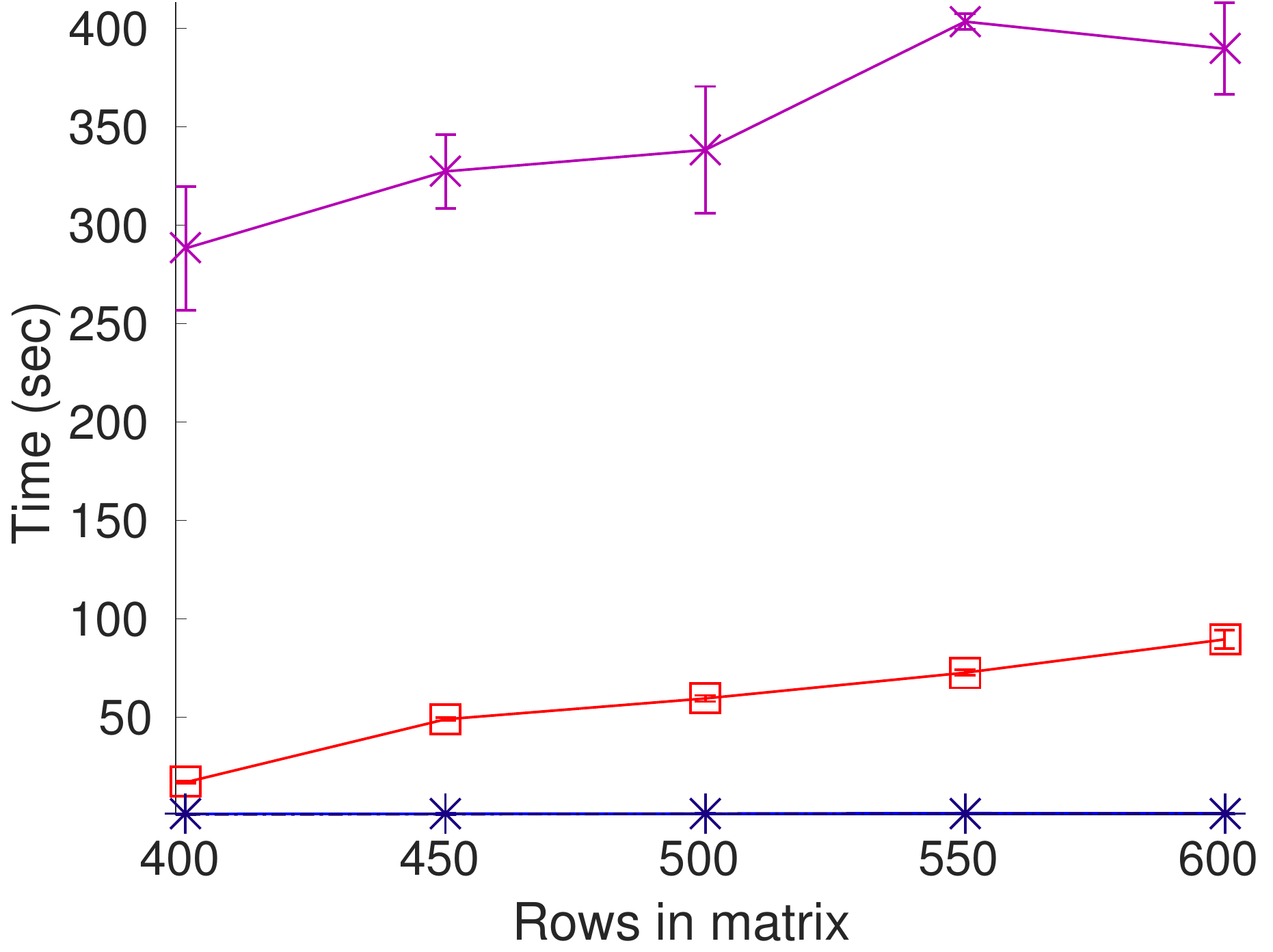}%
  }\hspace*{\smallfigsep}%
  \subfigure[Time, vary $k$]{%
    \label{fig:synth:big:k:norm:errtime}%
    \includegraphics[width=\smallfigwidth]{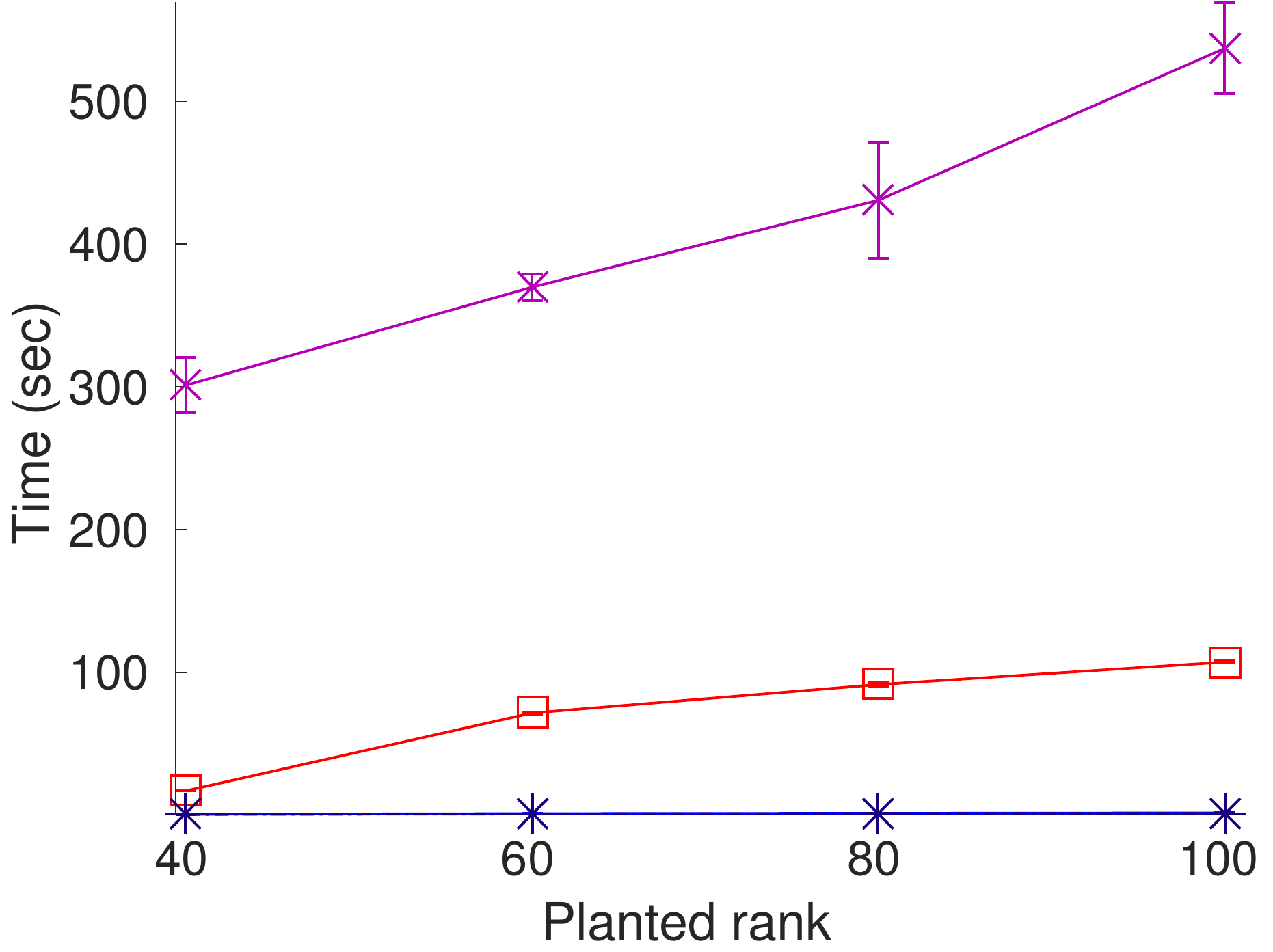}%
  }\hspace*{\smallfigsep}%
  \subfigure[Time, vary $d$]{%
    \label{fig:synth:big:dens:norm:errtime}%
    \includegraphics[width=\smallfigwidth]{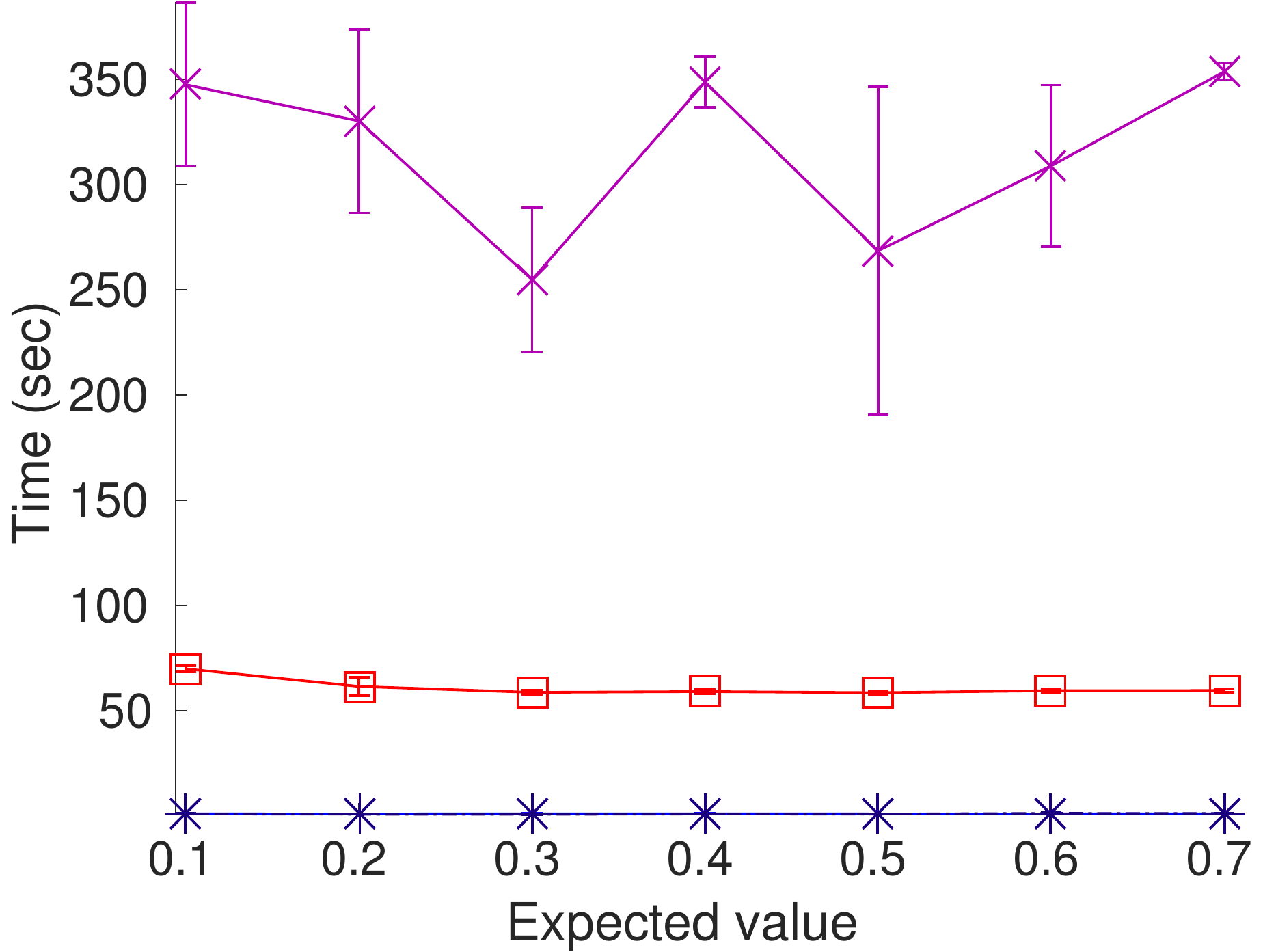}%
  }\hspace*{\smallfigsep}%
  \subfigure[Time, vary $p$]{%
    \label{fig:synth:big:noise:norm:errtime}%
    \includegraphics[width=\smallfigwidth]{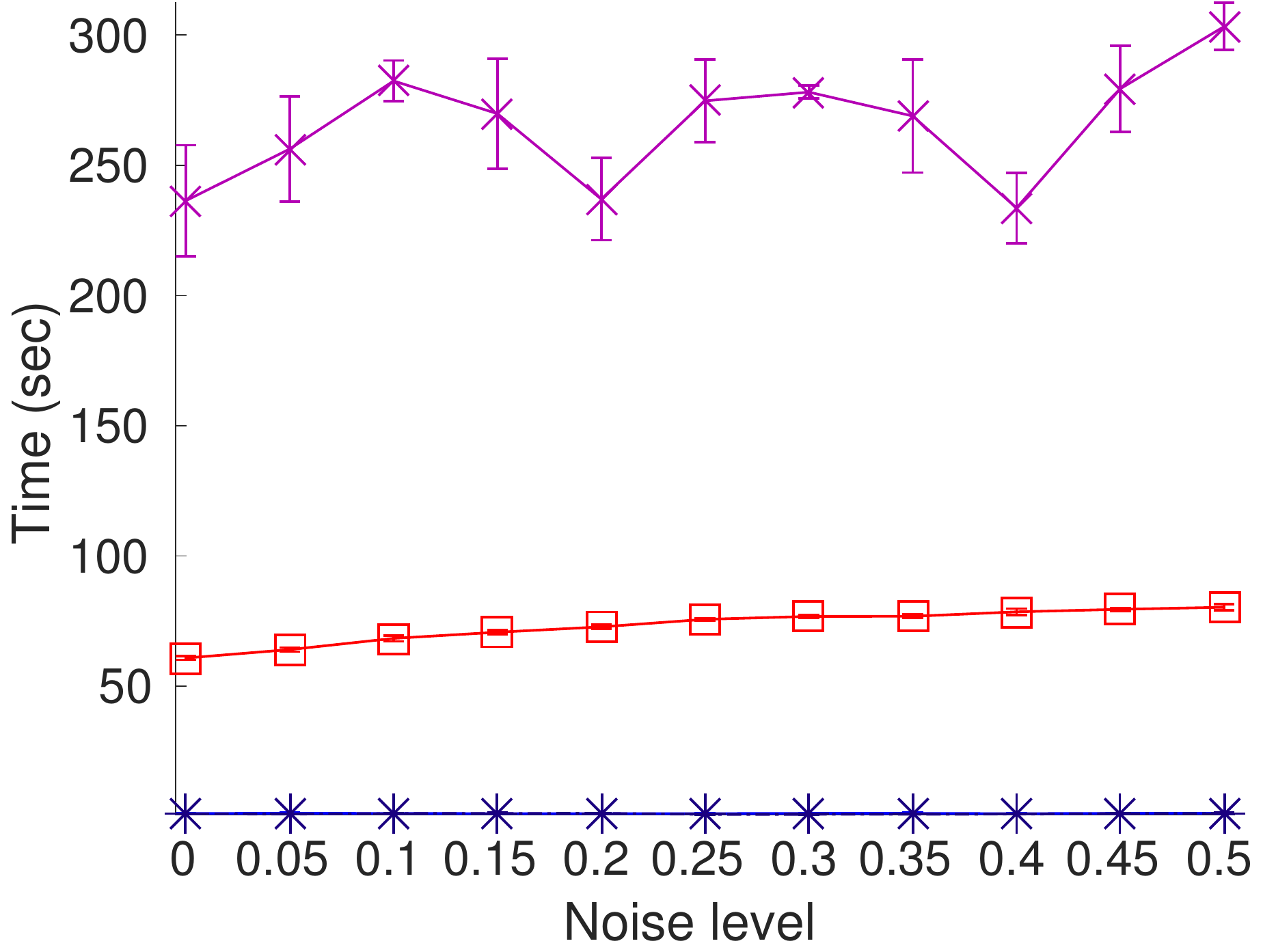}%
  }%
  \caption{Relative reconstruction errors and running times on medium-sized synthetic data with normally distributed factors. The top row gives the relative reconstruction error and the bottom row the running times. The results of \asso are omitted as they were significantly worse than the other results. All data points are averages over 10 random matrices and the width of the error bars is twice the standard deviation. Compare to Figure 3 of the main submission.}
  \label{fig:apx:synth:err}
\end{figure*}


\end{document}